\documentclass[pdflatex,sn-mathphys-num]{sn-jnl}

\usepackage{graphicx}%
\usepackage{multirow}%
\usepackage{amsmath,amssymb,amsfonts}%
\usepackage{amsthm}%
\usepackage{mathrsfs}%
\usepackage[title]{appendix}%
\usepackage{xcolor}%
\usepackage{textcomp}%
\usepackage{manyfoot}%
\usepackage{booktabs}%
\usepackage{algorithm}%
\usepackage{algorithmicx}%
\usepackage{algpseudocode}%
\usepackage{listings}%

\theoremstyle{thmstyleone}%
\newtheorem{theorem}{Theorem}
\newtheorem{proposition}[theorem]{Proposition}%

\theoremstyle{thmstyletwo}%
\newtheorem{remark}{Remark}%

\theoremstyle{thmstylethree}%

\theoremstyle{thmstyletwo}
\newtheorem{lemma}{Lemma}

\begin{document}

\title[]{Geometry-dependent Ekman layer approximations
	on curved domains: $L^\infty$ convergence}

\author[1]{\fnm{Yifei} \sur{Jia}}\email{jiayifei333@163.com}

\author*[2]{\fnm{Yi} \sur{Du}}\email{duyidy@jnu.edu.cn}

\author[1]{\fnm{Lihui} \sur{Guo}}\email{lihguo@126.com}

\affil*[1]{\orgdiv{College of Mathematics and System Science}, \orgname{Xinjiang University}, \orgaddress{\city{Urumqi}, \postcode{830017}, \country{P.R. China}}}

\affil[2]{\orgdiv{Department of Mathematics}, \orgname{Jinan University}, \orgaddress{\city{Guangzhou}, \postcode{510632}, \country{P.R. China}}}


\abstract{
The Ekman boundary layer is a fundamental concept in fluid dynamics that describes fluid motion near boundaries affected by Earth's rotation.

 Most theoretical studies have simplified their analysis by assuming a planar boundary surface, resulting in limited exploration of structures with general smooth boundary conditions. Investigating the impact of boundary geometry in the Ekman boundary layer is essential, as initially suggested by J.L. Lions and further examined in Masmoudi's study [Comm. Pure Appl. Math. 53 (2000), 432--483] under small amplitude periodic boundary conditions.

This paper clarifies how boundary geometry influences flow fields and characterizes its effects on near-boundary layer flow. We construct a class of multi-scale approximate solutions based on the boundary's geometric features and establish their convergence in the \(L^\infty\) framework. Our findings do not require a small-amplitude assumption, only an upper bound on the Gaussian curvature of the boundary surface. Notably, when the boundary is planar, our approach aligns with existing studies. Additionally, in the vanishing-viscosity limit, we derive a limiting-state system dependent on boundary geometric parameters. These contributions extend the theoretical understanding of boundary-layer interactions to general curved geometries and have possible applications in atmospheric, oceanic, and other geophysical flow contexts. 
}




\maketitle
\section{Introduction}\label{sec1}	
This paper is concerned with the  geophysical fluid dynamics under rotation,  described by the following system:
\begin{equation}\label{1.1}
	\begin{cases}
		\partial_t \boldsymbol{u}^\varepsilon-\mu_{h}\Delta_{h} \boldsymbol{u}^\varepsilon-\mu_{v}\partial_z^2 \boldsymbol{u}^\varepsilon+(\boldsymbol{u}^\varepsilon \cdot \nabla)\boldsymbol{u}^\varepsilon+\varepsilon^{-1}\boldsymbol{R}\boldsymbol{u}^\varepsilon+\varepsilon^{-1}\nabla p^\varepsilon=0,\\
		\nabla\cdot \boldsymbol{u}^\varepsilon=0,
	\end{cases}
\end{equation}
where $(t,\boldsymbol x)\in \mathbb{R^+}\times \Omega$, with $\Omega$ being a domain of $\mathbb{R}^3$ with smooth boudary,
and $\Delta_h$ is defined as the horizontal Laplacian operator.
Here, $\boldsymbol{u}^\varepsilon, p^\varepsilon$, $\mu_h=\mu_h(\varepsilon)>0$, and $\mu_v=\mu_v(\varepsilon)>0$ denote the velocity, pressure of the fluid, horizontal and vertical viscosity coefficients, respectively.
The Coriolis force $\boldsymbol{R}\boldsymbol{u}^\varepsilon$, essential in rotating fluid systems, is incorporated as
\begin{equation*}
	\boldsymbol{R}\boldsymbol{u}^\varepsilon=\begin{pmatrix}
		0&-1&0\\
		1&0&0\\
		0&0&0
	\end{pmatrix}\boldsymbol{u}^\varepsilon\,.
\end{equation*}

For the system \eqref{1.1}, the initial and boundary conditions read as
\begin{align}
	&\boldsymbol{u}^\varepsilon(t,x,y,z)|_{t=0}=\boldsymbol{u}^\varepsilon_0(x,y,z)\, ,\label{1.2}
	\\
	&\boldsymbol{u}^\varepsilon(t,x,y,z)|_{\partial \Omega}=0\, .\label{1.3}
\end{align}
Due to the axisymmetric structure of $\boldsymbol{R}\boldsymbol{u}^\varepsilon$ and following a standard procedure, there exists a weak solution pair $(\boldsymbol{u}^\varepsilon,{p}^\varepsilon)$ for the system \eqref{1.1}-\eqref{1.3} when $\boldsymbol{u}^\varepsilon_0 \in L^2(\Omega)$ (see Ref.~\cite{Chemin2006}).

The Ekman boundary layer (see Refs.~\cite{Greenspan1968,Pedlosky}) is a thin layer of fluid that forms adjacent to a boundary in a rotating fluid and is characterized by friction dominating the region near the boundary. 
The Ekman boundary layer's dynamics are influenced by various factors, including the fluid's viscosity, topography, wind stress, curvature of the earth, and turbulent frictional stress. 

Most investigations in Ekman boundary layer theory rely on the idealized assumption of a flat boundary (see Refs.~\cite{Chemin2006,Grenier1997,Li2019,Masmoudi2000,giga2007,giga2013}). However, natural boundaries often exhibit complex characteristics, featuring diverse seabed topographies and irregular landforms. In the general domain $\Omega$, many unresolved mathematical challenges remain (Chemin et al. \cite{Chemin2006}). Analyzing boundary layers in non-flat cases is significantly more complex than in flat cases due to varying normal directions caused by boundary curvature, which complicates the determination of boundary layer thickness \cite{Gerard2005}. The Ekman boundary layer's complexity is heightened under non-flat geometries. Even for geometrically favorable shapes like spherical shells, Stewartson \cite{Stewartson} found that boundary layer thicknesses ($\varepsilon^{1/3}, \varepsilon^{1/4}, \varepsilon^{2/7},...$) vary with latitude. Additionally, Rousset \cite{Rousset2007} provided approximate boundary layer solutions for well-posed initial values in regions of the unit sphere away from the equator.

Non-flat boundary geometric features can significantly impact boundary layer solutions, a concept first proposed by J. L. Lions and further explored by Masmoudi. In his seminal work \cite{Masmoudi2000}, Masmoudi analyzed the Ekman boundary layer solution in the domain $\mathbb{T}^2 \times [\varepsilon B(x,y), 1]$, where $\varepsilon$ is a small positive parameter and $B(x,y)$ is a periodic smooth function. He established convergence results within the $L^2$ framework, dependent on small boundary amplitude. 
Subsequent studies have continued to investigate boundary layers, as noted in Refs.~\cite{Bresch2004,Constantin2022,Gerard2017,Gerard2003,Rax2019}. 
For the Prandtl boundary layer, Liu and Wang \cite{Liu2014}, utilized curvilinear coordinate systems and multi-scale analysis to study the boundary layer theory of the incompressible Navier-Stokes equations within a two-dimensional curved bounded domain.

Recently, in Ref.~\cite{Jia2024}, the authors  extended the analysis to more general non-flat boundary geometries by removing the restriction of small boundary amplitude. Specifically, for the domain $\Omega = \mathbb{R}^2 \times [B(x, y), B(x, y) + 2]$ with $B(x,y) \in C^\infty(\mathbb{R}^2)$, they justified $L^2$-convergence under certain geometry curvature constraints on the boundary surface.

To facilitate compare the results of flat and non-flat boundary cases, we summarize two related  key findings for the flat boundary case below:

\begin{theorem}[Grenier and Masmoudi \cite{Grenier1997}]\label{thmGrenier}
	Let $\Omega=\mathbb{T}^2\times[0,1]$, 
	$(\boldsymbol{u}^\varepsilon,{p}^\varepsilon)$ be a pair of global weak solutions of system \eqref{1.1}-\eqref{1.3} with initial data $\boldsymbol{u}_0^\varepsilon$.
	Assume $\boldsymbol{u}_0^\varepsilon$ strongly converges to ${\boldsymbol{u}}_0 (x,y)= ({\boldsymbol{u}}_{0,h}, 0)$ in $L^2(\Omega)$,  with $\int_{\mathbb{T}^2} {\boldsymbol{u}}_{0,h} = 0$. 
	Then  following cases hold:
	\begin{enumerate}[1.]
		\item{Anisotropic case with $\mu_h\sim \mathcal{O}(1)$ , $ \mu_v\sim \mathcal{O}(\varepsilon)$ :}
		
		Assume ${\boldsymbol{u}}_0\in H^1(\mathbb{T}^2)$, then
		\begin{align*}
			\boldsymbol{u}^\varepsilon-{\boldsymbol{u}}\rightarrow 0^+\quad\text{in}\quad L^\infty(\mathbb{R_+};L^2(\Omega))
			\,, \\
			\nabla_{ h}(\boldsymbol{u}^\varepsilon-{\boldsymbol{u}})\rightarrow 0^+\quad\text{in}\quad L^2(\mathbb{R_+};L^2(\Omega))\,,
		\end{align*}
		where  $\nabla_h=(\partial_x,\partial_y)^T$, and ${\boldsymbol{u}}=(\boldsymbol{u}_h,0)$ is the solution of the following two-dimenszonnal system with damping
		\begin{equation}\label{1.6}
			\begin{cases}
				\partial_{t}  {\boldsymbol u}_h+( {\boldsymbol u}_h \cdot \nabla_h)  {\boldsymbol u}_h
				-\mu_{h}\Delta_{h} {\boldsymbol{u}}_h
				+\sqrt{2}{\boldsymbol u}_h+\nabla_h p=0\,,\\
				\nabla_h\cdot {\boldsymbol u}_h=0\,,\\
				{\boldsymbol u}_h|_{t=0}= {\boldsymbol u}_{0,h}\,.
			\end{cases}
		\end{equation}
		\item{Isotropic case with $\mu_h \sim \mu_v\sim \mathcal{O}(\varepsilon)$:}
		
		Assume that for $C_0$ is a universal constant such that
		\begin{equation}\label{1.7}
			\lVert{\boldsymbol{u}}_h\rVert_{{L}^{\infty}(\mathbb{T}^2)}\leq
			C_0\,.
		\end{equation}	
		Then 
		\begin{equation}\label{1.8}
			\boldsymbol{u}^\varepsilon-{\boldsymbol{u}}\rightarrow 0^+\quad\text{in}\quad L^\infty(\mathbb{R_+};L^2(\Omega))\,,
		\end{equation}
		and ${\boldsymbol{u}}=(\boldsymbol{u}_h,0)$ satisfies the following two-dimensional damped Euler type system
		\begin{equation}\label{1.9}
			\begin{cases}
				\partial_{t}  {\boldsymbol u}_h+( {\boldsymbol u}_h \cdot \nabla_h)  {\boldsymbol u}_h
				+\sqrt{2}{\boldsymbol u}_h+\nabla_h p=0\,,\\
				\nabla_h\cdot {\boldsymbol u}_h=0\,,\\
				{\boldsymbol u}_h|_{t=0}= {\boldsymbol u}_{0,h}\,.
			\end{cases}
		\end{equation}
	\end{enumerate}
\end{theorem}
\begin{remark}
	Viscosity is essential for verifying the convergence of approximate solutions (\cite{Chemin2006}). In the anisotropic case of Theorem \ref{thmGrenier}, viscous dissipation in the horizontal direction effectively reduces the loss of regularity caused by nonlinear terms, thus ensuring convergence. Conversely, in the isotropic case, the dissipation effect diminishes and fails to provide sufficient regularity to manage the nonlinear term. To ensure convergence, an additional condition \eqref{1.7} must be imposed on the limiting system \eqref{1.9} to address the lack of regularity stemming from vanishing viscosity. 
\end{remark}

The above convergence estimates are established within the $L^2$ framework.
In contrast, for the more physically relevant $L^\infty$ space, Gong, Guo, and Wang \cite{Gong2015} 
addressed the potential singularities in the derivatives of the residuals by constructing higher-order approximate solutions though boundary layer expanding.
Consequently, they obtained the $L^\infty$ convergence estimate for boundary layer dynamics over a flat boundary, enhancing the result of Theorem \ref{thmGrenier} in the anisotropic viscosity scenario.  
The result of Ref.~\cite{Gong2015}  reads as:
\begin{theorem}[Gong, Guo, and Wang \cite{Gong2015}]\label{thgong}
	Let $\Omega=\mathbb{T}^2\times[0,1]$, and let $\boldsymbol{u}^\varepsilon$ denote a family of local strong solutions in $[0, T)$ for the system \eqref{1.1}-\eqref{1.3} (where $\mu_h \sim \mathcal{O}(1)$ and $\mu_v \sim \mathcal{O}(\varepsilon)$) with initial data $\boldsymbol{u}^\varepsilon_0 \in H^6(\Omega)$.
	
	Assume
	\begin{equation}\label{1.101}
		\boldsymbol{u}^\varepsilon_{0}-(\boldsymbol{u}_{0,h},0)\rightarrow 0^+\quad\text{in}\quad  H^6(\Omega)\,,
	\end{equation} 
	then there holds
	\begin{equation}\label{1.10} \lim\limits_{\varepsilon\rightarrow0^+}
		\Big\lVert
		\boldsymbol{u}^\varepsilon
		-
		\boldsymbol{u}_{app}^\varepsilon
		\Big\rVert_{{L}^{\infty}([0,T)\times\Omega)}=0\,,
	\end{equation}
	where $\boldsymbol{u}_{app}^\varepsilon$ is the approximate solution of the form
	\begin{equation}\label{1.30}
		\boldsymbol{u}_{app}^\varepsilon= \sum^{2}_{i=0}\varepsilon^i \boldsymbol{u}^i\big(t,x,y,z,\tfrac{z}{\varepsilon},\tfrac{1-z}{\varepsilon}\big)\,.
	\end{equation}  
	And $\boldsymbol{u}^i(t,x,y,z,\frac{z}{\varepsilon},\frac{1-z}{\varepsilon})$ consists of internal and boundary terms, are constructed via a multiscale method based on the power of $\varepsilon$.
	Meanwhile, the principal part of the approximate solution \eqref{1.30} is the zero-order  interior term $(\boldsymbol{u}_h,0)$, which satisfies the limiting system \eqref{1.6} with initial value $\boldsymbol{u}_{0,h}$.	
\end{theorem}
 \begin{remark}
Although Theorem \ref{thgong} addresses the anisotropic case with horizontal dissipation, diffusion alone cannot control the singularity from the remainder terms in the $L^\infty$ framework. Consequently, a higher-order asymptotic expansion of the approximate solution is performed, as detailed in Ref.~\cite{Gong2015}. This requires additional regularity assumptions \eqref{1.101} on the limiting solution to validate the expansion.
 \end{remark}

This paper aims to examine the approximate solution of the Ekman boundary for non-flat domains and justify its convergence in a stronger norm \(L^\infty\).
We shall consider the case of isotropic viscosity ($\mu_h=\mu_v=\nu\varepsilon$) system as below:
\begin{equation}\label{1.11}
	\begin{cases}
		\partial_t \boldsymbol{u}^\varepsilon-\nu\varepsilon \Delta \boldsymbol{u}^\varepsilon+(\boldsymbol{u}^\varepsilon \cdot \nabla)\boldsymbol{u}^\varepsilon+\varepsilon^{-1}\boldsymbol{R}\boldsymbol{u}^\varepsilon+\varepsilon^{-1}\nabla p^\varepsilon=0\,,\\
		\nabla\cdot \boldsymbol{u}^\varepsilon=0\,,\\
		\boldsymbol{u}^\varepsilon|_{\partial\Omega}=0\,,\quad \boldsymbol{u}^\varepsilon|_{t=0}=\boldsymbol{u}_0^\varepsilon\,,
	\end{cases}  (t,x)\in \mathbb{R}^+\times \Omega\,,
\end{equation}
where 
\begin{equation}\label{1.12}
	\Omega=\mathbb{R}^2 \times [B(x, y), B(x, y) + 2]\,,\qquad  B(x,y)\in C^\infty(\mathbb{R}^2)\,.
\end{equation}
Noting that $\varepsilon^{-1}\boldsymbol{R}\boldsymbol{u}^\varepsilon$ is skew-symmetric for any fixed $\varepsilon > 0$, it is straightforward to prove that there exists  $T^*>0$ such that a local strong solution to the system \eqref{1.11}-\eqref{1.12} can be established in the interval $[0, T^*]$. This solution specifically satisfies the energy estimate:
\begin{equation}\label{1.13}
	\lVert\boldsymbol{u}^\varepsilon\rVert_{{L}^\infty([0,T^*); {H}^s(\Omega))}^2+\nu\varepsilon\int_0^{T^*}\lVert\nabla\boldsymbol{u}^\varepsilon(t,\cdot)\rVert_{H^s(\Omega)}^2\,dt \lesssim \lVert\boldsymbol{u}^\varepsilon_{0}\rVert_{H^s(\Omega)}^2\,.
\end{equation}

We will first introduce some notations before presenting our main results.

\subsection{Notations}\label{subsec1.1}
In this paper, we use the notations $\nabla=(\partial_x,\partial_y,\partial_z)^T$ and $\nabla^\bot_h=(-\partial_y,\partial_x)^T$. 
We denote the normal vector of the surface $z=B(x,y)$ as
\begin{equation*}
	{\boldsymbol n}_{B}=\Big(\tfrac{-B_x}{\sqrt{1+B_x^2+B_y^2}}, \tfrac{-B_y}{\sqrt{1+B_x^2+B_y^2}},\tfrac{1}{\sqrt{1+B_x^2+B_y^2}}\Big) =(\cos \alpha,\cos \beta,\cos\gamma)\,,
\end{equation*}
and the Hessian matrix of the surface $B(x,y)$  as
\begin{equation*}
	\boldsymbol{H}=\begin{pmatrix}
		B_{xx}&  B_{xy}\\
		B_{xy} & B_{yy}
	\end{pmatrix}.
\end{equation*}
Then it follows that
\begin{equation*}
	\nabla_{ h}\gamma(x,y)={\cos^3\gamma}{\sin^{-1}\gamma}\boldsymbol{H}\nabla_{ h}B\,.
\end{equation*}
The identity matrix is represented as $\boldsymbol{E}$, with other matrices represented by
\begin{align*}
	\boldsymbol{H_0}=
	\begin{pmatrix}
		1+B_x^2& B_xB_y\\
		B_xB_y& 1+B_y^2
	\end{pmatrix},
	\quad
	\boldsymbol{E_1}=
	\begin{pmatrix}
		0&1\\
		-1& 0
	\end{pmatrix}.
\end{align*}
The matrix $ \boldsymbol{H_0} $ is  positive definite and
$
	\det \boldsymbol{H_0}=\cos^{-2}\gamma\,.
$
The Gaussian curvature of the surface $ B(x,y) $ is given by
\begin{equation*}
	\boldsymbol{K}_{\mathrm{G}} = (\det \boldsymbol{H_0})^{-2} \det \boldsymbol{H}\,.
\end{equation*}

Throughout the paper, we sometimes use the notation $A\lesssim B$ as an equivalent to $A \leq CB$ with a uniform constant $C$.

\subsection{Main results}\label{subsec1.2}

For flat boundaries, depth of boundary layer  remains constant, whereas it varies with the function \( \gamma(x,y) \) in non-flat geometries. As demonstrated in Ref.~\cite{Jia2024}, the depth can be expressed as:

\begin{equation}\label{1.14}
	{\delta}(x,y)=\sqrt{\nu}\varepsilon\cos^{-\frac{3}{2}}\gamma(x,y)\,.
\end{equation}

We define the bottom and top boundary layer coordinates as follows:

\begin{equation}\label{1.15}
	(\tilde x,\tilde y,\tilde z)=\Big(x,y,\tfrac{z-B}{{\delta}}\Big)\,,
	\qquad
	(\bar{x},\bar{y},\bar{z})=\Big(x,y,\tfrac{2+B-z}{{\delta}}\Big)\,.
\end{equation}

The following theorem about the approximate solution holds:
\begin{theorem}\label{thapp}

As $\varepsilon\rightarrow 0^+$, assume $\boldsymbol{u}_0^\varepsilon$ strongly converges to ${\boldsymbol{u}}_0 (x,y)= ({\boldsymbol{u}}_{0,h}, 0)$ in $L^2(\Omega)$ and ${\boldsymbol{u}}_{0,h}\in H^6(\mathbb{R}^2)$, the boundary surface $B(x,y)\in C^\infty(\mathbb{R}^2)$ satisfies 	
	\begin{equation}\label{+16}
		\frac{\cos^2\alpha+\cos^2\beta}{\cos^2\gamma}<\tfrac{1}{8}\,,
		\quad  
		\sup\limits_{(x,y)\in \mathbb{R}^2} (1+\sqrt{\tfrac{2}{\nu}})\sqrt{\big| \boldsymbol{K}_{\mathrm{G}}\big|} <\tfrac{8}{9}\,,
	\end{equation}
where   $\boldsymbol{K}_{\mathrm{G}}$ is  the Gaussian curvature of $B(x,y)$.
Then there exists a pair of approximate solution of system \eqref{1.11}-\eqref{1.12}:
	\begin{equation}\label{1.16}
		\begin{cases}
			\boldsymbol{u}_{app}^{\varepsilon}=\sum\limits_{i=0}^{2}{\delta}^i
			 \boldsymbol u^{i}(t,x,y,z,\tilde z,\bar{z}) 
			+\boldsymbol u^{c},\\
			{p}_{app}^{\varepsilon}=\sum\limits_{i=0}^{3}{\delta}^i
			  p^{i}(t,x,y,z,\tilde z,\bar{z}) 
			+p^{c}.
		\end{cases}
	\end{equation}

	Additionally, the approximate solution satisfies:
	\begin{equation}\label{1.19}
		\begin{cases}
			\partial_{t} {\boldsymbol u}_{app}^{\varepsilon}-\nu\varepsilon\Delta{\boldsymbol u}_{app}^{\varepsilon}+({\boldsymbol u}_{app}^{\varepsilon} \cdot \nabla) {\boldsymbol u}_{app}^{\varepsilon}+\varepsilon^{-1}\boldsymbol{R} {\boldsymbol u}_{app}^{\varepsilon}
			+\varepsilon^{-1}\nabla{p}_{app}^{\varepsilon}=\boldsymbol \rho^{\varepsilon}\,,\\
			\nabla\cdot {\boldsymbol u}_{app}^{\varepsilon}=0\,,\quad
			{\boldsymbol u}_{app}^{\varepsilon}|_{\partial\Omega}=0\,,\\
			{\boldsymbol u}_{app}^{\varepsilon}|_{t=0}=\boldsymbol{u}_0^\varepsilon,
		\end{cases}
	\end{equation}
	where $\boldsymbol \rho^{\varepsilon}$ satisfies
	\begin{equation*}
		\lVert
		\boldsymbol \rho^{\varepsilon}\rVert_{{L}^{2}(\Omega)}
		\lesssim
		\varepsilon^2
		\lVert
		{\boldsymbol{u}}_{0,h}
		\rVert_{{H}^{5}(\mathbb{R}^2)}
		\big(
		\lVert
		{\boldsymbol{u}}_{0,h}
		\rVert_{{H}^{3}(\mathbb{R}^2)}
		+1\big)^2
		\notag,
	\end{equation*}	
	and
	\begin{equation*}
		\lVert
		\partial_t\boldsymbol \rho^{\varepsilon}\rVert_{{L}^{2}(\Omega)}
		\lesssim
		\varepsilon^2
		\lVert
		{\boldsymbol{u}}_{0,h}
		\rVert_{{H}^{6}(\mathbb{R}^2)}
		\big(
		\lVert
		{\boldsymbol{u}}_{0,h}
		\rVert_{{H}^{4}(\mathbb{R}^2)}
		+1
		\big)^2
		\notag.
	\end{equation*}	
\end{theorem}

\begin{remark}
	The approximate solution $\boldsymbol{u}_{app}^{\varepsilon}$ in \eqref{1.16} is dominated by $\bar{\boldsymbol{u}}=(\bar {\boldsymbol u}_h,\bar u_3)$, which admits the following two-dimensional damped Euler type system with rotational effects (see detail in \eqref{2.17}-\eqref{2.22}):
	\begin{equation}\label{1.20}
		\begin{cases}
			\partial_{t} \bar {\boldsymbol u}_h+(\bar {\boldsymbol u}_h \cdot \nabla_h) \bar {\boldsymbol u}_h
			+\sqrt{\frac{\nu}{2 }}
			(\boldsymbol{H_0}
			- \boldsymbol{E_1})
			\bar{\boldsymbol u}_h+\nabla_h\bar p=0\,,\\
			\bar u_3=\nabla_h{B}\cdot\bar{\boldsymbol{u}}_h\,,\\
			\nabla\cdot \bar{\boldsymbol u}=\nabla_h \cdot \bar{\boldsymbol u}_h=0\,,\\
			\bar{\boldsymbol u}|_{t=0}={\boldsymbol u}_0=({\boldsymbol u}_{0,h},0)\,.
		\end{cases}
	\end{equation}
Given initial data \( {\boldsymbol u}_{0,h} \in H^{s} (s > 5) \) and \( B(x,y) \in C^\infty(\mathbb{R}^2) \), it can be straightforwardly shown using standard methods that the two-dimensional damped system \eqref{1.20} has a global solution. 
\end{remark}

Considering the computation of error estimates within different norm frameworks, the present paper differs from previous work \cite{Jia2024} in the approximate solution construction process. 
The primary distinctions are as follows:
Firstly, in contrast to the specific expansion of the $\delta^1$-order approximate solution conducted in Ref.~\cite{Jia2024}, this paper performs a higher-order expansion (at least $\delta^2$-order) of the approximate solution. This approach is adopted to avoid the singularities arising from higher-order derivatives.
Secondly, this paper directly considers a multi-scale expansion based on the nonlinear equations in the process of constructing the solution to ensure that the residual term of the approximate system is sufficiently small. The form of the system \eqref{1.20} satisfied by the principal part $\bar{\boldsymbol{u}}$ of the approximate solution in this paper is basically the same as that of Ref.~\cite{Jia2024}, but simplifies the variable coefficient in the damping term of the system \eqref{1.20} to a constant for theoretical and numerical convenience.

Based on approximate solutions $\boldsymbol{u}_{app}^\varepsilon$, we establish the $L^\infty$ error estimates between the local strong solution of system \eqref{1.11}-\eqref{1.12} and the approximate solution $\boldsymbol{u}_{app}^\varepsilon$. The details are as follows:
\begin{theorem}\label{thm2}
	Let $\boldsymbol u^\varepsilon\in{L}^{\infty}([0,T^*);{H}^{2} (\Omega))\cap{L}^{2}([0,T^*);{H}^{3} (\Omega))$ be a family of local strong solutions of the system \eqref{1.11}-\eqref{1.12} associated with the initial data $\boldsymbol u_0^\varepsilon\in   H^{2}(\Omega)$. Under the following conditions:
	\begin{enumerate}[1.]
		\item{Well-prepared initial data:} 
		let 
		\begin{equation}\label{1.21}
			\mathop{\lim}\limits_{\varepsilon\rightarrow 0^+} \boldsymbol u^\varepsilon_0=
			({\boldsymbol u}_{0,h}, 0)
			\quad\text{with}\quad{\boldsymbol u}_{0,h}\in H^s(\mathbb{R}^2)~(s>5)\,,
		\end{equation}
		and
		\begin{equation}\label{1.26}   
		\lVert{\boldsymbol{u}}_{0,h}\rVert_{{W}^{1,\infty}(\mathbb{R}^2)}\leq\sqrt{\tfrac{\nu}{3}}\,;
		\end{equation}
		\item{The boundary $B(x,y)\in C^\infty(\mathbb{R}^2)$ satisfies:}
		\begin{align}
			|{B}(x,y)|<\tfrac14,\quad   \tfrac{\cos^2\alpha+\cos^2\beta}{\cos^2\gamma}<\tfrac{1}{8}\,,\label{1.23}\\
			\sup\limits_{(x,y)\in \mathbb{R}^2} (1+\sqrt{\tfrac{2}{\nu}})\sqrt{\big| \boldsymbol{K}_{\mathrm{G}}\big|} <\tfrac{8}{9};\label{1.24}	
		\end{align}
		\item{Compatibility condition : with  $p_0(x,y)$ be a scalar function}
	\begin{equation}\label{compatibility}
		({\boldsymbol{u}}_{0,h}\cdot\nabla_h){\boldsymbol{u}}_{0,h}
			+\sqrt{\tfrac{\nu}{2 }}
			\nabla_{ h}B\otimes\nabla_{ h}B
			{\boldsymbol u}_{0,h}
		 +\nabla_h p_0=0.
	\end{equation} 
	\end{enumerate}
	Then, for the approximate $\boldsymbol{u}_{app}^\varepsilon$ given in Theorem \ref{thapp}, there exists $T<\min(T^*,\pi/2)$ such that
	\begin{equation}\label{1.25} \lim\limits_{\varepsilon\rightarrow0^+}
		\lVert
		\boldsymbol{u}^\varepsilon
		-\boldsymbol{u}_{app}^\varepsilon 
		\rVert_{{L}^{\infty}([0,T)\times\Omega)}=0\,.
	\end{equation}	
\end{theorem}

\begin{remark}\label{re4}
	 Here, we give a brief explanation on the conditions in this thereom.
	 The third component ($\bar{u}_3=\nabla_{ h}B\cdot\bar{\boldsymbol{u}}_h$) induced by non-flat boundary geometries in the vertical direction under the dynamical limit state poses significant challenges to establishing uniform convergence in \( L^\infty \) space. To address this, we introduce the well-prepared initial data condition for system \eqref{1.11}-\eqref{1.12}:
	$$
	\mathop{\lim}\limits_{\varepsilon\rightarrow 0^+} \boldsymbol u^\varepsilon_0=
	 ({\boldsymbol u}_{0,h}, 0).
	$$
This condition indicates that the initial horizontal velocity is parallel to the tangent plane of the surface \( z = B(x,y) \), selected to prevent the emergence of an initial layer. Noting that \( \nabla_h \cdot \boldsymbol{u}_{0,h}(x_h) = 0 \), we have
\begin{equation}\label{compatibility1}
		\boldsymbol{E_1}\boldsymbol{u}_{0,h}\sim\nabla_{h}\phi(x_h).
	\end{equation}
	Then, by using the system \eqref{1.20}, we take the compatibility condition as \eqref{compatibility}.
\end{remark}

The main contributions of this paper are summarized as follows.

First, we identify the feedback mechanism through which non-flat boundary geometries affect the near-boundary layer flow.
    The multi-scale construction procedure and its principles are standard, but establishing the  boundary layer terms is complex due to the dependency on geometric features. During the development of the multi-scale approximate solutions, we noted intriguing regularities in the geometric structures, which may offer new insights for designing future numerical schemes. A key distinction from the flat-boundary case is the introduction of a normal (vertical) component \(\boldsymbol{u}_3=\nabla_h B\cdot \boldsymbol{u}_h\). To ensure that the corresponding approximate solutions converge to \(\boldsymbol{u}^\varepsilon \), controlling  the evolution of this quantity is essential. This work addresses this challenge by applying suitable initial and compatibility conditions.  From this mechanism, we derive a limiting-state system  that describes the coupled boundary-flow interaction (see \eqref{1.20}). 
    
Second, from the mathematical proof aspect, we develop approximate solutions that do not depend on the typical small-amplitude assumption for boundary perturbations; instead, they are based on the Gaussian curvature of the boundary surface. Notably, our approach recovers classical results in the flat-boundary limit: when the boundary is planar, defined by \(B(x,y)=\text{const}\), the solutions align with existing studies.

%

\subsection{Preliminaries}\label{subsec1.3}	In this subsection, two Gronwall-type inequalities are given for the sake of the subsequent proofs.
\begin{lemma}[Grenier and Masmoudi \cite{Grenier1997}]\label{f2}
	Let $\varepsilon_0>0$, and assume $f(t)$ and $a_i(t)~(i=0,1,2)$ are nonnegative functions satisfuing 
	\begin{equation*}
		\frac{d}{dt}(f^2(t))\leq
		a_0(t)f^2(t)+a_1(t)f(t)+a_2(t)\,,\quad\text{with}\quad
		f(0)\leq C\varepsilon_0\,,
	\end{equation*}
	and
	\begin{equation*}
		\int_{0}^{T}a_i(t)\,dt\leq C\varepsilon_0^i\,,\quad\text{with}\quad
		i=0,1,2.
	\end{equation*}
	Then there exists some $M>0$ depending only on $C$, for all $0\leq t\leq T$, such that $f\leq M\varepsilon_0$.
\end{lemma}

\begin{lemma}\label{f4}
	If \( g(0) = 0 \), \( g(t) \in C^1 \), and   
	\begin{equation}\label{1.28}
		\frac{d}{dt}g(t) \leq C_1^2 g^2(t) + C_2^2,
	\end{equation}
	where \( C_1, C_2 \) are positive constants, then the inequality 
	\begin{equation}\label{1.29}
		g(t) \leq \tan(C_1 C_2 t) \frac{C_2}{C_1}\,.
	\end{equation}
	holds for time  $t< \frac{\pi}{2C_1 C_2}$.
\end{lemma}
\begin{proof}
	By using \eqref{1.28}, we have
	\begin{equation*}
		\frac{\frac{c_1 g'(t)}{C_2}}{(\frac{c_1 g'(t)}{C_2})^2+1} \leq C_1  C_2\,,
	\end{equation*}
	so we have \eqref{1.29} when $t< \frac{\pi}{2C_1 C_2}$.
\end{proof}

This paper is organized as follows: 
Section \ref{sec2} extends the approximate solution to higher order. The process involves initially constructing the top and bottom boundary terms as well as the interior term, then coupling the top and bottom boundary regions using a cut-off function, and finally revising the incompressibility condition of the approximate solution. Section \ref{sec4} derive the $L^\infty$ error estimate between the approximate solution and the local strong solution of the three-dimensional system by establishing estimates related to the approximate solution and lower-order error estimates.

\section{Construction of the approximate solution}\label{sec2}
In this section, we will discuss the higher-order expansion of the solution in conjunction with the effect of the nonlinear term and the top and bottom boundary layers on the approximate solution,completing the proof of Theorem \ref{thapp}.

\begin{proof}[Proof of Theorem~{\upshape\ref{thapp}}]
In fact, the nonlinear terms do not affect the size of the boundary layer and the construction of the approximate solution but impact its stability \cite{Gerard2005}. In particular, to avoid singularities for the higher-order derivatives of the residual term in the $L^\infty$ framework, we need to consider nonlinear effects in the construction of the solution.
Clearly, some of the nonlinear terms are coupled terms involving both the top and bottom boundary terms.
Therefore, we introduce  $\chi(z)\in C^\infty(B(x,y), B(x,y)+2)$ is a cutoff function satisfying  $\chi'(z)\leq1$ and
	\begin{equation}\label{1.17}
		\chi(z)=
		\begin{cases}
			\begin{aligned}
				&0\,,  &B\leqslant z\leqslant B+1/2\,,\\
				&1\,,  &B+3/2 \leqslant z\leqslant B+2\,.
			\end{aligned}
		\end{cases}
	\end{equation}

We will approximate the solution as follows:
\begin{equation}\label{2.1}
	\begin{cases}
		\boldsymbol{u}_{app}^{\varepsilon}=\sum\limits_{i=0}^{m}{\delta}^i
		\big[\boldsymbol u^{I,i}(t,x,y,z)
		+(1-\chi)\boldsymbol u^{B,i}(t,\tilde x,\tilde y,\tilde z)
		+\chi\boldsymbol u^{T,i}(t,\bar{x},\bar{y},\bar{z})\big]
		,\\
		p_{app}^{\varepsilon}=\sum\limits_{i=0}^{m}{\delta}^i\big[p^{I,i}(t,x,y,z)+(1-\chi)p^{B,i}(t,\tilde x,\tilde y,\tilde z)+\chi p^{T,i}(t,\bar{x},\bar{y},\bar{z})
		\big],
	\end{cases}
\end{equation}
where the superscripts $I, T, B, c$ denote the internal, top-boundary, bottom-boundary, and correction terms respectively.

Since the approximate solution $\boldsymbol{u}^\varepsilon_{app}$ needs to satisfy the Dirichlet boundary condition, combined with the setting of the truncation function $\chi(z)$, the boundary and interior terms need to satisfy the following boundary conditions:
\begin{align}
	\boldsymbol u^{B,i}(t,\tilde{x},\tilde{y},\tilde{z})|_{\tilde z=0}&=-\boldsymbol u^{I,i}(t,x,y,z)|_{z=B}\,,\label{2.2}\\
	\boldsymbol u^{T,i}(t,\bar{x},\bar{y},\bar{z})|_{\bar z=0}&=-\boldsymbol u^{I,i}(t,x,y,z)|_{z=B+2}\,.\label{2.3}
\end{align}

\subsection{Calculation of boundary coordinates}\label{subsec2.1}
Given the expression of the approximate solution \eqref{2.1}, when calculating derivatives within the boundary layer under the new coordinate systems $(\tilde x,\tilde y,\tilde z)=(x,y,\frac{z - B}{\sqrt{\nu}\varepsilon}\cos^{\frac{3}{2}}\gamma)$ and $(\bar x,\bar y,\bar z)=(x,y,\frac{B + 2 - z}{\sqrt{\nu}\varepsilon}\cos^{\frac{3}{2}}\gamma)$, singularities may occur. To streamline subsequent derivations and notations, we decompose the gradient and Laplace operators based on the orders of $\delta$.

We define $\nabla_{\tilde{h}}=(\partial_{\tilde{x}},\partial_{\tilde{y}})^T$, $\nabla_{\bar{h}}=(\partial_{\bar{x}},\partial_{\bar{y}})^T$, $\Delta_{\tilde{h}}=\partial^2_{\tilde{x}}+\partial^2_{\tilde{y}}$, and $\Delta_{\bar{h}}=\partial^2_{\bar{x}}+\partial^2_{\bar{y}}$. Taking the bottom boundary as an illustration, we obtain the expression
\begin{align}
	\nabla=\,
	&\tilde{\nabla}_{\delta^{0}}
	+\delta^{-1}\tilde{\nabla}_{\delta^{-1}}\,,\label{2.4}\\
	\Delta=\,
	&\tilde{\Delta}_{\delta^{0}}
	+\delta^{-1}\tilde{\Delta}_{\delta^{-1}}
	+\delta^{-2}\tilde{\Delta}_{\delta^{-2}}\,,\label{2.5}
\end{align}
where
\begin{align*}
	\tilde{\nabla}_{\delta^{0}}=&
	\,\big(\nabla_{\tilde{h}}-\tfrac{3}{2}\tilde{z}\tan\gamma\nabla_{h}\gamma\partial_{\tilde{z}},0\big)^T\,,\\	
	\tilde{\Delta}_{\delta^{0}}
	=&\,\Delta_{\tilde{h}}+\tfrac{9}{4}|\tilde{z}\tan\gamma\nabla_{h}\gamma|^2\partial_{\tilde{z}}^2
	-3\tilde{z}
	\tan\gamma\nabla_{h}\gamma\cdot\nabla_{\tilde h}\partial_{\tilde{z}}
	\\
	& \,
	+\tilde{z}\big(\tfrac{3}{4}(1-2\cot^2\gamma)|\tan\gamma\nabla_{h}\gamma|^2
	-\tfrac{3}{2}\tan\gamma\Delta_h\gamma
	\big)\partial_{\tilde{z}}
	\,,\\
	\tilde{\nabla}_{\delta^{-1}}=&
	\,
	\nabla(\delta\tilde{z})\partial_{\tilde{z}}\,,\\
	\tilde{\Delta}_{\delta^{-1}}=&
	\,
	2\nabla_h(\delta\tilde{z})\cdot\nabla_{\tilde{h}}
	\partial_{\tilde{z}}
	-3\tilde{z}\tan\gamma\nabla_{h}\gamma\cdot\nabla_{h}(\delta\tilde{z})
	\partial^2_{\tilde{z}}\\
	&	\,+\Delta_h(\delta\tilde{z})\partial_{\tilde{z}}
	-3\tan\gamma\nabla_{h}\gamma\cdot\nabla_{h}(\delta\tilde{z})
	\partial_{\tilde{z}}
	\,,\\
	\tilde{\Delta}_{\delta^{-2}}=&
	\,\cos^{-2}\gamma\partial^2_{\tilde{z}}\,.
\end{align*}

For the terms related to the top boundary layer, we set 
\begin{align}
	\nabla=\,
	&\bar{\nabla}_{\delta^{0}}
	+\delta^{-1}\bar{\nabla}_{\delta^{-1}}
	\,,\label{2.6}\\
	\Delta=\,
	&\bar{\Delta}_{\delta^{0}}
	+\delta^{-1}\bar{\Delta}_{\delta^{-1}}
	+\delta^{-2}\bar{\Delta}_{\delta^{-2}}
	\,.\label{2.7}
\end{align}
The forms of equations \eqref{2.6}-\eqref{2.7} are similar to those of equations \eqref{2.4}-\eqref{2.5}, where $\bar{z}$ and $\bar{h}$ simply substitutes for $\tilde{z}$ and $\tilde{h}$.

Furthermore, for the sake of concise expression, we write
\begin{equation}\label{2.8}
	\boldsymbol{u}^{i}=
	\boldsymbol u^{I,i}
	+(1-\chi)\boldsymbol u^{B,i}
	+\chi\boldsymbol u^{T,i}\,,\quad\text{with}\quad
	i=0,1,2,\cdots.
\end{equation}

\subsection{Asymptotic analysis}\label{subsec2.2}
In this subsection, we asymptotically analyze the approximate solution $\boldsymbol{u}_{app}^\varepsilon$ by substituting it into the following approximate system:
\begin{equation}\label{2.9}
	\begin{cases}
		\partial_{t} {\boldsymbol u}_{app}^{\varepsilon}-\nu\varepsilon\Delta{\boldsymbol u}_{app}^{\varepsilon}+({\boldsymbol u}_{app}^{\varepsilon} \cdot \nabla) {\boldsymbol u}_{app}^{\varepsilon}+\varepsilon^{-1}\boldsymbol{R} {\boldsymbol u}_{app}^{\varepsilon}
		+\varepsilon^{-1}\nabla{p}_{app}^{\varepsilon}=\boldsymbol \rho^\varepsilon\,,\\
		\nabla\cdot {\boldsymbol u}_{app}^{\varepsilon}=0\,,\quad
		{\boldsymbol u}_{app}^{\varepsilon}|_{\partial\Omega}=0\,,
	\end{cases}
\end{equation}
and give an expression for the pressure ${p}_{app}^\varepsilon$ accordingly.

\subsubsection{Construction of ($\boldsymbol{u}^{I,0}, \boldsymbol{u}^{B,0}, \boldsymbol{u}^{T,0}$)  and (${p}^{I,0}, {p}^{B,0}, {p}^{T,0},{p}^{B,1},{p}^{T,1}$)}\label{subsec2.2.1}
First, from the boundary $\delta^{-2}$-order part of the approximate equations, it can be inferred that $\partial_{\tilde{z}}p^{B,0}=\partial_{\tilde{z}}p^{T,0}=0$.  Without loss of generality, we take 
\begin{equation}\label{2.10}
	p^{B,0}=p^{T,0}=0\,.
\end{equation}

Second, from the $\delta^{-1}$-order part of the boundary in the incompressibility condition, the non-flat boundary induces a vertical component of the $\delta^{0}$-order boundary terms:
\begin{equation}\label{2.11}
	u_3^{T,0}(t,\bar{x},\bar{y},\bar{z})=\nabla_hB\cdot\boldsymbol u_h^{T,0}\,,
\end{equation} 
and
\begin{equation}\label{2.12}
	u_3^{B,0}(t,\tilde{x},\tilde{y},\tilde{z})=\nabla_hB\cdot\boldsymbol u_h^{B,0}\,.
\end{equation}
Obviously, from boundary conditions \eqref{2.2} and \eqref{2.3}, it follows that the $\delta^{0}$-order interior term satisfies
\begin{equation}\label{2.13}
	u_3^{I,0}=\nabla_hB\cdot\boldsymbol u_h^{I,0}\,.
\end{equation}

Then, it follows from the internal $\delta^{-1}$-order part that $p^{I,0}$ and $\boldsymbol{u}^{I,0}$ are independent of $z$ and have the following relationship:
\begin{equation}\label{2.14}
	p^{I,0}(t,x,y)=\Delta_h^{-1}\nabla^\perp_h \cdot \boldsymbol{u}^{I,0}_h\,.
\end{equation}
Also, according to the internal term $\delta^{0}$-order part, the following system holds		
\begin{equation*}
	\begin{cases}
		\partial_{t} \boldsymbol{u}^{I,0}_h+\boldsymbol{u}^{I,0}_h \cdot \nabla_h \boldsymbol{u}^{I,0}_h
		-\sqrt{\nu}\cos^{-\frac32}\gamma
		\boldsymbol{E_1}\boldsymbol{u}^{I,1}_h
		+\nabla_h(\sqrt{\nu}\cos^{-\frac32}\gamma
		{p}^{I,1})=0\,,\\
		\nabla\cdot \boldsymbol{u}^{I,0}=\nabla_h \cdot \boldsymbol{u}^{I,0}_h=0\,.
	\end{cases}
\end{equation*}

Finally, based on the above analysis, the bottom boundary $\delta^{-1}$-order part of the approximate equations reduces to
\begin{equation*}
	\begin{cases}	
		\boldsymbol{H_0} \partial_{\tilde z}^2\boldsymbol u_h^{B,0}+\cos^{-1}\gamma\boldsymbol{E_1}\boldsymbol u_h^{B,0}=0\,,\\
		\partial_{\tilde{z}}p^{B,1}=\cos\gamma\partial_{\tilde{z}}^2 u_3^{B,0},\\
		\boldsymbol u^{B,0}_h|_{\tilde z=0}=-\boldsymbol{u}^{I,0}_h,\quad
		\lim\limits_{\tilde{z}\rightarrow \infty}\boldsymbol u^{B,0}_h=0\,,
	\end{cases}
\end{equation*}
which is given by
\begin{equation}\label{2.15}
	\begin{cases}
		\boldsymbol u_h^{B,0}=\boldsymbol{M}(\tfrac{\tilde z}{\sqrt 2})
		\boldsymbol{u}^{I,0}_h\,,\\
		p^{B,1}=\mathrm{e}^{-\frac{\tilde z}{\sqrt 2}}\cos\gamma\big( \sin(\tfrac{\tilde z}{\sqrt 2}-\tfrac\pi4)\cos\gamma\nabla^\perp_h B
		+\cos(\tfrac{\tilde z}{\sqrt 2}-\tfrac\pi4)\nabla_hB\big)\cdot\boldsymbol{u}^{I,0}_h\,,
	\end{cases}
\end{equation}
where
\begin{equation*}
	\boldsymbol{M}(Z)=
	-\mathrm{e}^{-Z}
	\Big(
	\cos(Z)\boldsymbol{E}+\sin(Z)\cos\gamma\boldsymbol{E_1}\boldsymbol{H_0}
	\Big)\,.
\end{equation*}
In a similar way, $(\boldsymbol{u}^{T,0}_h,p^{T,1})$ can be defined as
\begin{equation}\label{2.16}
	\begin{cases}
		\boldsymbol u_h^{T,0}=\boldsymbol{M}(\tfrac{\bar z}{\sqrt 2})
		\boldsymbol{u}^{I,0}_h\,,\\
		p^{T,1}=\mathrm{e}^{-\frac{\bar z}{\sqrt 2}}\cos\gamma
		\big( 
		\sin(\tfrac\pi4-\tfrac{\bar z}{\sqrt 2})\cos\gamma\nabla^\perp_h B
		-\cos(\tfrac\pi4-\tfrac{\tilde z}{\sqrt 2})\nabla_hB\big)\cdot\boldsymbol{u}^{I,0}_h\,.
	\end{cases}
\end{equation}

\subsubsection{Construction of ($\boldsymbol{u}^{I,1}, \boldsymbol{u}^{B,1}, \boldsymbol{u}^{T,1}$)  and (${p}^{I,1}, {p}^{B,2}, {p}^{T,2}$)}\label{subsec2.2.2}
For the $\delta^1$-order top and bottom boundary layer terms, the following relation follows from the incompressibility conditions:
\begin{align}\label{2.17}
	\nabla(z-B)\cdot \boldsymbol{u}^{B,1}
	=&	
	\tfrac{3}{2}\cos^2\gamma\nabla_{h}B^T\boldsymbol{H}\big(
	\mathrm{e}^{-\frac{\tilde z}{\sqrt 2}}\sin(\tfrac{\tilde z}{\sqrt 2}-\tfrac\pi4)\boldsymbol{u}^{I,0}_h+\tilde{z}\boldsymbol u_h^{B,0}
	\big)\notag\\
	&-\mathrm{e}^{-\frac{\tilde z}{\sqrt 2}}\cos(\tfrac{\tilde z}{\sqrt 2}-\tfrac\pi4)
	\nabla_h\cdot(\cos\gamma\boldsymbol{E_1}\boldsymbol{H_0}\boldsymbol{u}^{I,0}_h)
	\notag\\
	&-\mathrm{e}^{-\frac{\tilde z}{\sqrt 2}}\cos(\tfrac{\tilde z}{\sqrt 2}-\tfrac\pi4)
	\tfrac{3\cos^2\gamma}{2}\nabla_{h}B^T\boldsymbol{H}\boldsymbol{E_1}\boldsymbol{H_0}\boldsymbol{u}^{I,0}_h\,,
\end{align}
and
\begin{align}\label{2.18}
	\nabla(z-B)\cdot \boldsymbol{u}^{T,1}
	=
	&
	-\tfrac{3}{2}\cos^2\gamma\nabla_{h}B^T\boldsymbol{H}\big(
	\mathrm{e}^{-\frac{\bar z}{\sqrt 2}}\sin(\tfrac{\bar z}{\sqrt 2}-\tfrac\pi4)\boldsymbol{u}^{I,0}_h+\bar{z}\boldsymbol u_h^{T,0}
	\big)\notag\\
	&+\mathrm{e}^{-\frac{\bar z}{\sqrt 2}}\cos(\tfrac{\bar z}{\sqrt 2}-\tfrac\pi4)
	\nabla_h\cdot(\cos\gamma\boldsymbol{E_1}\boldsymbol{H_0}\boldsymbol{u}^{I,0}_h)
	\notag\\
	&+\mathrm{e}^{-\frac{\bar z}{\sqrt 2}}\cos(\tfrac{\bar z}{\sqrt 2}-\tfrac\pi4)\tfrac{3\cos^3\gamma}{2}\nabla_{h}B^T\boldsymbol{H}\boldsymbol{E_1}\boldsymbol{H_0}\boldsymbol{u}^{I,0}_h\,.
\end{align}

Using the boundary condition for the interior term $\boldsymbol{u}^{I,1}$, it follows that 
\begin{equation}\label{2.19}
		u_3^{I,1}|_{z=B}=-u_3^{B,1}|_{\tilde{z}=0}\,,\quad
		u_3^{I,1}|_{z=B+2}=-u_3^{T,1}|_{\bar{z}=0}\,.
\end{equation}
Combining \eqref{2.17}-\eqref{2.19}, we can define $\boldsymbol{u}^{I,1}$ in a natural way that
\begin{equation}\label{2.20}
	\begin{cases}
		\boldsymbol u_h^{I,1}
		=\tfrac{1}{\sqrt{2}}
		\cos^{\frac32}\gamma
		\big(
		\boldsymbol{E_1}\boldsymbol{H_0}+\boldsymbol{E}	
		\big)\boldsymbol{u}^{I,0}_h\,,\\
		u_3^{I,1}
		=
		-(1-\chi)	u_3^{B,1}|_{z=B}-\chi	u_3^{T,1}|_{z=B+2}\,,
	\end{cases}
\end{equation}
where the structure of $\boldsymbol u_h^{I,1}$ follows the framework of the literature \cite{Jia2024}, but its coefficients are simplified for easier numerical realization.

Therefore, combining the expression $\eqref{2.20}_1$ for $\boldsymbol{u}^{I,1}_h(t,x,y)$, the limiting system is ultimately obtained as
\begin{equation}\label{2.21}
	\begin{cases}
		\partial_{t} \boldsymbol{u}^{I,0}_h+(\boldsymbol{u}^{I,0}_h \cdot \nabla_h) \boldsymbol{u}^{I,0}_h
		+\sqrt{\frac{\nu}{2}}
		(\boldsymbol{H_0}
		-\boldsymbol{E_1})
		\boldsymbol{u}^{I,0}_h+\nabla_h\bar{p}=0\,,\\
		u^{I,0}_3=\nabla_h{B}\cdot\boldsymbol{u}^{I,0}_h\,,\\
		\nabla\cdot \boldsymbol{u}^{I,0}=\nabla_h \cdot \boldsymbol{u}^{I,0}_h=0\,,
	\end{cases}
\end{equation}
where 
\begin{equation}\label{2.22}
	\bar{p}=\sqrt{\nu}\cos^{-\frac32}\gamma
	{p}^{I,1}\,.
\end{equation}

In order to obtain expressions for the $\delta^{1}$-order boundary terms, we simplify and diagonalize the $\delta^{0}$-order boundary part of the approximate equations. Taking $\boldsymbol{u}_h^{B,1}$ as an example, we have
\begin{equation}\label{2.23}
	\partial_{\tilde z}^2 \boldsymbol u_h^{B,1}+
	\boldsymbol{Q}\begin{pmatrix}
		i&0\\
		0&-i
	\end{pmatrix}\boldsymbol{Q}^{-1}\boldsymbol u_h^{B,1}=\boldsymbol{H_0}^{-1}
	\bigg(\begin{pmatrix}
		1&0&B_x\\
		0&1&B_y	
	\end{pmatrix}\boldsymbol{D}^{B,1}+\partial_{\tilde{z}} D^{B,1}_0\nabla_h B\bigg)\,,
\end{equation}
where
\begin{equation*}
	\boldsymbol{Q}=\begin{pmatrix}
		\cos\gamma(1+B_y^2)
		&\cos\gamma(1+B_y^2)\\
		i-\cos\gamma B_xB_y
		&-i-\cos\gamma B_xB_y
	\end{pmatrix}\,,
\end{equation*}
and
\begin{align*}
	\boldsymbol{D}^{B,1}
	=&(\cos\gamma\delta)^{-1}\tilde{\nabla}_{\delta^0}(\delta p^{B,1})
	-\cos^2\gamma\tilde{\Delta}_{\delta^{-1}}\boldsymbol u^{B,0}
	\\
	&
	+\sqrt{\tfrac{\cos\gamma}{\nu}}
	\Big(
	\partial_t\boldsymbol u^{B,0}+\boldsymbol{u}^{0}\cdot\tilde{\nabla}_{\delta^0}\boldsymbol{u}^{B,0}
	+(\boldsymbol{u}_h^{B,0}\cdot\nabla_h)\boldsymbol{u}^{I,0}+\boldsymbol{u}^{1}\cdot\nabla(\delta\tilde{z})\partial_{\tilde{z}}\boldsymbol{u}^{B,0}
	\Big)\,,\\
	D^{B,1}_0=
	&\tilde{\nabla}_{\delta^0}\cdot\boldsymbol u^{B,0}\,.
\end{align*}
For ease of presentation, for $i\in\mathbb{Z_+}$, the following notation is given with respect to the top and bottom boundary layers:
\begin{align*}
	&\begin{pmatrix}
		G^{B,i}_1\\
		G^{B,i}_2
	\end{pmatrix}=\boldsymbol{H_0}^{-1}
	\bigg(\begin{pmatrix}
		1&0&B_x\\
		0&1&B_y	
	\end{pmatrix}\boldsymbol{D}^{B,i}+\partial_{\tilde{z}} D^{B,i}_0\nabla_h B\bigg)\,,\\
	&\begin{pmatrix}
		G^{T,i}_1\\
		G^{T,i}_2
	\end{pmatrix}=\boldsymbol{H_0}^{-1}
	\bigg(\begin{pmatrix}
		1&0&B_x\\
		0&1&B_y	
	\end{pmatrix}\boldsymbol{D}^{T,i}+\partial_{\bar{z}} D^{T,i}_0\nabla_h B\bigg)\,.
\end{align*}
The structure of $\boldsymbol{D}^{T,i}$ and ${D}_0^{T,i}$
are analogly to that of  $\boldsymbol{D}^{B,i}$ and ${D}_0^{B,i}$,
with the only modification being the replacement of the bottom boundary terms and their derivative operators in 
$\boldsymbol{D}^{B,i}$ and ${D}_0^{B,i}$ by the corresponding top boundary terms.
To streamline the presentation, we focus solely on the bottom boundary terms in the following discussion.

Note that due to the $\delta^{1}$-order part of the approximate solution having relations \eqref{2.17}, \eqref{2.18} and \eqref{2.20}, it is known that $\boldsymbol{u}^{1}\cdot\nabla(\delta\tilde{z})=\big(\boldsymbol{u}^{I,1}+(1-\chi)\boldsymbol{u}^{B,1}+\chi\boldsymbol{u}^{T,1}\big)\cdot\nabla(z-B)$ in $\boldsymbol{D}^{B,1}$.
Thus, solving equations \eqref{2.23},  yields
\begin{equation}\label{2.24}
	\boldsymbol u_h^{B,1}=
	\boldsymbol{M}(\tfrac{\tilde z}{\sqrt 2})
	\boldsymbol u_h^{I,1}
	+
	Re\boldsymbol{Q}
	\int_0^{\tilde{z}}
	\boldsymbol{Q}^{-1}\begin{pmatrix}
		{\mathrm{e}}^{-\sqrt{-i}(\tilde{z}-\tau)} G^{B,1}_1(\tau)
		\\
		{\mathrm{e}}^{-\sqrt{i}(\tilde{z}-\tau)} G^{B,1}_2(\tau) 
	\end{pmatrix}{d}\tau\,.
\end{equation}
Similarly, we have
\begin{equation}\label{2.25}
	\boldsymbol u_h^{T,1}=
	\boldsymbol{M}(\tfrac{\bar z}{\sqrt 2})
	\boldsymbol u_h^{I,1}
	+
	Re\boldsymbol{Q}
	\int_0^{\bar{z}}
	\boldsymbol{Q}^{-1}\begin{pmatrix}
		{\mathrm{e}}^{-\sqrt{-i}(\bar{z}-\tau)} G^{T,1}_1(\tau)
		\\
		{\mathrm{e}}^{-\sqrt{i}(\bar{z}-\tau)} G^{T,1}_2(\tau) 
	\end{pmatrix}{d}\tau\,.
\end{equation}

Meanwhile, from the relation equation between $p^{B,2}$ and $u_3^{B,1}$, we can get
\begin{equation}\label{2.26}
	p^{B,2}=-\int_{\tilde{z}}^{+\infty}\big(
	\cos\gamma\partial_{\tilde{z}'}^2 u_3^{B,1}-\cos\gamma D_3^{B,1}
	\big)d\tilde{z}'\,.
\end{equation}
Similarly, we have
\begin{equation}\label{2.27}
	p^{T,2}=-\int_{\bar{z}}^{+\infty}\big(
	-\cos\gamma\partial_{\bar{z}'}^2 u_3^{T,1}+\cos\gamma D_3^{T,1}
	\big)d\bar{z}'\,.
\end{equation}

\subsubsection{Construction of ($\boldsymbol{u}^{I,2}, \boldsymbol{u}^{B,2}, \boldsymbol{u}^{T,2}$)  and (${p}^{I,2}, {p}^{B,3}, {p}^{T,3}$)}\label{subsec2.2.3}

Given the stronger norm framework employed in this study to explore convergence, we require more sufficiently small remainder terms. 
Unlike the prior work, which only constructed the third components of $\boldsymbol{u}^{B,2}$ and $\boldsymbol{u}^{T,2}$ to satisfy the incompressibility conditions coupled with $\boldsymbol{u}^{B,1}$ and $\boldsymbol{u}^{T,1}$, we specifically construct all components of $\boldsymbol{u}^{B,2}$ and $\boldsymbol{u}^{T,2}$.

Firstly, similar to the previous parts, we need to simplify 
$\boldsymbol{u}^{2}\cdot\nabla(\delta\tilde{z})$
according to the incompressibility condition for ease of computation.
It is worth noting that $\boldsymbol{u}^{1}\cdot\nabla(\delta\tilde{z})$ was previously shown to be derived and expressed in terms of $\boldsymbol{u}^{I,0}_h$ because of the need to construct the limit system. 
Here, for $\boldsymbol{u}^{2}\cdot\nabla(\delta\tilde{z})$ and remaining items, only the construction is needed, and no detailed expression is required.

From the incompressibility conditions of $\boldsymbol{u}^{B,2}$ and $\boldsymbol{u}^{T,2}$, it follows that
\begin{equation*}
	\nabla(z-B)\cdot\boldsymbol{u}^{B,2}=
	\int_{\tilde{z}}^{+\infty}
	\delta^{-1}\tilde{\nabla}_{\delta^0}
	\cdot(\delta\boldsymbol{u}^{B,1})
	d\tilde{z}'\,,
\end{equation*}
and
\begin{equation*}
	\nabla(z-B)\cdot\boldsymbol{u}^{T,2}=
	-\int_{\bar{z}}^{+\infty}
	\delta^{-1}\bar{\nabla}_{\delta^0}
	\cdot(\delta\boldsymbol{u}^{T,1})
	d\bar{z}'\,.
\end{equation*}

Let us now solve for $\boldsymbol{u}^{I,2}$. First, the $\delta^1$-order interior part of the approximate equations is organized as
\begin{multline}\label{2.28}
	\partial_{t} \boldsymbol{u}^{I,1}
	-\sqrt{\nu}\cos^{\frac{3}{2}}\gamma\Delta_h\boldsymbol{u}^{I,0}
	+(\boldsymbol{u}_h^{I,1}\cdot\nabla_{ h})\boldsymbol{u}^{I,0}
	+	\delta^{-1}(\boldsymbol{u}^{I,0}\cdot\nabla)(\delta\boldsymbol{u}^{I,1})\\
	+\sqrt{\nu}\cos^{-\frac{3}{2}}\gamma\boldsymbol{R}\boldsymbol{u}^{I,2}
	+\sqrt{\nu}\cos^{\frac{3}{2}}\gamma\nabla(\cos^{-3}\gamma p^{I,2})
	=0\,.
\end{multline} 
In particular, the third part implies 
\begin{align*}
	\partial_z(\cos^{-3}\gamma p^{I,2})
	=&-(\sqrt{\nu}\cos^{\frac{3}{2}}\gamma)^{-1}
	(\boldsymbol{u}_h^{I,1}\cdot\nabla_{ h}){u}^{I,0}_3
	+\Delta_h{u}^{I,0}_3\\
	&-(\sqrt{\nu}\cos^{\frac{3}{2}}\gamma)^{-1}\delta^{-1}
	\big(
	\partial_{t}	+\boldsymbol{u}^{I,0}\cdot\nabla\big)(\delta{u}_3^{I,1})\,.
\end{align*}
Thus, we have
\begin{align}\label{2.29}
	p^{I,2}
	=&-\cos^3\gamma
	\int_0^z\big((\sqrt{\nu}\cos^{\frac{3}{2}}\gamma)^{-1}
	(\boldsymbol{u}_h^{I,1}\cdot\nabla_{ h}){u}^{I,0}_3
	+\Delta_h{u}^{I,0}_3\big)\,dz'
	\notag\\
	&-\cos^3\gamma
	\int_0^z\big(
	(\sqrt{\nu}\cos^{\frac{3}{2}}\gamma)^{-1}\delta^{-1}
	\big(
	\partial_{t}	+\boldsymbol{u}^{I,0}\cdot\nabla\big)(\delta{u}_3^{I,1})
	\big)\,dz'.
\end{align}
Then, it follows from \eqref{2.28} that
\begin{align}\label{2.30}
	\boldsymbol{u}_h^{I,2}
	=&-\tfrac{\cos^{\frac{3}{2}}\gamma}{\sqrt{\nu}}\boldsymbol{E_1}\big(
	\partial_{t} \boldsymbol{u}_h^{I,1}
	-\sqrt{\nu}\cos^{\frac{3}{2}}\gamma\Delta_h\boldsymbol{u}^{I,0}_h
	+(\boldsymbol{u}_h^{I,1}\cdot\nabla_{ h})\boldsymbol{u}^{I,0}_h\big)\notag\\
	&-\tfrac{\cos^{\frac{3}{2}}\gamma}{\sqrt{\nu}}\boldsymbol{E_1}\big(	\delta^{-1}(\boldsymbol{u}_h^{I,0}\cdot\nabla_h)(\delta\boldsymbol{u}_h^{I,1})	
	+\sqrt{\nu}\cos^{\frac{3}{2}}\gamma\nabla_h(\cos^{-3}\gamma p^{I,2})
	\big)\,.
\end{align} 
Based on this, and using the boundary conditions \eqref{2.2}-\eqref{2.3} of $\boldsymbol{u}^{B,2}$ and $\boldsymbol{u}^{T,2}$, as well as the setup \eqref{1.17} of $\chi(z)$, we can derive
\begin{align}\label{2.31}
	u_3^{I,2}
	=&-(1-\chi)	u_3^{B,2}|_{z=B}-\chi	u_3^{T,2}|_{z=B+2}\notag\\
	=&(1-\chi)\Big(
	\nabla_hB\cdot\boldsymbol{u}_h^{I,2}|_{z=B}
	-\int_{0}^{+\infty}
	\delta^{-1}\tilde{\nabla}_{\delta^0}
	\cdot(\delta\boldsymbol{u}^{B,1})
	d\tilde{z}'
	\Big)\notag\\
	&+\chi\Big(
	\nabla_hB\cdot\boldsymbol{u}_h^{I,2}|_{z=B+2}
	+\int_{0}^{+\infty}
	\delta^{-1}\bar{\nabla}_{\delta^0}
	\cdot(\delta\boldsymbol{u}^{T,1})
	d\bar{z}'
	\Big)\,.
\end{align}

Repeating the process of Subsubsection \ref{subsec2.2.2} yields
\begin{equation}\label{2.32}
	\begin{cases}
		\boldsymbol u_h^{B,2}=
		\boldsymbol{M}(\tfrac{\tilde z}{\sqrt 2})
		\boldsymbol u_h^{I,2}|_{z=B}
		+
		Re\boldsymbol{Q}
		\int_0^{\tilde{z}}
		\boldsymbol{Q}^{-1}\begin{pmatrix}
			{\mathrm{e}}^{-\sqrt{-i}(\tilde{z}-\tau)} G^{B,2}_1(\tau)
			\\
			{\mathrm{e}}^{-\sqrt{i}(\tilde{z}-\tau)} G^{B,2}_2(\tau) 
		\end{pmatrix}{d}\tau\,,\\
		u_3^{B,2}=\nabla_{h}B\cdot\boldsymbol u_h^{B,2}+\int_{\tilde{z}}^{+\infty} D_0^{B,2}d\tilde{z}'\,,\\
		p^{B,3}=-\int_{\tilde{z}}^{+\infty}\big(
		\cos\gamma\partial_{\tilde{z}'}^2 u_3^{B,2}-\cos\gamma D_3^{B,2}
		\big)d\tilde{z}'\,,
	\end{cases}
\end{equation}
and
\begin{equation}\label{2.33}
	\begin{cases}
		\boldsymbol u_h^{T,2}=
		\boldsymbol{M}(\tfrac{\bar z}{\sqrt 2})
		\boldsymbol u_h^{I,2}|_{z=B+2}
		+
		Re\boldsymbol{Q}
		\int_0^{\bar{z}}
		\boldsymbol{Q}^{-1}\begin{pmatrix}
			{\mathrm{e}}^{-\sqrt{-i}(\bar{z}-\tau)} G^{T,2}_1(\tau) 
			\\
			{\mathrm{e}}^{-\sqrt{i}(\bar{z}-\tau)} G^{T,2}_2(\tau)
		\end{pmatrix} {d}\tau\,,\\	
		u_3^{T,2}=\nabla_{h}B\cdot\boldsymbol u_h^{T,2}+\int_{\bar{z}}^{+\infty} D_0^{T,2}d\bar{z}'\,,\\
		p^{T,3}=-\int_{\bar{z}}^{+\infty}\big(
		-\cos\gamma\partial_{\bar{z}'}^2 u_3^{T,2}+\cos\gamma D_3^{T,2}
		\big)d\bar{z}'\,,	
	\end{cases}
\end{equation}
where, for $i,j\in \{0,1\}$, there are
\begin{align*}
	\boldsymbol{D}^{B,2}
	=&\cos^{-1}\gamma\tilde{\nabla}_{\delta^0}p^{B,2}
	+2(\cos\gamma\delta)^{-1}p^{B,2}\nabla\delta
	-\cos^2\gamma\tilde{\Delta}_{\delta^{-1}}\boldsymbol u^{B,1}
	-\cos^2\gamma\tilde{\Delta}_{\delta^{0}}\boldsymbol{u}^{B,0}
	\\
	&
	+\sqrt{\tfrac{\cos\gamma}{\nu}}
	\Big(
	\partial_t\boldsymbol u^{B,1}
	+\tfrac{3}{2}\tan\gamma\nabla_{ h}\gamma\cdot
	\big(
	\boldsymbol{u}_h^{B,0} \boldsymbol{u}^{I,1}+\boldsymbol{u}^{0}_h \boldsymbol{u}^{B,1}
	\big)\Big)
	\\
	&+\sqrt{\tfrac{\cos\gamma}{\nu}}
	\sum_{i+j=1}\big(
	\boldsymbol{u}^{i}\cdot\tilde{\nabla}_{\delta^0}\boldsymbol{u}^{B,j}
	+\boldsymbol{u}^{1+i}\cdot\nabla(\delta\tilde{z})\partial_{\tilde{z}}\boldsymbol{u}^{B,j}+(\boldsymbol{u}^{B,i}\cdot\nabla)\boldsymbol{u}^{I,j}\big)\,,
	\\
	D^{B,2}_0=
	&\delta^{-1}\tilde{\nabla}_{\delta^0}\cdot(\delta\boldsymbol u^{B,1})\,.
\end{align*}

In summary, the approximate solution constructed above satisfies only
\begin{equation}\label{2.34}
	\sum_{i=0}^{2}{\delta}^i
	\boldsymbol u^{i}|_{\partial\Omega}=0\,.
\end{equation}

\begin{remark}
	In fact, in analyzing the approximate equations \eqref{2.9}, some terms involve the derivatives of $\chi(z)$. However, due to the setting \eqref{1.17} of $\chi(z)$, we know that the derivatives of $\chi(z)$ are zero near the boundary layer. Therefore, these terms are neither discussed in the boundary parts nor analyzed in the interior parts because of their construction. The details will be amended in the next subsection.
\end{remark}

\subsection{Modification of incompressible condition}\label{subsec2.3}
This subsection aims to recover the divergence-free condition, which the previously constructed solutions do not satisfy, namely,
\begin{align}\label{2.35}
	\nabla\cdot \Big(\sum_{i=0}^{2}{\delta}^i
	\boldsymbol u^{i}\Big)
	=&\chi'\sum_{i=0}^{2}{\delta}^i
	(u_3^{T,i}-u_3^{B,i})
	+\nabla\cdot(\delta\boldsymbol{u}^{I,1})
	+
	\nabla\cdot(\delta^2\boldsymbol{u}^{I,2})
	\notag\\
	&+
	(1-\chi)\tilde{\nabla}_{\delta^0}\cdot(\delta^2\boldsymbol u^{B,2})
	+\chi\bar{\nabla}_{\delta^0}\cdot(\delta^2\boldsymbol u_h^{T,2})
	\,.
\end{align}
Note that the correction term needs to satisfy the boundary conditions.
In addition, for the part of the equation involving the derivative of $\chi(z)$ and the velocity field correction term, we will correct it with the pressure term.

Firstly, we correct the incompressibility condition according to the order of $\delta$, one by one. By observing the order of \eqref{2.35}, take $\boldsymbol{u}^c=\sum_{i=0}^2\delta^i\boldsymbol{u}^{c,i}$, and let it be satisfied
\begin{align}
	&\begin{cases}
		\nabla \cdot \boldsymbol u^{c,0}=\chi'({u}_3^{B,0}-{u}_3^{T,0})\,,\\
		\boldsymbol u^{c,0}|_{\partial\Omega}=0\,,
	\end{cases}\label{2.36}\\
	&\begin{cases}
		\nabla \cdot (\delta\boldsymbol u^{c,1})=\chi'\delta({u}_3^{B,1}-{u}_3^{T,1})-\nabla\cdot(\delta\boldsymbol{u}^{I,1})\,,\\
		\delta\boldsymbol u^{c,1}|_{\partial\Omega}=0\,,
	\end{cases}\label{2.37}\\
	&		\begin{cases}
		\begin{aligned}
			\nabla \cdot (\delta^2\boldsymbol u^{c,2})=
			&\chi'\delta^2({u}_3^{B,2}-{u}_3^{T,2})
			-\nabla\cdot(\delta^2\boldsymbol{u}^{I,2})
			\\
			&-
			(1-\chi)\tilde{\nabla}_{\delta^0}\cdot(\delta^2\boldsymbol u^{B,2})
			-\chi\bar{\nabla}_{\delta^0}\cdot(\delta^2\boldsymbol u_h^{T,2})\,,
		\end{aligned}\\
		\delta^2\boldsymbol u^{c,2}|_{\partial\Omega}=0\,.
	\end{cases}\label{2.38}
\end{align}
For the correction term $\boldsymbol{u}^c$, we have:	

\begin{proposition}\label{propC}	
	Assuming ${\boldsymbol{u}}^{I,0}_{h}\in L^\infty(\mathbb{R}^2)\cap H^s(\mathbb{R}^2)~(s\geqslant5)$ and $B(x,y)\in C^\infty(\mathbb{R}^2)$, 
	there exist $\boldsymbol{u}^{c,0},\delta\boldsymbol{u}^{c,1},\delta^2\boldsymbol{u}^{c,2}\in W^{1,\infty}(\mathbb{R}_+;H^1(\Omega))$ satisfying \eqref{2.36}-\eqref{2.38}
	with the following estimation holds:
	\begin{align*} 
		\Big\lVert 
		\sum_{i=0}^{2}\delta^i\boldsymbol u^{c,i}
		(t)
		\Big\rVert_{W^{1,\infty}(\mathbb{R}_+;H^1(\Omega))}
		\lesssim&
		\lVert
		\boldsymbol{u}^{I,0}_h
		\rVert_{{L}^{\infty}(\mathbb{R}^2)}
		\lVert
		\boldsymbol{u}^{I,0}_h
		\rVert_{{H}^{1}(\mathbb{R}^2)}+ 
		\lVert
		\boldsymbol{u}^{I,0}_h
		\rVert_{{L}^{2}(\mathbb{R}^2)}\\
		&+\varepsilon
		(\lVert
		\boldsymbol{u}^{I,0}_h
		\rVert_{{L}^{\infty}(\mathbb{R}^2)}
		\lVert
		\boldsymbol{u}^{I,0}_h
		\rVert_{{H}^{2}(\mathbb{R}^2)}+ 
		\lVert
		\boldsymbol{u}^{I,0}_h
		\rVert_{{H}^{1}(\mathbb{R}^2)})
		\notag\\
		&
		+\varepsilon^2
		(\lVert
		\boldsymbol{u}^{I,0}_h
		\rVert_{{L}^{\infty}(\mathbb{R}^2)}
		\lVert
		\boldsymbol{u}^{I,0}_h
		\rVert_{{H}^{5}(\mathbb{R}^2)}+ 
		\lVert
		\boldsymbol{u}^{I,0}_h
		\rVert_{{H}^{4}(\mathbb{R}^2)})\,,
	\end{align*}
	where $t\in [0,T)$.
\end{proposition}
\begin{proof}
	According to Refs.~\cite{Galdi1994,Ukai1986},  it is sufficient to show that the right end terms of \eqref{2.36}-\eqref{2.38} belong to the $W^{1,\infty}(\mathbb{R}_+;L^2(\Omega))$ space.
	
	The following is an example of $\boldsymbol{u}^{c,0}$. The remaining terms can be estimated similarly.
	\begin{align*}
		&\big\lVert\chi'({u}_3^{B,0}-{u}_3^{T,0})\big\rVert_{W^{1,\infty}(\mathbb{R}_+;L^2(\Omega))}\\
		&\qquad\leq\big\lVert
		\mathrm{e}^{-\tfrac{\tilde z}{\sqrt 2}}
		\nabla_hB^T\big(
		\cos(\tfrac{\tilde z}{\sqrt 2})\boldsymbol{E}+\sin(\tfrac{\tilde z}{\sqrt 2})\cos\gamma\boldsymbol{E_1}\boldsymbol{H_0}
		\big)\partial_t{\boldsymbol{u}}^{I,0}_h
		\big\rVert_{L^2(\Omega)}\notag\\
		&\qquad\quad+\big\lVert
		\mathrm{e}^{-\tfrac{\bar z}{\sqrt 2}}
		\nabla_hB^T\big(
		\cos(\tfrac{\bar z}{\sqrt 2})\boldsymbol{E}+\sin(\tfrac{\bar z}{\sqrt 2})\cos\gamma\boldsymbol{E_1}\boldsymbol{H_0}
		\big)\partial_t{\boldsymbol{u}}_h^{I,0}
		\big\rVert_{L^2(\Omega)}\notag\\
		&\qquad\leq
		(\lVert\nabla_{ h}B\rVert_{{L}^{\infty}(\mathbb{R}^2)}
		+\lVert\nabla^\bot_{ h}B\rVert_{{L}^{\infty}(\mathbb{R}^2)})
		\lVert
		\partial_t{\boldsymbol{u}}^{I,0}_h
		\rVert_{{L}^{2}(\mathbb{R}^2)}\notag\\
		&\qquad\lesssim
		\lVert
		\boldsymbol{u}^{I,0}_h
		\rVert_{{L}^{\infty}(\mathbb{R}^2)}
		\lVert
		\boldsymbol{u}^{I,0}_h
		\rVert_{{H}^{1}(\mathbb{R}^2)}+ 
		\lVert
		\boldsymbol{u}^{I,0}_h
		\rVert_{{L}^{2}(\mathbb{R}^2)}
		\,,\notag
	\end{align*}
	the last inequalitie is obtained from the equations \eqref{2.21} for $\boldsymbol{u}^{I,0}_h$.
\end{proof}

Secondly, due to the appearance of the correction term and the derivative term of $\chi(z)$, the residual term of the approximate equation cannot achieve the expected result. Then, the construction is completed by correcting the pressure term $\nabla{p}^c=\sum_{i=0}^2\nabla(\delta^i{p}^{c,i})\in H_0^1(\Omega)$, and for $i,j\in \{0,1,2\}$ that it satisfies 
\begin{align}
	\varepsilon^{-1}\nabla p^{c,0}=&
	-\varepsilon^{-1}R \boldsymbol u^{c,0}+\nu\varepsilon \Delta  \boldsymbol u^{c,0}
	-\varepsilon^{-1}\chi'(0,0,p^{T,0}-p^{B,0})^T\notag\\
	&-\delta^{-1}\boldsymbol{u}^{c,0}\cdot\nabla(z-B)\big(
	(1-\chi)\partial_{\tilde{z}}\boldsymbol{u}^{B,0}-\chi\partial_{\bar{z}}\boldsymbol{u}^{T,0}\big)\,,\label{2.39}\\
	\varepsilon^{-1}\nabla (\delta p^{c,1})=
	&-\partial_t\boldsymbol{u}^{c,0}-\sqrt{\nu}\cos^{-\frac{3}{2}}\gamma \boldsymbol{R} \boldsymbol u^{c,1}+\nu\varepsilon \Delta  (\delta\boldsymbol u^{c,1})-\boldsymbol{u}_h^{c,0}\cdot\nabla_h {\boldsymbol{u}}^{I,0}\notag\\
	&-\big({\boldsymbol{u}}^{0}+\boldsymbol{u}^{c,0}\big)\cdot\nabla\boldsymbol{u}^{c,0}+\chi'\sqrt{\nu}\cos^{\frac32}\gamma
	\big(
	\partial_{\bar{z}}\boldsymbol{u}^{T,0}
	-\partial_{\tilde{z}}\boldsymbol{u}^{B,0}
	\big)
	\notag\\
	&
	-\chi'{\sqrt{\nu}}{\cos^{-\frac32}\gamma}(0,0,p^{T,1}-p^{B,1})^T
	-\chi'\big({u}^{0}_3+{u}_3^{c,0}\big)(\boldsymbol{u}^{T,0}-\boldsymbol{u}^{B,0})
	\notag\\
	&
	-(1-\chi)\boldsymbol{u}^{c,0}\cdot\tilde{\nabla}_{\delta^0}\boldsymbol{u}^{B,0}
	-\chi\boldsymbol{u}^{c,0}\cdot\bar{\nabla}_{\delta^0}\boldsymbol{u}^{T,0}
	\notag\\
	&
	-\sum_{i+j=1}\boldsymbol{u}^{c,i}
	\cdot\nabla(z-B)\big(
	(1-\chi)\partial_{\tilde{z}}\boldsymbol{u}^{B,j}-\chi\partial_{\bar{z}}\boldsymbol{u}^{T,j}\big)\,,\label{2.40}\\
	\varepsilon^{-1}\nabla (\delta^2 p^{c,2})=&
	-\partial_t(\delta\boldsymbol{u}^{c,1})-\nu\varepsilon\cos^{-3}\gamma \boldsymbol{R} \boldsymbol u^{c,2}
	-\chi'{\nu\varepsilon}{\cos^{-3}\gamma}(0,0,p^{T,2}-p^{B,2})^T\notag\\
	&+\nu\varepsilon \Delta  (\delta^2\boldsymbol u^{c,2})
	+\chi'{\nu}\varepsilon
	\big(
	\partial_{\bar{z}}\boldsymbol{u}^{T,1}
	-\partial_{\tilde{z}}\boldsymbol{u}^{B,1}
	\big)-\chi''\nu\varepsilon(\boldsymbol{u}^{B,0}-\boldsymbol{u}^{T,0})
	\notag\\
	&-\sum_{i+j=1}
	\Big(
	\delta^i\boldsymbol{u}^{c,i}\cdot\nabla (\delta^j\boldsymbol{u}^{I,j})
	+\delta^i\big(\boldsymbol{u}^{i}+\boldsymbol{u}^{c,i}\big)\cdot\nabla(\delta^j\boldsymbol{u}^{c,j})\notag\\
	&\qquad\qquad+\chi'\delta\big({u}_3^{i}+{u}_3^{c,i}\big)(\boldsymbol{u}^{T,j}-\boldsymbol{u}^{B,j})+\chi(\delta^i\boldsymbol{u}_h^{c,i})\cdot\bar{\nabla}_{\delta^0}(\delta^j\boldsymbol{u}^{T,j})
	\notag\\
	&\qquad\qquad+(1-\chi)(\delta^i\boldsymbol{u}_h^{c,i})\cdot\tilde{\nabla}_{\delta^0}(\delta^j\boldsymbol{u}^{B,j})
	\Big)
	\notag\\
	&+\sum_{i+j=2}\delta\boldsymbol{u}^{c,i}
	\cdot\nabla(z-B)\big(
	(1-\chi)\partial_{\tilde{z}}\boldsymbol{u}^{B,j}-\chi\partial_{\bar{z}}\boldsymbol{u}^{T,j}\big)\,.\label{2.41}
\end{align}

The above completes the corrections concerning the incompressibility condition and the lower-order terms in the equations.

\subsection{Estimates of the approximate solution \eqref{2.42}}

Through the analysis in the previous subsections, we construct the approximate solution $(\boldsymbol{u}^\varepsilon_{app},p^\varepsilon_{app})$:
\begin{equation}\label{2.42}
	\begin{cases}
		\boldsymbol{u}_{app}^{\varepsilon}=
		\sum\limits_{i=0}^{2}
		{\delta}^i\big(
		\boldsymbol{u}^{I,i}
		+
		(1-\chi)\boldsymbol u^{B,i}
		+\chi\boldsymbol u^{T,i}
		+\boldsymbol u^{c,i}\big)\,,\\
		p_{app}^{\varepsilon}=
		\sum\limits_{i=0}^{2}
		\Big(
		{\delta}^{i+1}\big(
		(1-\chi)p^{B,i+1}+\chi p^{T,i+1}
		\big)
		+\delta^i(p^{I,i}+p^{c,i})\Big)\,,
	\end{cases}
\end{equation}
where main part $(\boldsymbol{u}^{I,0},\delta p^{I,1})$ satisfies \eqref{2.21}-\eqref{2.22}, and the residual internal terms are expressed by equations \eqref{2.14}, \eqref{2.20} and \eqref{2.28}-\eqref{2.31}.
The boundary terms are denoted by equations \eqref{2.11}-\eqref{2.12},\eqref{2.15}-\eqref{2.18}, \eqref{2.24}-\eqref{2.27} and \eqref{2.32}-\eqref{2.33}, as well as the correction items are shown by equations \eqref{2.36}-\eqref{2.41}.

The approximate solution is mainly based on the term  $\boldsymbol{u}^{I,0}$. To go ahead, we write $(\boldsymbol{u}^{I,0},\delta p^{I,1})=(\bar{\boldsymbol{u}},\varepsilon\bar{p})$, which satisfy the 
following system:
\begin{equation}\label{2.43}
	\begin{cases}
		\partial_{t} \bar {\boldsymbol u}_h+(\bar {\boldsymbol u}_h \cdot \nabla_h) \bar {\boldsymbol u}_h
		+\sqrt{\frac{\nu}{2}}
		(\boldsymbol{H_0}
		-\boldsymbol{E_1})
		\bar{\boldsymbol u}_h+\nabla_h\bar p=0\,,\\
		\bar u_3=\nabla_h{B}\cdot\bar{\boldsymbol{u}}_h\,,\\
		\nabla\cdot \bar{\boldsymbol u}=\nabla_h \cdot \bar{\boldsymbol u}_h=0\,.
	\end{cases}
\end{equation}

By applying $\nabla_h^\bot \cdot $ to $\eqref{2.43}_{1}$, and denoting $\bar{\omega}=\nabla_h^\bot \cdot \bar {\boldsymbol u}_h$, we have 
\begin{equation}\label{2.44}
	\partial_t\bar{\omega}+(\bar {\boldsymbol u}_h\cdot\nabla_h) {\bar{\omega}}
	+\nabla^\bot_h\cdot\big(\sqrt{\tfrac{\nu}{2}}
	\boldsymbol{H_0}
	\bar{\boldsymbol u}_h
	-\sqrt{\tfrac{\nu}{2}}\boldsymbol{E_1}\bar{\boldsymbol u}_h\big)
	=0\,,
\end{equation}
which yields
\begin{align}\label{2.45}
	&\partial_t\bar{\omega}+(\bar {\boldsymbol u}_h\cdot\nabla_h) {\bar{\omega}}
	+\sqrt{\tfrac{\nu}{2}}\bar{\omega}\notag\\
	&\qquad=-\sqrt{\tfrac{\nu}{2}}
	\nabla_hB^T\boldsymbol{E_1}
	\boldsymbol{H}
	\bar{\boldsymbol{u}}_h
	+\sqrt{\tfrac{\nu}{2}}
	\big(
\nabla_hB\cdot(\nabla_h^\bot B\cdot\nabla_h)\bar{\boldsymbol{u}}_h
	\big)\,.
\end{align}

First, we need to re-establish the estimated limit equations for $\bar{\boldsymbol{u}}_h$, $\nabla_{ h}\bar{\boldsymbol{u}}_h$ and $\bar{{\omega}}$ based on simplifying the limiting system.
\begin{proposition}\label{lem1}
	Let $ {\boldsymbol u}_{0,h}(x,y) \in{H}^{3}(\mathbb{R}^2) $ be a divergence-free vector field,   $(\bar{\boldsymbol u},\bar p)$ be a pair of solutions to the system \eqref{2.43} with initial data $ {\boldsymbol u}_0=( {\boldsymbol u}_{0,h},0)$,
	and the vorticity $\bar{\omega}$ be a solution of \eqref{2.45} with initial value ${{\omega}_0=\nabla_h^{\bot}\cdot {\boldsymbol u}_{0,h}}$.
	Assume that the boundary surface $B(x,y)\in C^\infty(\mathbb{R}^2)$ satisfies 	
	\begin{equation}\label{2.46}
		\frac{\cos^2\alpha+\cos^2\beta}{\cos^2\gamma}<\tfrac{1}{8}\,,
		\quad  
		\sup\limits_{(x,y)\in \mathbb{R}^2} (1+\sqrt{\tfrac{2}{\nu}})\sqrt{\big| \boldsymbol{K}_{\mathrm{G}}\big|} <\tfrac{8}{9}\,,
	\end{equation}
where   $\boldsymbol{K}_{\mathrm{G}}$ is  the Gaussian curvature of $B(x,y)$.

	Then,  the following  estimates hold:
	\begin{align}
		&\lVert\bar {\boldsymbol u}\rVert_{L^2(\mathbb{R}^2)}^2
		+\lVert\bar{\omega}\rVert_{L^2(\mathbb{R}^2)}^2
		\lesssim
		\Big(\lVert{\boldsymbol{u}}_{0,h}\rVert_{L^2(\mathbb{R}^2)}^2+\lVert{\omega}_0\rVert_{L^2(\mathbb{R}^2)}^2\Big) \mathrm{e}^{-\frac{\sqrt{2\nu}}{8}  t}\,,\label{2.47}\\
		&\lVert\bar {\boldsymbol u}\rVert_{W^{1,\infty}(\mathbb{R}^2)}^2
		\lesssim
		\Big(\lVert{\boldsymbol{u}}_{0,h}\rVert_{L^2(\mathbb{R}^2)}^2+\lVert{\omega}_0\rVert_{H^2(\mathbb{R}^2)}^2\Big) \mathrm{e}^{-\frac{\sqrt{2\nu}}{8}  t}\label{2.48}\,.
	\end{align}
\end{proposition}	
\begin{proof}
		Firstly, given that $\bar{u}_3=\bar{\boldsymbol{u}}_h\cdot \nabla_hB$, it suffices to verify the results for $\bar{\boldsymbol{u}}_h$.

	Note that the eigenvalues of $\boldsymbol{H_0}$ are $1$ and $\cos^{-2}\gamma$, which implies the system \eqref{2.43} admits damping. 
	Furthermore, the eigenvalues of matrix $\boldsymbol{E_1}\boldsymbol{H}$ are denoted by $\pm\sqrt{-\det\boldsymbol{H}}$, with
	\begin{equation}\label{H}
		\mathcal{R}e\sqrt{-\det\boldsymbol{H}}<\sup\limits_{(x,y)\in \mathbb{R}^2}\tfrac98\sqrt{|\boldsymbol{K}_{\mathrm{G}}|},
	\end{equation}

	Due to the divergence-free condition of $\bar{\boldsymbol{u}}_h$ and antisymmetric construction of $\boldsymbol{E_1}$, We obtain the $L^2$ estimate of $\bar {\boldsymbol u}_h$ as 
	\begin{equation}\label{5.9}
		\frac12\frac{d}{dt}\lVert\bar {\boldsymbol u}_h\rVert_{L^2(\mathbb{R}^2)}^s
		+ \sqrt{\tfrac{\nu}{2}}
		\lVert\bar {\boldsymbol u}_h\rVert_{{L}^2(\mathbb{R}^2)}^2
		\leq0.
	\end{equation}
	Next, we utilize \eqref{2.45}, \eqref{2.46} and \eqref{H}  to estimate $\bar{\omega}$, there holds
	\begin{align}\label{5.12}
		&	\frac12\frac{d}{dt}\lVert\bar{\omega}\rVert_{L^2(\mathbb{R}^2)}^2+\sqrt{\tfrac{\nu}{2}}\lVert\bar{\omega}\rVert_{L^2(\mathbb{R}^2)}^2\notag
		\\\notag
		&\leq\sqrt{\tfrac{\nu}{2}} \int_{\mathbb{R}^2} |\nabla_hB^T\boldsymbol{E_1} \boldsymbol{H}
		 \bar{\boldsymbol{u}}_h| |\bar{\omega}| \,dx{d}y 
		+\sqrt{\tfrac{\nu}{2}}
		 \int_{\mathbb{R}^2}
		 |\nabla_hB|^2|\nabla_h\bar{\boldsymbol{u}}_h||\bar{\omega}| \,dx{d}y 
		\\\notag
		&< 
		\tfrac{9}{16}\sqrt{\tfrac{\nu}{2}}
		\sup\limits_{(x,y)\in \mathbb{R}^2}\sqrt{|\boldsymbol{K}_{\mathrm{G}}|}
		\lVert\bar{\boldsymbol{u}}_h\lVert_{L^2(\mathbb{R}^2)}
		\lVert\bar{\omega}\lVert_{L^2(\mathbb{R}^2)}
		+\tfrac18\sqrt{\tfrac{\nu}{2}}
		\lVert\bar{\omega}\lVert^{2}_{L^2(\mathbb{R}^2)}
		\\ 
		&\leq
		\big(
		\tfrac18
		+\tfrac{9}{16}\sup\limits_{(x,y)\in \mathbb{R}^2}\sqrt{|\boldsymbol{K}_{\mathrm{G}}|}
		\big)\sqrt{\tfrac{\nu}{2}}
		\lVert\bar{\omega}\lVert^{2}_{L^2(\mathbb{R}^2)}
		+
		\tfrac{9}{16}\sqrt{\tfrac{\nu}{2}}
		\sup\limits_{(x,y)\in \mathbb{R}^2}\sqrt{|\boldsymbol{K}_{\mathrm{G}}|}
		\lVert\bar{\boldsymbol{u}}_h\lVert^2_{L^2(\mathbb{R}^2)}.
	\end{align}
	Therefore, by combining \eqref{2.46} and \eqref{5.9}--\eqref{5.12}, we obtain
	\begin{displaymath}
		 \frac{d}{dt}
		\big(\lVert\bar {\boldsymbol u}_h\rVert_{L^2(\mathbb{R}^2)}^2
		+\lVert\bar{\omega}\rVert_{L^2(\mathbb{R}^2)}^2\big)
		+ \frac14\sqrt{\frac{\nu}{2}}
		\big(\lVert\bar {\boldsymbol u}_h\rVert_{L^2(\mathbb{R}^2)}^2
		+\lVert\bar{\omega}\rVert_{L^2(\mathbb{R}^2)}^2\big)
		\leq0,
	\end{displaymath}
	eventually leading to \eqref{2.47}.
The proof of result \eqref{2.48} follows a line of reasoning similar to that in \cite{Jia2024}, and thus the detailed steps are omitted here.	
\end{proof}

Since the convergence analysis requires higher regularity of the limiting states, we now provide the $H^s$-estimates $(s>5)$ for $\bar{\boldsymbol{u}}_h$.

\begin{proposition}\label{prop3}
	Let $ {\boldsymbol u}_{0,h}(x,y) \in{H}^{s}(\mathbb{R}^2)~(s>5)$ be a divergence-free vector field,  $(\bar{\boldsymbol u},\bar p)$ be a pair of solutions to the system \eqref{2.43} with initial data $ {\boldsymbol u}_0=( {\boldsymbol u}_{0,h},0)$.
	Assume that the boundary surface $B(x,y)\in C^\infty(\mathbb{R}^2)$.	
	Then, for $s>5$ and $t\in [0,+\infty)$, the following  estimates hold:			
	\begin{align}
		&\lVert\bar {\boldsymbol u}(t)\rVert_{H^s(\mathbb{R}^2)}
		\lesssim \lVert {\boldsymbol u}_{0,h}\rVert_{H^s(\mathbb{R}^2)}
		\,,\label{2.49}\\
		&\lVert\partial_t\bar {\boldsymbol u}(t)\rVert_{H^{s-1}(\mathbb{R}^2)}
		\lesssim \lVert
		 {\boldsymbol u}_{0,h}\rVert_{H^{s}(\mathbb{R}^2)}
		\label{2.50}\,.
	\end{align}
\end{proposition}
\begin{proof}
	In the first place, applying $\varLambda^s$ to \eqref{2.43} and multiplying the resulting equations by $\varLambda^s\bar{\boldsymbol{u}}_h$, then we have
	\begin{multline}\label{2.51}
		\frac12\frac{d}{dt}\lVert\varLambda^s\bar {\boldsymbol u}_h\rVert_{L^2(\mathbb{R}^2)}^2
		+\int_{\mathbb{R}^2}			
		\varLambda^s\big(
		\sqrt{\tfrac{\nu}{2}}
		(\boldsymbol{H_0}
		-\boldsymbol{E_1})
		\bar{\boldsymbol u}_h
		\big)
		\cdot
		\varLambda^s
		\bar {\boldsymbol u}_h
		dxdy\\
		+\int_{\mathbb{R}^2}
		\varLambda^s\big(
		(\bar {\boldsymbol u}_h \cdot \nabla_h) \bar {\boldsymbol u}_h
		\big)
		\cdot
		\varLambda^s
		\bar {\boldsymbol u}_h
		dxdy=0\,.
	\end{multline}

	We simplify the second term on the left-hand side of equation \eqref{2.51}, yielding
	\begin{align}\label{2.52}
		\varLambda^s\big(
		\sqrt{\tfrac{\nu}{2}}
		(\boldsymbol{H_0}
		-\boldsymbol{E_1})
		\bar{\boldsymbol u}_h
		\big)
		=	&	\sqrt{\tfrac{\nu}{2}}
		(\boldsymbol{H_0}
		-\boldsymbol{E_1})\varLambda^s\bar{\boldsymbol u}_h+\big[
		\varLambda^s,\sqrt{\tfrac{\nu}{2}}
		\boldsymbol{H_0}
		\big]	\bar{\boldsymbol u}_h\,,
	\end{align}
	where $[,]$ is the commutator.
	Additionally, it is shown in Ref.~\cite{Kato1988} that for $p\in(1,\infty)$ satisfying 
	\begin{equation}\label{2.53}
		\lVert
		\big[
		\varLambda^s,f
		\big]g
		\rVert_{L^p}
		\lesssim
		\lVert\nabla f\rVert_{L^\infty}
		\lVert
		\varLambda^{s-1}g
		\rVert_{L^p}
		+
		\lVert
		\varLambda^{s}f
		\cdot g\rVert_{L^p}\,.
	\end{equation}
	Then the second term at the right end of \eqref{2.52} can be written as
	\begin{align}\label{2.54}
			\lVert\big[
		\varLambda^s,\sqrt{\tfrac{\nu}{2}}
		\boldsymbol{H_0}
		\big]	\bar{\boldsymbol u}_h\rVert_{L^2(\mathbb{R}^2)}
		&\lesssim	
		\lVert\nabla \big(\sqrt{\tfrac{\nu}{2}}
		\boldsymbol{H_0}
		\big)\rVert_{L^\infty(\mathbb{R}^2)}
		\lVert
		\varLambda^{s-1}\bar{\boldsymbol u}_h
		\rVert_{L^2(\mathbb{R}^2)}\notag\\
		&\qquad\quad+
		\lVert\varLambda^s \big(\sqrt{\tfrac{\nu}{2}}
		\boldsymbol{H_0}
		\big)\rVert_{L^\infty(\mathbb{R}^2)}
		\lVert
		\bar{\boldsymbol u}_h
		\rVert_{L^2(\mathbb{R}^2)}
		\,.
	\end{align}

	Combining the incompressible condition with \eqref{2.53}, the third term of \eqref{2.51} reads
	\begin{align}\label{2.55}
		&\int_{\mathbb{R}^2}
		\varLambda^s\big(
		(\bar {\boldsymbol u}_h \cdot \nabla_h) \bar {\boldsymbol u}_h
		\big)
		\cdot
		\varLambda^s
		\bar {\boldsymbol u}_h
		dxdy\notag\\	
		&\qquad=	
		\int_{\mathbb{R}^2}
		\Big(
		\varLambda^s\big(
		(\bar {\boldsymbol u}_h \cdot \nabla_h) \bar {\boldsymbol u}_h
		\big)
		-
		\big(
		(\bar {\boldsymbol u}_h \cdot \nabla_h) \varLambda^s\bar {\boldsymbol u}_h
		\big)
		\Big)
		\cdot
		\varLambda^s
		\bar {\boldsymbol u}_h
		dxdy\notag\\	
		&\qquad\leq
		\lVert
		\varLambda^s\big(
		(\bar {\boldsymbol u}_h \cdot \nabla_h) \bar {\boldsymbol u}_h
		\big)
		-
		\big(
		(\bar {\boldsymbol u}_h \cdot \nabla_h) \varLambda^s\bar {\boldsymbol u}_h
		\big)
		\rVert_{L^2(\mathbb{R}^2)}
		\lVert\varLambda^s\bar {\boldsymbol u}_h\rVert_{L^2(\mathbb{R}^2)}\notag\\
		&\qquad\lesssim
		\lVert\nabla_h  \bar {\boldsymbol u}_h\rVert_{L^\infty(\mathbb{R}^2)}
		\lVert\varLambda^s\bar {\boldsymbol u}_h\rVert_{L^2(\mathbb{R}^2)}^2
		\,.
	\end{align}
	Based on  $B(x,y)\in C^\infty(\mathbb{R}^2)$ , the eigenvalues of $\boldsymbol{H_0}$ are $1$ and $\cos^{-2}\gamma$, $\boldsymbol{E}_1$ is skew-symmetry, and equations \eqref{2.52}-\eqref{2.55}, we can derive equation \eqref{2.51} as follows
	\begin{align}\label{2.56}
		&\frac12\frac{d}{dt}\lVert\varLambda^s\bar {\boldsymbol u}_h\rVert_{L^2(\mathbb{R}^2)}^2			
		+
		\sqrt{\tfrac{\nu}{2}}
		\lVert\varLambda^s\bar {\boldsymbol u}_h\rVert_{L^2(\mathbb{R}^2)}^2
		\notag\\
		&\qquad\lesssim
		\lVert
		\bar{\boldsymbol u}_h
		\rVert_{L^2(\mathbb{R}^2)}
		\lVert\varLambda^s\bar{\boldsymbol u}_h\rVert_{L^2(\mathbb{R}^2)}
		+
		\lVert\varLambda^{s-1}\bar{\boldsymbol u}_h\rVert_{L^2(\mathbb{R}^2)}
		\lVert\varLambda^s\bar{\boldsymbol u}_h\rVert_{L^2(\mathbb{R}^2)}\notag\\
		&\qquad\quad+\lVert\nabla_h  \bar {\boldsymbol u}_h\rVert_{L^\infty(\mathbb{R}^2)}
		\lVert\varLambda^s\bar{\boldsymbol u}_h\rVert_{L^2(\mathbb{R}^2)}^2
		\notag\\
		&\qquad	\lesssim
		\tfrac12\sqrt{\tfrac{\nu}{2}}
		\lVert\varLambda^s\bar {\boldsymbol u}_h\rVert_{L^2(\mathbb{R}^2)}^2
		+
		\lVert
		\bar{\boldsymbol u}_h
		\rVert_{L^2(\mathbb{R}^2)}^2+\lVert\nabla_h  \bar {\boldsymbol u}_h\rVert_{L^\infty(\mathbb{R}^2)}
		\lVert\varLambda^s\bar{\boldsymbol u}_h\rVert_{L^2(\mathbb{R}^2)}^2\,.
	\end{align}

	Combining the result of Proposition \ref{lem1} and  \eqref{2.56}, it is easy to obtain 
	\begin{align}\label{2.57}
		\frac{d}{dt}\lVert\varLambda^s\bar {\boldsymbol u}_h\rVert_{L^2(\mathbb{R}^2)}^2	
		\lesssim&
		\lVert\varLambda^s\bar{\boldsymbol u}_h\rVert_{L^2(\mathbb{R}^2)}^2
		\Big(\lVert{\boldsymbol{u}}_{0,h}\rVert_{L^2(\mathbb{R}^2)}+\lVert{\omega}_0\rVert_{H^2(\mathbb{R}^2)}\Big) \mathrm{e}^{-\frac{\sqrt{2\nu}}{8}  t}\notag\\
		&+\lVert{\boldsymbol{u}}_{0,h}\rVert_{L^2(\mathbb{R}^2)}^2 \mathrm{e}^{-\frac{\sqrt{2\nu}}{8}  t}
		\,.
	\end{align}	
	Using the Gronwall inequality for \eqref{2.57} yields
	\begin{align}\label{2.58}
		\lVert\varLambda^s {\boldsymbol u}_h(t)\rVert_{L^2(\mathbb{R}^2)}^2	
		\lesssim&
		\big(
		\lVert\varLambda^s {\boldsymbol u}_{h}(0)\rVert_{L^2(\mathbb{R}^2)}^2	
		+
		\tfrac{8}{\sqrt{2\nu}}
		\lVert{\boldsymbol{u}}_{0,h}\rVert_{L^2(\mathbb{R}^2)}^2
		\big)\notag\\
		&\times
		\exp
		\big(
		\tfrac{8}{\sqrt{2\nu}}(\lVert{\boldsymbol{u}}_{0,h}\rVert_{L^2(\mathbb{R}^2)}+\lVert{\omega}_0\rVert_{H^2(\mathbb{R}^2)})
		\big)\,.
	\end{align}

	To estimate $\lVert\partial_t \bar{\boldsymbol{u}}_h\rVert_{H^{s-1}(\mathbb{R}^2)}$, we first recall the equations \eqref{2.43} of $\bar{\boldsymbol{u}}_h$:
	\begin{equation}\label{2.59}
		\partial_{t} \bar{\boldsymbol{u}}_h
		=-(\bar {\boldsymbol u}_h \cdot \nabla_h) \bar {\boldsymbol u}_h
		-\sqrt{\tfrac{\nu}{2}}
		(\boldsymbol{H_0}
		-\boldsymbol{E_1})
		\bar{\boldsymbol u}_h-\nabla_h\bar p\,,
	\end{equation}
	and combine it with the incompressibility condition, yielding 
	\begin{equation*}
		\nabla_h\bar p=-\nabla_h\Delta_h^{-1}\nabla_h\cdot\Big((\bar {\boldsymbol u}_h \cdot \nabla_h) \bar {\boldsymbol u}_h+\sqrt{\tfrac{\nu}{2}}
		(\boldsymbol{H_0}
		-\boldsymbol{E_1})
		\bar{\boldsymbol u}_h\Big)\,.
	\end{equation*}
	Thus combining $\nabla_h \cdot \bar{\boldsymbol{u}}_h$ and acting on the Leray projection  $\mathbb{P}$ for \eqref{2.59} yields
	\begin{equation}\label{2.60}
		\partial_{t} \bar{\boldsymbol{u}}_h
		=-\mathbb{P}\big(
		(\bar {\boldsymbol u}_h \cdot \nabla_h) \bar {\boldsymbol u}_h+\sqrt{\tfrac{\nu}{2}}
		(\boldsymbol{H_0}
		-\boldsymbol{E_1})
		\bar{\boldsymbol u}_h
		\big)\,.
	\end{equation}
	Further utilizing $B(x,y)\in C^\infty(\mathbb{R}^2)$ and \eqref{2.49}, we can derive
	\begin{align*}
		\lVert\partial_t\bar{\boldsymbol{u}}_h\rVert_{{H}^{s-1}(\mathbb{R}^2)}
		=
		&\big\lVert
		\mathbb{P}\big(
		(\bar {\boldsymbol u}_h \cdot \nabla_h) \bar {\boldsymbol u}_h+\sqrt{\tfrac{\nu}{2}}
		(\boldsymbol{H_0}
		-\boldsymbol{E_1})
		\bar{\boldsymbol u}_h
		\big)
		\big\rVert_{{H}^{s-1}(\mathbb{R}^2)}\notag\\
		\lesssim&
		\lVert(\bar {\boldsymbol u}_h \cdot \nabla_h) \bar {\boldsymbol u}_h\rVert_{{H}^{s-1}(\mathbb{R}^2)}
		+
		\lVert\bar {\boldsymbol u}_h\rVert_{{H}^{s-1}(\mathbb{R}^2)}
		\notag\\
		\lesssim&\lVert\bar{\boldsymbol{u}}_h\rVert_{{L}^\infty(\mathbb{R}^2)}
		\lVert\bar {\boldsymbol u}_h\rVert_{{H}^{s}(\mathbb{R}^2)}
		+
		(1+\lVert\nabla_{ h}\bar{\boldsymbol{u}}_h\rVert_{{L}^\infty(\mathbb{R}^2)})
		\lVert\bar {\boldsymbol u}_h\rVert_{{H}^{s-1}(\mathbb{R}^2)}\notag\\
		\lesssim &\lVert
		 {\boldsymbol u}_{0,h}\rVert_{H^{s}(\mathbb{R}^2)}
		+\lVert{\boldsymbol{u}}_{0,h}\rVert_{L^2(\mathbb{R}^2)}\,.
		\notag
	\end{align*}	
\end{proof}
%
%
%

As the subsequent proofs require, we give the propositions related to $\bar{u}_3$ below.

\begin{proposition}\label{prop4}
	Let ${\boldsymbol{u}}_{0}$ and $B(x,y)$ be defined as in Proposition \ref{lem1}-\ref{prop3}, $(\bar{\boldsymbol u},\bar p)$ be a pair of solutions to the system \eqref{2.43} with initial data $ {\boldsymbol u}_0$. Then, for $t \in [0,+\infty)$, there holds
	\begin{align}
		\lVert
		\partial_t \bar{u}_3+(\bar{\boldsymbol{u}}_h\cdot\nabla_h)\bar{u}_3\rVert_{{L}^{2}(\mathbb{R}^2)}&\leq
		\lVert\tilde{\boldsymbol{u}}_h\rVert_{L^2(\mathbb{R}^2)}
		+\lVert\tilde{\omega}\rVert_{L^2(\mathbb{R}^2)},\label{n13}\\
		\lVert
		\partial^2_t \bar{u}_3+\partial_t\big((\bar{\boldsymbol{u}}_h\cdot\nabla_h)\bar{u}_3\big)\rVert_{{L}^{2}(\mathbb{R}^2)} &\leq
		\lVert\tilde{\boldsymbol{u}}_h\rVert_{L^2(\mathbb{R}^2)}
		+\lVert\tilde{\omega}\rVert_{H^1(\mathbb{R}^2)}
		,\label{n14}
	\end{align}
	where 
	\begin{equation}\label{tildeu}
		\tilde{\boldsymbol{u}}_h=\bar {\boldsymbol u}_h-\mathrm{e}^{-\sqrt{\frac\nu2}t}{\boldsymbol{u}}_{0,h}, \quad \text{and} \quad
		\tilde{\omega}=\bar{\omega}-\mathrm{e}^{-\sqrt{\frac\nu2}t}\omega_0.
	\end{equation}
\end{proposition}
\begin{proof}
	From the limiting system \eqref{2.43}, we have 
	\begin{align}\label{n1}
		&\partial_t \bar{u}_3+(\bar{\boldsymbol{u}}_h\cdot\nabla_h)\bar{u}_3\notag\\
		=&-\nabla_{h}B\cdot\Big(
		(\bar {\boldsymbol u}_h \cdot \nabla_h) \bar {\boldsymbol u}_h+\sqrt{\tfrac{\nu}{2\cos\gamma}}
		(\boldsymbol{H_0}
		-\cos^{-1}\gamma\boldsymbol{E_1})
		\bar{\boldsymbol u}_h +\nabla_{ h}\bar{p}
		\Big)+(\bar{\boldsymbol{u}}_h\cdot\nabla_h)\bar{u}_3\notag\\
		=&-\nabla_{h}B\cdot\Big(
		(\bar {\boldsymbol u}_h \cdot \nabla_h) \tilde{\boldsymbol{u}}_h+\big(\tilde{\boldsymbol{u}}_h\cdot  \nabla_h\big)\big(\mathrm{e}^{-\sqrt{\frac\nu2}t}{\boldsymbol{u}}_{0,h}\big)+\sqrt{\tfrac{\nu}{2}}
		(\boldsymbol{H_0}
		-\boldsymbol{E_1})
		\tilde{\boldsymbol{u}}_h\notag\\
		&\qquad\qquad
		+\nabla_{ h}\Big(\bar{p}-\mathrm{e}^{-2\sqrt{\frac\nu2}t}p_0+\sqrt{\tfrac\nu2}\mathrm{e}^{-\sqrt{\frac\nu2}t}\phi\Big)
		\Big)\notag\\
		&+(\bar{\boldsymbol{u}}_h\cdot\nabla_h)\big(\nabla_{ h}B\cdot\tilde{\boldsymbol{u}}_h\big).
	\end{align}
	The final equality is derived from the compatibility condition \eqref{compatibility},  with
	\begin{equation}
		\phi=\Delta^{-1}_h\nabla_{h}\cdot(\boldsymbol{E_1}\boldsymbol{u}_{0,h}).
	\end{equation}
 Also, combining the limiting system and the compatibility condition yields
	\begin{align}\label{n2}
		&\nabla_{ h}\Big(\bar{p}-\mathrm{e}^{-2\sqrt{\frac\nu2}t}p_0+\sqrt{\tfrac{\nu}{2}}\mathrm{e}^{-\sqrt{\frac\nu2}t}\phi\Big)\notag\\
		=&-\nabla_{ h} \Delta_h^{-1}\nabla_{ h}\cdot
		\Big(
		(\bar {\boldsymbol u}_h \cdot \nabla_h) \tilde{\boldsymbol{u}}_h+\big(\tilde{\boldsymbol{u}}_h\cdot  \nabla_h\big)\big(\mathrm{e}^{-\sqrt{\frac\nu2}t}{\boldsymbol{u}}_{0,h}\big)+\sqrt{\tfrac{\nu}{2}}
		(\boldsymbol{H_0}
		-\boldsymbol{E_1})\tilde{\boldsymbol{u}}_h
		\Big).
	\end{align}
	Then, from  Propositions \ref{lem1}--\ref{prop3}, it follows that
	\begin{align}\label{n3}
		&\lVert
		\partial_t \bar{u}_3+(\bar{\boldsymbol{u}}_h\cdot\nabla_h)\bar{u}_3\rVert_{{L}^{2}(\mathbb{R}^2)}\notag\\
		\leq&\lVert\nabla_{h}B\cdot
		\mathbb{P}\Big(	(\bar {\boldsymbol u}_h \cdot \nabla_h) \tilde{\boldsymbol{u}}_h+\big(\tilde{\boldsymbol{u}}_h\cdot  \nabla_h\big)\big(\mathrm{e}^{-\sqrt{\frac\nu2}t}{\boldsymbol{u}}_{0,h}\big)+\sqrt{\tfrac{\nu}{2}}
		(\boldsymbol{H_0}
		-\boldsymbol{E_1})\tilde{\boldsymbol{u}}_h
		\Big)\rVert_{{L}^{2}(\mathbb{R}^2)}\notag\\
		&+\lVert(\bar{\boldsymbol{u}}_h\cdot\nabla_h)\big(\nabla_{ h}B\cdot\tilde{\boldsymbol{u}}_h\big)\rVert_{{L}^{2}(\mathbb{R}^2)},\notag\\
		\lesssim&
		\lVert
		\tilde{\boldsymbol{u}}_h
		\rVert_{{L}^{2}(\mathbb{R}^2)}
		+
		\lVert\nabla_{ h}
		\tilde{\boldsymbol{u}}_h
		\rVert_{{L}^{2}(\mathbb{R}^2)}.
	\end{align}
	
	Similarly we can expand to get
	\begin{align}\label{n15}
	&\lVert
	\partial^2_t \bar{u}_3+	\partial_t((\bar{\boldsymbol{u}}_h\cdot\nabla_h)\bar{u}_3)\rVert_{{L}^{2}(\mathbb{R}^2)}\notag\\
	\leq&\lVert\nabla_{h}B\cdot
		\partial_t\mathbb{P}\Big(	(\bar {\boldsymbol u}_h \cdot \nabla_h) \tilde{\boldsymbol{u}}_h+\big(\tilde{\boldsymbol{u}}_h\cdot  \nabla_h\big)\big(\mathrm{e}^{-\sqrt{\frac\nu2}t}{\boldsymbol{u}}_{0,h}\big)+\sqrt{\tfrac{\nu}{2}}
	(\boldsymbol{H_0}
	-\boldsymbol{E_1})\tilde{\boldsymbol{u}}_h
	\Big)\rVert_{{L}^{2}(\mathbb{R}^2)}\notag\\
	&+\lVert	\partial_t\big((\bar{\boldsymbol{u}}_h\cdot\nabla_h)\big(\nabla_{ h}B\cdot\tilde{\boldsymbol{u}}_h\big)\big)\rVert_{{L}^{2}(\mathbb{R}^2)},\notag\\
	\leq&
	\lVert\nabla_{h}B\cdot
	 \mathbb{P}\Big(	
	(\partial_t\bar {\boldsymbol u}_h \cdot \nabla_h) \tilde{\boldsymbol{u}}_h
	+(\bar {\boldsymbol u}_h \cdot \nabla_h) \partial_t\tilde{\boldsymbol{u}}_h
	+\sqrt{\tfrac{\nu}{2}}
	(\boldsymbol{H_0}
	-\boldsymbol{E_1})\partial_t\tilde{\boldsymbol{u}}_h
	\Big)\rVert_{{L}^{2}(\mathbb{R}^2)}\notag\\
		&+
	\lVert\nabla_{h}B\cdot
	\mathbb{P}\Big(	
	\big(\partial_t\tilde{\boldsymbol{u}}_h\cdot  \nabla_h\big)\big(\mathrm{e}^{-\sqrt{\frac\nu2}t}{\boldsymbol{u}}_{0,h}\big)
	-\sqrt{\tfrac\nu2}
	\big(\tilde{\boldsymbol{u}}_h\cdot  \nabla_h\big)\big(\mathrm{e}^{-\sqrt{\frac\nu2}t}{\boldsymbol{u}}_{0,h}\big)
	\Big)\rVert_{{L}^{2}(\mathbb{R}^2)}\notag\\
	&+
	\lVert	\partial_t \bar{\boldsymbol{u}}^T_h\boldsymbol{H}\tilde{\boldsymbol{u}}_h
	+ \bar{\boldsymbol{u}}^T_h\boldsymbol{H}\partial_t\tilde{\boldsymbol{u}}_h
	+
	\nabla_h B\cdot (\partial_t\bar{\boldsymbol{u}}_h\cdot\nabla_h)\tilde{\boldsymbol{u}}_h+
	\nabla_h B\cdot (\bar{\boldsymbol{u}}_h\cdot\nabla_h)\partial_t\tilde{\boldsymbol{u}}_h
	\rVert_{{L}^{2}(\mathbb{R}^2)},\notag\\
	\lesssim&
	\lVert
	\tilde{\boldsymbol{u}}_h
	\rVert_{{L}^{2}(\mathbb{R}^2)}
	+
	\lVert\nabla_{ h}
	\tilde{\boldsymbol{u}}_h
	\rVert_{{L}^{2}(\mathbb{R}^2)}
	+\lVert
	\partial_t\tilde{\boldsymbol{u}}_h
	\rVert_{{L}^{2}(\mathbb{R}^2)}
	+
	\lVert\nabla_{ h}
	\partial_t\tilde{\boldsymbol{u}}_h
	\rVert_{{L}^{2}(\mathbb{R}^2)}.
\end{align}

	We recall that for incompressible fluids, 
	\begin{equation*}
		\nabla_{h}
		\tilde{\boldsymbol{u}}_h=\nabla_{h}(\bar {\boldsymbol u}_h-\mathrm{e}^{-\sqrt{\frac\nu2}t}{\boldsymbol{u}}_{0,h})
		=\mathcal{R} (\bar{\omega}-\mathrm{e}^{-\sqrt{\frac\nu2}t}\Delta_h \phi)=:\mathcal{R}\tilde{\omega},
	\end{equation*}
	 where $\mathcal{R}$ is a homogeneous Calder\'{o}n-Zygmund operator of order zero. Thus, we get 
	\begin{equation*}\label{n4}
		\lVert
		\partial_t \bar{u}_3+(\bar{\boldsymbol{u}}_h\cdot\nabla_h)\bar{u}_3\rVert_{{L}^{2}(\mathbb{R}^2)}
		\lesssim
		\lVert
			\tilde{\boldsymbol{u}}_h
		\rVert_{{L}^{2}(\mathbb{R}^2)}
		+
		\lVert \tilde{\omega}
		\rVert_{{L}^{2}(\mathbb{R}^2)},
	\end{equation*}
and	
	\begin{equation*}\label{n16}
	\lVert
	\partial^2_t \bar{u}_3+\partial_t((\bar{\boldsymbol{u}}_h\cdot\nabla_h)\bar{u}_3)\rVert_{{L}^{2}(\mathbb{R}^2)}
	\lesssim
	\lVert
	\tilde{\boldsymbol{u}}_h
	\rVert_{{L}^{2}(\mathbb{R}^2)}
	+
	\lVert \tilde{\omega}
	\rVert_{{L}^{2}(\mathbb{R}^2)}
	+\lVert
	\partial_t\tilde{\boldsymbol{u}}_h
	\rVert_{{L}^{2}(\mathbb{R}^2)}
	+
	\lVert 
	\partial_t\tilde{\omega}
	\rVert_{{L}^{2}(\mathbb{R}^2)}.
\end{equation*}

Next, we use equations \eqref{2.43} and \eqref{compatibility} to give the form of $\partial_t\tilde{\boldsymbol{u}}_h$ and give estimates via Proposition \ref{lem1}:	
\begin{align}\label{n17}
	\lVert 
	\partial_t\tilde{\boldsymbol{u}}_h
	\rVert_{{L}^{2}(\mathbb{R}^2)}
	\leq&
	\lVert
	\mathbb{P}\Big((\bar {\boldsymbol u}_h \cdot \nabla_h) \tilde{\boldsymbol{u}}_h
	+\big(\tilde{\boldsymbol{u}}_h\cdot\nabla_{ h}\big)\big(\mathrm{e}^{-\sqrt{\frac\nu2}t}{\boldsymbol{u}}_{0,h}\big)
	+\sqrt{\tfrac{\nu}{2}}
	(\boldsymbol{H_0}
	-\boldsymbol{E_1})
	\tilde{\boldsymbol{u}}_h\Big)
	\rVert_{{L}^{2}(\mathbb{R}^2)}\notag\\
	\lesssim&
	\lVert\tilde{\boldsymbol{u}}_h\rVert_{L^2(\mathbb{R}^2)}
	+\lVert\tilde{\omega}\rVert_{L^2(\mathbb{R}^2)}.
\end{align}

	Finally, we use equation \eqref{2.44} to give the expression for $\partial_t\tilde{\omega}$ and derive estimates via Proposition \ref{lem1}--\ref{prop3}:	
	\begin{align}\label{n18}
		\lVert 
		\partial_{t}\tilde{\omega}
		\rVert_{{L}^{2}(\mathbb{R}^2)}
		\leq&
		\lVert
		(\bar {\boldsymbol u}_h \cdot \nabla_h) \tilde{\omega}
		+\big(\tilde{\boldsymbol{u}}_h\cdot\nabla_{ h}\big)\big(\mathrm{e}^{-\sqrt{\frac\nu2}t}{\omega}_0\big)
		+\nabla_h^\bot\cdot
		\Big(
		\sqrt{\tfrac{\nu}{2}}
		(\boldsymbol{H_0}
		-\boldsymbol{E_1})
		\tilde{\boldsymbol{u}}_h\Big)
		\rVert_{{L}^{2}(\mathbb{R}^2)}\notag\\
		\lesssim&
		\lVert\tilde{\boldsymbol{u}}_h\rVert_{L^2(\mathbb{R}^2)}
		+\lVert\tilde{\omega}\rVert_{H^1(\mathbb{R}^2)}.
	\end{align}

Based on the above analysis, we finally obtained the results \eqref{n13}-\eqref{n14}.
	\end{proof}

\begin{proposition}\label{propT}
	Let ${\boldsymbol{u}}_{0}$, $B(x,y)$, $\tilde{\boldsymbol{u}}_h$ and $\tilde{\omega}$ be defined as in Proposition \ref{lem1}-\ref{prop3} and \eqref{tildeu}. Then, for $t \in [0,+\infty)$, there holds
	\begin{align}
		\lVert
		\tilde{\boldsymbol{u}}_h(t)\rVert_{{L}^{2}(\mathbb{R}^2)}&\leq
		\lVert
		\tilde{\boldsymbol{u}}_h(0)\rVert_{{L}^{2}(\mathbb{R}^2)},\label{n22}\\
		\lVert\tilde{\omega}(t)\rVert_{H^1(\mathbb{R}^2)} &\leq
	\lVert\tilde{\omega}(0)\rVert_{H^1(\mathbb{R}^2)}.\label{n23}
	\end{align}
\end{proposition}
\begin{proof}
	First, we estimate $\lVert
	\tilde{\boldsymbol{u}}_h
	\rVert_{{L}^{2}(\mathbb{R}^2)}$. Combining with \eqref{compatibility} and \eqref{2.43}, we have
	\begin{equation}\label{n5}
		\begin{cases}
			\partial_{t}\tilde{\boldsymbol{u}}_h+(\bar {\boldsymbol u}_h \cdot \nabla_h) \tilde{\boldsymbol{u}}_h
			+\big(\tilde{\boldsymbol{u}}_h\cdot\nabla_{ h}\big)\big(\mathrm{e}^{-\sqrt{\frac\nu2}t}{\boldsymbol{u}}_{0,h}\big)
			+\sqrt{\tfrac{\nu}{2}}
			(\boldsymbol{H_0}
			-\boldsymbol{E_1})
			\tilde{\boldsymbol{u}}_h
			+\nabla_h \tilde{p} =0\,,\\
			\nabla\cdot \tilde{\boldsymbol{u}}_h=0,
		\end{cases}
	\end{equation}
	where $\tilde{p}=\bar{p}-\mathrm{e}^{-2\sqrt{\frac\nu2}t}p_0+\sqrt{\tfrac{\nu}{2}}\mathrm{e}^{-\sqrt{\frac\nu2}t}\phi$.
	Due to the well-prepared initial data $\nabla_{ h}B\cdot{\boldsymbol{u}}_{0,h}=0$ and \eqref{H}, one gets
	\begin{align}\label{n6}
		\frac{1}{2}\frac{\mathrm{d}}{\mathrm{d}t}\|\tilde{\boldsymbol{u}}_h\|^2_{{L}^2(\mathbb{R}^2)}
		+\sqrt{\tfrac\nu2}\|\tilde{\boldsymbol{u}}_h\|^2_{{L}^2(\mathbb{R}^2)}
		\leq
		\int_{\mathbb{R}^2}
		|\tilde{\boldsymbol{u}}_h^T
		\boldsymbol{E_1}\boldsymbol{H}\tilde{\boldsymbol{u}}_h|
		{d}x{d}y
		\leq\tfrac98\sqrt{|\boldsymbol{K}_{\mathrm{G}}|}\|\tilde{\boldsymbol{u}}_h\|^2_{{L}^2(\mathbb{R}^2)}\,.
	\end{align}
	Under the curvature constraint \eqref{2.46}, one obtains
	\begin{equation}\label{n8}
		\|\tilde{\boldsymbol{u}}_h(t)\|_{{L}^2(\mathbb{R}^2)}\leq
		\|\tilde{\boldsymbol{u}}_h(0)\|_{{L}^2(\mathbb{R}^2)}.
	\end{equation}
	
	Second, we estimate $\lVert \tilde{\omega}
	\rVert_{{L}^{2}(\mathbb{R}^2)}$.
	Applying the $\nabla_{ h}^\bot\cdot$ operator to \eqref{n5} yields
	\begin{equation}\label{n9}
		\partial_{t}\tilde{\omega}+(\bar {\boldsymbol u}_h \cdot \nabla_h) \tilde{\omega}
		+\big(\tilde{\boldsymbol{u}}_h\cdot\nabla_{ h}\big)\big(\mathrm{e}^{-\sqrt{\frac\nu2}t}{\omega}_0\big)
		+\nabla_h^\bot\cdot
		\Big(
		\sqrt{\tfrac{\nu}{2}}
		(\boldsymbol{H_0}
		-\boldsymbol{E_1})
		\tilde{\boldsymbol{u}}_h\Big)=0,
	\end{equation}
		Then, we utilize \eqref{2.44}, \eqref{2.45} and \eqref{n9}  to estimate $\tilde{\omega}$, there holds
	\begin{align*} 
		&	\frac12\frac{d}{dt}\lVert\tilde{\omega}\rVert_{L^2(\mathbb{R}^2)}^2+\sqrt{\tfrac{\nu}{2}}\lVert
		\tilde{\omega}
		\rVert_{L^s(\mathbb{R}^2)}^s
		\notag\\\notag
		&\leq\sqrt{\tfrac{\nu}{2}} \int_{\mathbb{R}^2} |\nabla_hB^T\boldsymbol{E_1} \boldsymbol{H}
		\tilde{\boldsymbol{u}}_h| |\tilde{\omega}| \,dx{d}y 
		+\sqrt{\tfrac{\nu}{2}}
		\int_{\mathbb{R}^2}
		|\nabla_hB|^2|\nabla_h\tilde{\boldsymbol{u}}_h||\tilde{\omega}| \,dx{d}y \\
		&\quad+\int_{\mathbb{R}^2}
		|\big(\tilde{\boldsymbol{u}}_h\cdot\nabla_{ h}\big)\big(\mathrm{e}^{-\sqrt{\frac\nu2}t}{\omega}_0\big)\tilde{\omega}|
		\,dx{d}y\notag
		\\
		&\leq 
		\big(
		\lVert\nabla_{h}{\omega}_0\lVert_{L^\infty(\mathbb{R}^2)}
		+\tfrac{9}{16}\sqrt{\tfrac{\nu}{2}}
		\sup\limits_{(x,y)\in \mathbb{R}^2}\sqrt{|\boldsymbol{K}_{\mathrm{G}}|}\big)
		\lVert\tilde{\boldsymbol{u}}_h\lVert_{L^2(\mathbb{R}^2)}
		\lVert\tilde{\omega}\lVert_{L^2(\mathbb{R}^2)}
		+\tfrac18\sqrt{\tfrac{\nu}{2}}
		\lVert\tilde{\omega}\lVert^{2}_{L^2(\mathbb{R}^2)}.
	\end{align*}
	Combining \eqref{2.46}, \eqref{n8} and Proposition \ref{prop3} yields
	\begin{equation}\label{n12}
		\lVert\tilde{\omega}(t)\rVert_{L^2(\mathbb{R}^2)}\leq
		\lVert\tilde{\omega}(0)\rVert_{L^2(\mathbb{R}^2)}
		.
	\end{equation}

	For the estimate of $\lVert\nabla_{ h}\tilde{\omega}\rVert_{L^2(\mathbb{R}^2)}$, we have below that after acting the $\nabla_{ h}$ operator on $\eqref{n9}_1$, the inner product with $\nabla_{h}\tilde{\omega}$ yields
	\begin{align}\label{n20}
		&	\frac12\frac{d}{dt}\lVert\nabla_{h}\tilde{\omega}\rVert_{L^2(\mathbb{R}^2)}^2+\sqrt{\tfrac{\nu}{2}}\lVert\nabla_{h}
		\tilde{\omega}
		\rVert_{L^s(\mathbb{R}^2)}^s
		\notag\\\notag
		&=
		-\int_{\mathbb{R}^2}
		\nabla_{ h}\big(
		(\bar {\boldsymbol u}_h \cdot \nabla_h) \tilde{\omega}
		\big)\cdot
		\nabla_{ h}\tilde{\omega}
		\,dx{d}y
		-\int_{\mathbb{R}^2}
		\nabla_{ h}\big(
		\big(\tilde{\boldsymbol{u}}_h\cdot\nabla_{ h}\big)\big(\mathrm{e}^{-\sqrt{\frac\nu2}t}{\omega}_0\big)
		\big)\cdot
		\nabla_{ h}\tilde{\omega}
		\,dx{d}y\notag\\
		&\quad-\sqrt{\tfrac{\nu}{2}}\int_{\mathbb{R}^2}
		\nabla_{ h}\big(
		\nabla_hB^T\boldsymbol{E_1} \boldsymbol{H}
		\tilde{\boldsymbol{u}}_h
		-\nabla_hB\cdot(\nabla_h^\bot B\cdot\nabla_h)\tilde{\boldsymbol{u}}_h
		\big)\cdot
		\nabla_{ h}\tilde{\omega}
		\,dx{d}y
		\notag\\
		&\lesssim
		\lVert
		\nabla_{ h}\bar{\boldsymbol{u}}_h
		\lVert_{L^\infty(\mathbb{R}^2)}
		\lVert
		\nabla_{ h}\tilde{\omega}
		\lVert^{2}_{L^2(\mathbb{R}^2)}
		+
		\Big(\lVert
		\tilde{\boldsymbol{u}}_h
		\lVert_{L^2(\mathbb{R}^2)}
		+\lVert
		\tilde{\omega}
		\lVert_{L^2(\mathbb{R}^2)}\Big)
		\lVert
		\nabla_{h}\tilde{\omega}
		\lVert_{L^2(\mathbb{R}^2)}\notag\\
		&\leq\lVert
		\nabla_{ h}\bar{\boldsymbol{u}}_h
		\lVert_{L^\infty(\mathbb{R}^2)}
		\lVert
		\nabla_{ h}\tilde{\omega}
		\lVert^{2}_{L^2(\mathbb{R}^2)}.
	\end{align}
	Using \eqref{1.26}, we get
	\begin{equation}\label{n21}
		\lVert\nabla_{ h}\tilde{\omega}(t)\rVert_{L^2(\mathbb{R}^2)}\leq
		\lVert\nabla_{ h}\tilde{\omega}(0)\rVert_{L^2(\mathbb{R}^2)}
		.
	\end{equation}
\end{proof}

The above is a study of the main part of the approximate solution.
Next, we formally present the approximate solution equation and show the remainder terms. The approximate solution satisfies
\begin{equation}\label{2.61}
	\begin{cases}
		\partial_{t} {\boldsymbol u}_{app}^{\varepsilon}-\nu\varepsilon\Delta{\boldsymbol u}_{app}^{\varepsilon}+({\boldsymbol u}_{app}^{\varepsilon} \cdot \nabla) {\boldsymbol u}_{app}^{\varepsilon}+\varepsilon^{-1}\boldsymbol{R} {\boldsymbol u}_{app}^{\varepsilon}
		+\varepsilon^{-1}\nabla{p}_{app}^{\varepsilon}=\boldsymbol \rho^{\varepsilon},\\
		\nabla\cdot {\boldsymbol u}_{app}^{\varepsilon}=0,\quad
		{\boldsymbol u}_{app}^{\varepsilon}|_{\partial\Omega}=0.
	\end{cases}
\end{equation}
where $\boldsymbol{\rho}^\varepsilon=\boldsymbol \rho^{\varepsilon,1}+\boldsymbol \rho^{\varepsilon,2}+\boldsymbol \rho^{\varepsilon,3}$, and for $i,j\in\{0,1,2\}$, there are
\begin{align}
	\boldsymbol{\rho}^{\varepsilon,1}=
	&
	\frac{1}{\varepsilon}\Big(
	(1-\chi)\tilde{\nabla}_{\delta^0}
	(\delta^3p^{B,3})
	+\chi\bar{\nabla}_{\delta^0}(\delta^3p^{T,3})
	+\chi'\delta^3(0,0,p^{T,3}-p^{B,3})^T\Big)
	\notag\\
	&-\nu\varepsilon\sum_{i=1}^2
	\Big(
	(1-\chi)
	\tilde{\Delta}_{\delta^{0}}(\delta^i\boldsymbol{u}^{B,i})
	+\chi
	\bar{\Delta}_{\delta^0}(\delta^i\boldsymbol{u}^{T,i})
	+\chi''\delta^i(\boldsymbol{u}^{T,i}-\boldsymbol{u}^{B,i})
	\Big)
	\notag\\
	&-(1-\chi)\nu\varepsilon\delta\tilde{\Delta}_{\delta^{-1}}\boldsymbol u^{B,2}
	-\chi\nu\varepsilon\delta\bar{\Delta}_{\delta^{-1}}\boldsymbol u^{T,2}
	-\chi'\nu\varepsilon\delta
	\big(
	\partial_{\bar{z}}\boldsymbol{u}^{T,2}
	-\partial_{\tilde{z}}\boldsymbol{u}^{B,2}
	\big)\,,\label{rho1}\\
	\boldsymbol{\rho}^{\varepsilon,2}=
	&
	(0,0,\partial_t \bar{u}_3+\bar{\boldsymbol{u}}_h\cdot\nabla_{h}\bar{u}_3)^T
	+\delta^2\partial_t\big(\boldsymbol{u}^{2}+\boldsymbol{u}^{c,2}\big) -\nu\varepsilon\sum\limits_{i=1}^2
	\Delta(\delta^i\boldsymbol{u}^{I,i})\,,
	\label{rho2}\\
	\boldsymbol{\rho}^{\varepsilon,3}=
	&\sum_{2\leq i+j\leq4} 
	\big(
	\delta^i(\boldsymbol{u}^{i}+\boldsymbol{u}^{c,i})\cdot\nabla (\delta^j\boldsymbol{u}^{I,j}+\delta^j\boldsymbol{u}^{c,j})
	+\chi'\delta^{i+j}
	({u}_3^{i}+{u}_3^{c,i})
	(\boldsymbol{u}^{T,j}-\boldsymbol{u}^{B,j})
	\notag\\
	&\qquad\qquad
	+\delta^{i}
	(\boldsymbol{u}^{i}+\boldsymbol{u}^{c,i})
	\cdot
	\big(
	\chi\bar{\nabla}_{\delta^{0}}(\delta^j\boldsymbol{u}^{T,j})
	+(1-\chi)
	\tilde{\nabla}_{\delta^{0}}(\delta^j\boldsymbol{u}^{B,j})
	\big)
	\big)
	\notag\\
	&+\sum\limits_{3\leq i+j\leq4}\delta^{i+j-1}
	(\boldsymbol{u}^{i}+\boldsymbol{u}^{c,i})
	\cdot\nabla(z-B)
	\big(
	(1-\chi)\partial_{\tilde{z}}\boldsymbol{u}^{B,j}-\chi\partial_{\bar{z}}\boldsymbol{u}^{T,j}
	\big)\,.
	\label{rho3}
\end{align}

In order to demonstrate the necessity of convergence, based on Propositions \ref{propC}-\ref{propT}, we give its estimations of $\boldsymbol{\rho}^\varepsilon$:
\begin{equation}\label{pho}
	\lVert
	\boldsymbol \rho^{\varepsilon}\rVert_{{L}^{2}(\Omega)}
	\lesssim
	 \varepsilon^2
	\lVert
	{\boldsymbol{u}}_{0,h}
	\rVert_{{H}^{5}(\mathbb{R}^2)}
	\big(
	\lVert
	{\boldsymbol{u}}_{0,h}
	\rVert_{{H}^{3}(\mathbb{R}^2)}
	+1\big)^2
	\,,
\end{equation}	
and
\begin{equation}\label{phot}
	\lVert
	\partial_t\boldsymbol \rho^{\varepsilon}\rVert_{{L}^{2}(\Omega)}
	\lesssim
	 \varepsilon^2
	\lVert
	{\boldsymbol{u}}_{0,h}
	\rVert_{{H}^{6}(\mathbb{R}^2)}
	\big(
	\lVert
	{\boldsymbol{u}}_{0,h}
	\rVert_{{H}^{4}(\mathbb{R}^2)}
	+1
	\big)^2
	\,.
\end{equation}	
	This completes the proof of Theorem \ref{thapp}.
\end{proof}

\section{$L^\infty$ Error estimation}\label{sec4}

In this section, 
we shall prove Theorem \ref{thm2}.
Writing $\boldsymbol{v}^\varepsilon=\boldsymbol{u}^\varepsilon-\boldsymbol{u}_{app}^\varepsilon$ and $p_v=p^\varepsilon-p_{app}^\varepsilon$ , we have
\begin{equation}\label{2.65}
	\begin{cases}
		\partial_{t} {\boldsymbol v}^{\varepsilon}-\nu\varepsilon\Delta{\boldsymbol v}^{\varepsilon}
		+{\boldsymbol u}^{\varepsilon} \cdot \nabla {\boldsymbol v}^{\varepsilon}
		+{\varepsilon}^{-1}\boldsymbol{R}\boldsymbol v^{\varepsilon}
		+{\varepsilon}^{-1}\nabla p_v 
		=-{\boldsymbol v}^{\varepsilon} \cdot \nabla\boldsymbol{u}^\varepsilon_{app}
		-\boldsymbol \rho^{\varepsilon}\,,\\
		\nabla\cdot {\boldsymbol v}^{\varepsilon}=0\,,\\
		{\boldsymbol v}^{\varepsilon}|_{t=0}={\boldsymbol v}^{\varepsilon}|_{\partial\Omega}=0\,.
	\end{cases}
\end{equation}
To prove  Theorem \ref{thm2}, it is sufficient to bound $\lVert \boldsymbol{v}^\varepsilon \rVert_{{L}^{2}(\Omega)}$, $\lVert \partial_t \boldsymbol{v}^\varepsilon \rVert_{{L}^{2}(\Omega)}$,  $\lVert \nabla\boldsymbol{v}^\varepsilon \rVert_{{L}^{2}(\Omega)}$ and $\lVert \Delta\boldsymbol{v}^\varepsilon \rVert_{{L}^{2}(\Omega)}$.
During this process,  we need to use the estimation of approximate solution $\boldsymbol{u}_{app}^\varepsilon$, so we first present some propositions about it.

\subsection{Estimations of approximate solution $\boldsymbol{u}_{app}^\varepsilon$}\label{subsec3.1}
In the subsequent estimation process, derivative estimates of the approximate solution are involved. As the derivatives of the internal and correction terms do not generate singularities, only the boundary terms do, we decompose the approximate solution $\boldsymbol{u}^\varepsilon_{app}$ as follows: 
\begin{equation}\label{3.1}
	\boldsymbol{u}^\varepsilon_{app}=\boldsymbol{u}^I+\boldsymbol{u}^{BL}+\boldsymbol{u}^c\,,
\end{equation}
where $\boldsymbol{u}^I=\sum_{i=0}^{2}
{\delta}^i\boldsymbol u^{I,i}$, $\boldsymbol{u}^{BL}=\sum_{i=0}^{2}
{\delta}^i\big(
(1-\chi)\boldsymbol u^{B,i}
+\chi\boldsymbol u^{T,i}\big)$, and $\boldsymbol{u}^c=\sum_{i=0}^{2}\delta^i\boldsymbol u^{c,i}$ represent the internal term, boundary term, and correction term, respectively.

For these three terms, we propose the following propositions to facilitate subsequent analysis.
First, we analyze the properties of $\boldsymbol{u}^I$ and $\boldsymbol{u}^c$.

\begin{proposition}\label{propIC}
	Assuming ${\boldsymbol{u}}_{0,h}\in H^s(\mathbb{R}^2)(s>5)$ and $B(x,y)\in C^\infty(\mathbb{R}^2)$, 
	for $t\in [0,+\infty)$,
	the internal terms $\boldsymbol{u}^I$ and correction terms $\boldsymbol{u}^c$ satisfy 
	\begin{align*} 
		\lVert \nabla\big(
		\boldsymbol{u}^{I} 
		+\boldsymbol{u}^{c}
		\big)(t)
		\rVert_{L^{\infty}(\Omega)}
		\lesssim
		&\Big(\lVert{\boldsymbol{u}}_{0,h}\rVert_{L^2(\mathbb{R}^2)}^2+\lVert{\omega}_0\rVert_{H^2(\mathbb{R}^2)}^2\Big) \mathrm{e}^{-\frac{\sqrt{2\nu}}{8}  t}
		\notag\\
		&+\varepsilon
		\lVert {\boldsymbol{u}}_{0,h} \rVert_{H^{4}(\mathbb{R}^2)}
		+\varepsilon^2
		\big(
		1+
		\lVert {\boldsymbol{u}}_{0,h} \rVert_{H^{4}(\mathbb{R}^2)}
		\big)
		\lVert {\boldsymbol{u}}_{0,h} \rVert_{H^{s}(\mathbb{R}^2)}\,,
	\end{align*}
	and
	\begin{align*} 
		&\lVert \partial_t
		\big(
		\boldsymbol{u}^{I} 
		+\boldsymbol{u}^{c}\big)(t) \rVert_{L^{\infty}(\Omega)}\\
		&\qquad\lesssim
		\lVert{\boldsymbol{u}}_{0,h}\rVert_{H^{3}(\mathbb{R}^2)}
		+\varepsilon
		\lVert{\boldsymbol{u}}_{0,h}\rVert_{H^{4}(\mathbb{R}^2)}
		+\varepsilon^2
		\big(
		1+
		\lVert {\boldsymbol{u}}_{0,h} \rVert_{H^{5}(\mathbb{R}^2)}
		\big)
		\lVert {\boldsymbol{u}}_{0,h} \rVert_{H^{s}(\mathbb{R}^2)}\,.
	\end{align*}
\end{proposition}
\begin{proof}
	First, consider the estimate of $\lVert 
	\boldsymbol{u}^{I}(t) 
	\rVert_{L^{\infty}(\Omega)}$. Based on the expressions  \eqref{2.13},\eqref{2.20}, \eqref{2.29}-\eqref{2.31} of the internal terms, for $t\in [0,+\infty)$,  the following results can be directly derived:
	\begin{align*}
		&\lVert \boldsymbol{u}^{I}(t) \rVert_{L^\infty(\Omega)}
		=\lVert \bar{\boldsymbol{u}} \rVert_{L^\infty(\Omega)}
		+\lVert \delta\boldsymbol{u}^{I,1} \rVert_{L^\infty(\Omega)}
		+\lVert \delta^2\boldsymbol{u}^{I,2} \rVert_{L^\infty(\Omega)}
		\\\notag
		&\qquad\lesssim	 
		\lVert \bar{\boldsymbol{u}}_h \rVert_{L^\infty(\mathbb{R}^2)}
		+\varepsilon\big(
		\lVert \bar{\boldsymbol{u}}_h \rVert_{L^\infty(\mathbb{R}^2)}
		+\lVert \nabla_h\bar{\boldsymbol{u}}_h \rVert_{L^\infty(\mathbb{R}^2)}
		\big)
		\\\notag
		&\qquad\quad+\varepsilon^2\big(
		\lVert \partial_t\bar{\boldsymbol{u}}_h \rVert_{W^{2,\infty}(\mathbb{R}^2)}
		+\lVert \bar{\boldsymbol{u}}_h \rVert_{W^{3,\infty}(\mathbb{R}^2)}
		+\lVert \bar{\boldsymbol{u}}_h \rVert_{W^{1,\infty}(\mathbb{R}^2)}
		\lVert \bar{\boldsymbol{u}}_h \rVert_{W^{3,\infty}(\mathbb{R}^2)}\big)
		\\\notag
		&\qquad\lesssim 
		\lVert \bar{\boldsymbol{u}}_h \rVert_{L^\infty(\mathbb{R}^2)}
		+\varepsilon\big(
		\lVert \bar{\boldsymbol{u}}_h \rVert_{L^\infty(\mathbb{R}^2)}
		+\lVert \nabla_h\bar{\boldsymbol{u}}_h \rVert_{L^\infty(\mathbb{R}^2)}
		\big)
		\\\notag
		&\qquad\quad+\varepsilon^2\big(
		\lVert \partial_t\bar{\boldsymbol{u}}_h \rVert_{H^{4}(\mathbb{R}^2)}
		+\lVert \bar{\boldsymbol{u}}_h \rVert_{H^{5}(\mathbb{R}^2)}
		+\lVert \bar{\boldsymbol{u}}_h \rVert_{H^{3}(\mathbb{R}^2)}
		\lVert \bar{\boldsymbol{u}}_h \rVert_{H^{5}(\mathbb{R}^2)}\big)
		\\\notag
		&\qquad\lesssim
		(1+\varepsilon)\big(
		\lVert{\boldsymbol{u}}_{0,h}\rVert_{L^2(\mathbb{R}^2)}+\lVert{\omega}_0\rVert_{H^2(\mathbb{R}^2)}
		\big)\mathrm{e}^{-\frac{\sqrt{2\nu}}{8}  t}\\
		&\qquad\quad+\varepsilon^2
		\big(
		1+
		\lVert {\boldsymbol{u}}_{0,h} \rVert_{H^{3}(\mathbb{R}^2)}
		\big)
		\lVert {\boldsymbol{u}}_{0,h} \rVert_{H^{5}(\mathbb{R}^2)}
		\,,\notag
	\end{align*}
	the last inequality is obtained from Propositions \ref{lem1}--\ref{prop3}.

	Based on the smoothness of boundary $B(x,y)$, and according to the expression of the corrective term, it is easy to see that $\lVert \nabla\boldsymbol{u}^{c}(t) \rVert_{L^\infty(\Omega)}$ is equivalent to $\lVert \bar{\boldsymbol{u}} \rVert_{L^\infty(\Omega)}
	+\lVert \delta\boldsymbol{u}^{I,1} \rVert_{L^\infty(\Omega)}
	+\lVert \delta^2\boldsymbol{u}^{I,2} \rVert_{L^\infty(\Omega)}
	+\lVert \nabla(\delta^2\boldsymbol{u}^{I,2}) \rVert_{L^\infty(\Omega)}$. Therefore, we can conclude that
	\begin{align*}
		&\lVert \nabla(\boldsymbol{u}^{I}+\boldsymbol{u}^{c})(t) \rVert_{L^\infty(\Omega)}
		=\lVert {\boldsymbol{u}}^I \rVert_{W^{1,\infty}(\Omega)}
		\\\notag
		&\qquad\lesssim	 
		\lVert \bar{\boldsymbol{u}}_h \rVert_{W^{1,\infty}(\Omega)}
		+\varepsilon\big(
		\lVert \bar{\boldsymbol{u}}_h \rVert_{W^{1,\infty}(\Omega)}
		+\lVert \nabla_h\bar{\boldsymbol{u}}_h \rVert_{W^{1,\infty}(\Omega)}
		\big)
		\\\notag
		&\qquad\quad+\varepsilon^2\big(
		\lVert \partial_t\bar{\boldsymbol{u}}_h \rVert_{W^{3,\infty}(\mathbb{R}^2)}
		+\lVert \bar{\boldsymbol{u}}_h \rVert_{W^{4,\infty}(\mathbb{R}^2)}
		+\lVert \bar{\boldsymbol{u}}_h \rVert_{W^{2,\infty}(\mathbb{R}^2)}
		\lVert \bar{\boldsymbol{u}}_h \rVert_{W^{4,\infty}(\mathbb{R}^2)}\big)
		\\\notag
		&\qquad\lesssim
		\big(
		\lVert{\boldsymbol{u}}_{0,h}\rVert_{L^2(\mathbb{R}^2)}+\lVert{\omega}_0\rVert_{H^2(\mathbb{R}^2)}
		\big)\mathrm{e}^{-\frac{\sqrt{2\nu}}{8}  t}
		\notag\\
		&\qquad\quad+\varepsilon
		\lVert {\boldsymbol{u}}_{0,h} \rVert_{H^{4}(\mathbb{R}^2)}
		+\varepsilon^2
		\big(
		1+
		\lVert {\boldsymbol{u}}_{0,h} \rVert_{H^{4}(\mathbb{R}^2)}
		\big)
		\lVert {\boldsymbol{u}}_{0,h} \rVert_{H^{s}(\mathbb{R}^2)}
		\,.\notag
	\end{align*}

	Finally,  we use the results from Proposition \ref{prop3} to complete the estimate for $\lVert \partial_t
	\big(
	\boldsymbol{u}^{I} 
	+\boldsymbol{u}^{c}\big)(t) \rVert_{L^{\infty}(\Omega)}$.
	\begin{align*}
		&\lVert \partial_t(\boldsymbol{u}^{I}+\boldsymbol{u}^{c})(t) \rVert_{L^\infty(\Omega)}\\
		&\qquad\lesssim\lVert \partial_t\bar{\boldsymbol{u}} \rVert_{L^\infty(\Omega)}
		+\lVert \delta\partial_t\boldsymbol{u}^{I,1} \rVert_{L^\infty(\Omega)}
		+\lVert \delta^2\partial_t\boldsymbol{u}^{I,2} \rVert_{L^{\infty}(\Omega)}
		\\\notag
		&\qquad\lesssim	 
		\lVert \partial_t\bar{\boldsymbol{u}}_h \rVert_{L^{\infty}(\Omega)}
		+\varepsilon
		\lVert \partial_t\bar{\boldsymbol{u}}_h \rVert_{W^{1,\infty}(\Omega)}
		\\\notag
		&\qquad\quad+\varepsilon^2\big(
		\lVert \partial^2_t\bar{\boldsymbol{u}}_h \rVert_{W^{2,\infty}(\mathbb{R}^2)}
		+\lVert \partial_t\bar{\boldsymbol{u}}_h \rVert_{W^{3,\infty}(\mathbb{R}^2)}
		(1+\lVert \partial_t\bar{\boldsymbol{u}}_h \rVert_{W^{1,\infty}(\mathbb{R}^2)})
		\big)
		\\\notag
		&\qquad\lesssim 
		\lVert \partial_t\bar{\boldsymbol{u}}_h \rVert_{H^{2}(\Omega)}
		+\varepsilon
		\lVert \partial_t\bar{\boldsymbol{u}}_h \rVert_{H^{3}(\Omega)}
		+\varepsilon^2
		\lVert \partial_t\bar{\boldsymbol{u}}_h \rVert_{H^{5}(\mathbb{R}^2)}
		(1+\lVert \partial_t\bar{\boldsymbol{u}}_h \rVert_{H^{3}(\mathbb{R}^2)})
		\\\notag
		&\qquad\quad+\varepsilon^2
		\lVert 
		\partial_t\mathbb{P}\big(
		(\bar {\boldsymbol u}_h \cdot \nabla_h) \bar {\boldsymbol u}_h+\sqrt{\tfrac{\nu}{2\cos\gamma}}
		(\boldsymbol{H_0}
		-\cos^{-1}\gamma\boldsymbol{E_1})
		\bar{\boldsymbol u}_h
		\big)
		\rVert_{H^{4}(\mathbb{R}^2)}
		\\\notag
		&\qquad\lesssim
		\lVert{\boldsymbol{u}}_{0,h}\rVert_{H^{3}(\mathbb{R}^2)}
		+\varepsilon
		\lVert{\boldsymbol{u}}_{0,h}\rVert_{H^{4}(\mathbb{R}^2)}
		+\varepsilon^2
		\big(
		1+
		\lVert {\boldsymbol{u}}_{0,h} \rVert_{H^{5}(\mathbb{R}^2)}
		\big)\lVert {\boldsymbol{u}}_{0,h} \rVert_{H^{s}(\mathbb{R}^2)}
		\,.\notag
	\end{align*}
\end{proof}	

Next, we analyze the properties of $\boldsymbol{u}^{BL}$. Due to the $\mathrm{e}^{-\frac{\tilde{z}}{\sqrt{2}}}$ or $\mathrm{e}^{-\frac{\bar{z}}{\sqrt{2}}}$ structure in the boundary layer terms, each differentiation incurs a loss of $\varepsilon$. To address this, in the subsequent proposition,  we will construct a compensating function $d(z)$, which denotes the distance to the boundary, to offset such loss.

\begin{proposition}\label{propBL}
	Assuming $\bar{\boldsymbol{u}}_h\in H^s(\mathbb{R}^2)(s>5)$ and $B(x,y)\in C^\infty(\mathbb{R}^2)$, for $k=0,1$, the boundary term $\boldsymbol{u}^{BL}$  satisfy
	\begin{align*} 
		&\lVert d(z)^{\frac{1}{2}} |\nabla^k\boldsymbol{u}^{BL}|(t) \rVert_{{L}^2([B,B+2];{L}^\infty(\mathbb{R}^2))}
		\\
		&\qquad\lesssim
		\tfrac{3\sqrt{3\nu}}{4}
		\varepsilon^{1-k}
		\big(
		\lVert{\boldsymbol{u}}_{0,h}\rVert_{L^2(\mathbb{R}^2)}+\lVert{\omega}_0\rVert_{H^2(\mathbb{R}^2)}	
		\big)\mathrm{e}^{-\frac{\sqrt{2\nu}}{8}  t}
		\notag\\
		&\qquad\quad+\tfrac{3\sqrt{3\nu}}{4}
		\varepsilon^{1-k}
		\big(
		\varepsilon
		\lVert {\boldsymbol{u}}_{0,h} \rVert_{H^{4}(\mathbb{R}^2)}+
		\varepsilon^2
		\big(
		1+
		\lVert {\boldsymbol{u}}_{0,h} \rVert_{H^{4}(\mathbb{R}^2)}
		\big)
		\lVert {\boldsymbol{u}}_{0,h} \rVert_{H^{s}(\mathbb{R}^2)}
		\big)\,,
	\end{align*}
	and
	\begin{align*} 
		\lVert \partial_t\boldsymbol{u}^{BL} \rVert_{L^\infty(\Omega)}
		\lesssim
		&\lVert{\boldsymbol{u}}_{0,h}\rVert_{H^{3}(\mathbb{R}^2)}
		+\varepsilon
		\lVert{\boldsymbol{u}}_{0,h}\rVert_{H^{4}(\mathbb{R}^2)}\\
		&+\varepsilon^2
		\big(
		1+
		\lVert {\boldsymbol{u}}_{0,h} \rVert_{H^{5}(\mathbb{R}^2)}
		\big)\lVert {\boldsymbol{u}}_{0,h} \rVert_{H^{s}(\mathbb{R}^2)}\,,
	\end{align*}	
	where $d(z)$ denotes the distance to the boundary.
\end{proposition}
\begin{proof}
	In the scenario where $k=0$, based on the boundary term expressions \eqref{2.11}-\eqref{2.12},\eqref{2.15}-\eqref{2.18}, \eqref{2.24}-\eqref{2.25} and \eqref{2.32}-\eqref{2.33}, 
	we can expand to obtain
	\begin{align}\label{3.2}
		&\lVert d(z)^{\frac{1}{2}} |\boldsymbol{u}^{BL}| \rVert_{{L}^2([B,B+2];{L}^\infty(\mathbb{R}^2))}\notag\\
		&\qquad=\sum\limits_{i=0}^2\lVert d(z)^{\frac{1}{2}} \delta^i|(1-\chi)\boldsymbol{u}^{B,i}+\chi\boldsymbol{u}^{T,i}| \rVert_{{L}^2([B,B+2];{L}^\infty(\mathbb{R}^2))}\,,
	\end{align}
	where the first term of \eqref{3.2} can be calculated to get
	\begin{align}
		&\lVert d(z)^{\frac{1}{2}} |(1-\chi)\boldsymbol{u}^{B,0}+\chi\boldsymbol{u}^{T,0}| \rVert_{{L}^2([B,B+2];{L}^\infty(\mathbb{R}^2))}\notag\\
		&\qquad\leq
		\tfrac{\delta}{\sqrt{2}}\lVert\bar{\boldsymbol{u}}_h\rVert_{{L}^\infty(\mathbb{R}^2)}
		\int_{[0,\frac{2\sqrt{2}}{\delta}]}
		\tfrac{\sqrt{2}(z-B)}{\delta}
		\mathrm{e}^{-\frac{\sqrt{2}(z-B)}{\delta}}
		\,{d}\tfrac{\sqrt{2}(z-B)}{\delta}\notag\\
		&\qquad\quad+\tfrac{\delta}{\sqrt{2}}\lVert\bar{\boldsymbol{u}}_h\rVert_{{L}^\infty(\mathbb{R}^2)}
		\int_{[0,\frac{2\sqrt{2}}{\delta}]}
		\tfrac{\sqrt{2}(2+B-z)}{\delta}
		\mathrm{e}^{-\frac{\sqrt{2}(2+B-z)}{\delta}}\,{d}\tfrac{\sqrt{2}(2+B-z)}{\delta}\notag\\
		&\qquad\lesssim\tfrac{3\sqrt{3\nu}}{4}\varepsilon\lVert\bar{\boldsymbol{u}}_h\rVert_{{L}^\infty(\mathbb{R}^2)}\,.\label{3.3}
	\end{align}	
	In the above estimate \eqref{3.3}, the introduction of $\varepsilon$ is based on the $\mathrm{e}^{-\frac{(z-B)}{\delta}}$ or $\mathrm{e}^{-\frac{(2+B-z)}{\delta}}$ structure in the boundary terms. Therefore, similarly treating the remaining terms in \eqref{3.2}, we can obtain
	\begin{align*}
		&\lVert d(z)^{\frac{1}{2}} |\boldsymbol{u}^{BL}| \rVert_{{L}^2([B,B+2];{L}^\infty(\mathbb{R}^2))}\notag\\
		&\qquad\lesssim\tfrac{3\sqrt{3\nu}}{4}\varepsilon\Big(
		\lVert{\boldsymbol{u}}_{0,h}\rVert_{{L}^\infty(\mathbb{R}^2)} \mathrm{e}^{-\frac{\sqrt{2\nu}}{8}  t}
		+\varepsilon\big(
		\lVert{\boldsymbol{u}}_{0,h}\rVert_{L^2(\mathbb{R}^2)}+\lVert{\omega}_0\rVert_{H^2(\mathbb{R}^2)}
		\big)\mathrm{e}^{-\frac{\sqrt{2\nu}}{8}  t}\Big)
		\notag\\
		&\qquad\quad+\tfrac{3\sqrt{3\nu}}{4}\varepsilon^3
		\big(
		1+
		\lVert {\boldsymbol{u}}_{0,h} \rVert_{H^{3}(\mathbb{R}^2)}
		\big)
		\lVert {\boldsymbol{u}}_{0,h} \rVert_{H^{5}(\mathbb{R}^2)}
		\,.
	\end{align*}
	For the case when $k=1$, we need to note that each differentiation of the boundary terms induces $\varepsilon^{-1}$.

	For term $\lVert \partial_t\boldsymbol{u}^{BL} \rVert_{L^\infty(\Omega)}$, we apply the estimate of $\partial_t\boldsymbol{u}^{I}$ from Proposition \ref{prop3} and obtain
	\begin{align*}
		\lVert \partial_t\boldsymbol{u}^{BL} \rVert_{L^\infty(\Omega)}
		\lesssim&
		\lVert{\boldsymbol{u}}_{0,h}\rVert_{H^{3}(\mathbb{R}^2)}
		+\varepsilon
		\lVert{\boldsymbol{u}}_{0,h}\rVert_{H^{4}(\mathbb{R}^2)}\\
		&+\varepsilon^2
		\big(
		1+
		\lVert {\boldsymbol{u}}_{0,h} \rVert_{H^{5}(\mathbb{R}^2)}
		\big)\lVert {\boldsymbol{u}}_{0,h} \rVert_{H^{s}(\mathbb{R}^2)}.
	\end{align*}
\end{proof}

\subsection{$L^2$ estimation for $\boldsymbol{v}^\varepsilon$}\label{sub4.1}

The following energy estimate for the system \eqref{2.65} can be obtained by a standard energy method.
\begin{proposition}\label{prop5}
	Under the assumptions of Theorem \ref{thm2}, then there exist some $\varepsilon>0$ such that
	\begin{equation}\label{4.1}
		\lVert\boldsymbol v^{\varepsilon}\rVert^2_{L^\infty([0,T),{L}^2(\Omega))}
		+\nu\varepsilon\lVert\nabla\boldsymbol v^{\varepsilon}\rVert^2_{L^2([0,T),{L}^2(\Omega))}
		\lesssim\varepsilon^4\,.
	\end{equation} 
\end{proposition}
\begin{proof}
	Fristly, computing formally the ${L}^2$ scalar product of \eqref{2.65} by $\boldsymbol v^{\varepsilon}$ leads to
	\begin{align}\label{4.3}
		\frac{1}{2}\frac{d}{dt}\lVert\boldsymbol v^{\varepsilon}\rVert^2_{{L}^2(\Omega)}
		+\nu\varepsilon\lVert\nabla\boldsymbol v^{\varepsilon}\rVert^2_{{L}^2(\Omega)}
		\leq&|\langle{\boldsymbol u}^{\varepsilon} \cdot \nabla {\boldsymbol v}^{\varepsilon}
		+	{\varepsilon}^{-1}\boldsymbol{R}\boldsymbol v^{\varepsilon}
		+	{\varepsilon}^{-1}\nabla p_v,{\boldsymbol v}^{\varepsilon}\rangle|\notag\\
		&+|\langle
		{\boldsymbol v}^{\varepsilon} \cdot \nabla\boldsymbol{u}^\varepsilon_{app},{\boldsymbol v}^{\varepsilon}\rangle|
		+|\langle
		\boldsymbol \rho^{\varepsilon}	,{\boldsymbol v}^{\varepsilon}\rangle|\,,	
	\end{align}
	where we use the notation $<, >$ as the ${L}^2(\Omega)$ inner production.

	Using the incompressible condition and noting that $\boldsymbol{R}$ is skew-symmetry, one has
	\begin{equation}\label{4.4}
		|\langle{\boldsymbol u}^{\varepsilon} \cdot \nabla {\boldsymbol v}^{\varepsilon}
		+	{\varepsilon}^{-1}\boldsymbol{R}\boldsymbol v^{\varepsilon}
		+		{\varepsilon}^{-1}\nabla p_v,{\boldsymbol v}^{\varepsilon}\rangle|=0\,.	
	\end{equation}

	As to the term  $|\langle
	{\boldsymbol v}^{\varepsilon} \cdot \nabla\boldsymbol{u}^\varepsilon_{app},{\boldsymbol v}^{\varepsilon}\rangle|$, due to the different nature of internal, boundary, and corrective, we decompose the approximate solution \eqref{3.1} to obtain
	\begin{align*}
		|\langle
		{\boldsymbol v}^{\varepsilon} \cdot \nabla\boldsymbol{u}^\varepsilon_{app},{\boldsymbol v}^{\varepsilon}\rangle|
		\leq\big|\big\langle{\boldsymbol v}^{\varepsilon} \cdot \nabla \boldsymbol{u}^{BL}
		,{\boldsymbol v}^{\varepsilon}\big\rangle\big|
		+|\langle{\boldsymbol v}^{\varepsilon} \cdot \nabla
		(
		\boldsymbol{u}^I+
		\boldsymbol{u^c}
		),{\boldsymbol v}^{\varepsilon}\rangle|\,.
	\end{align*}
	
	First, we estimate the nonlinear terms involving the boundary terms. By integration by parts, we obtain
	\begin{align*}
		\big|\big\langle{\boldsymbol v}^{\varepsilon} \cdot \nabla \boldsymbol{u}^{BL}
		,{\boldsymbol v}^{\varepsilon}\big\rangle\big|
		=&
		\big|\big\langle{\boldsymbol v}^{\varepsilon} \cdot \nabla{\boldsymbol v}^{\varepsilon},\boldsymbol{u}^{BL}\big\rangle\big|\\\notag
		\leq&\Big|\int_{\mathbb{R}^2\times[B,B+1]} {\boldsymbol v}^{\varepsilon} \cdot \nabla {\boldsymbol v}^{\varepsilon}\cdot
		\boldsymbol{u}^{BL}
		\,{d}x{d}y{d}z
		\Big|\\
		&+\Big|\int_{\mathbb{R}^2\times[B+1,B+2]} {\boldsymbol v}^{\varepsilon} \cdot \nabla {\boldsymbol v}^{\varepsilon}\cdot
		\boldsymbol{u}^{BL}
		\,{d}x{d}y{d}z\Big|\,.
	\end{align*}
	Due to ${\boldsymbol v}^{\varepsilon}$ vanishes at the boundary, we have
	\begin{align*}
		|{\boldsymbol v}^{\varepsilon}|
		=\Big| \int_B^z \partial_\varsigma {\boldsymbol v}^{\varepsilon}(t,x,y,\varsigma)\,{d}\varsigma\Big|
		\leq d(z)^{\frac{1}{2}}\lVert\partial_z {\boldsymbol v}^{\varepsilon}\rVert_{{L}^2([B,B+2])}\,,
	\end{align*}
	where $d(z)$ denotes the distance to the boundary.
	Therefore, using Proposition \ref{propBL},  we can conclude that
	\begin{align*}
		&\Big|\int_{\mathbb{R}^2\times[B,B+1]} {\boldsymbol v}^{\varepsilon} \cdot \nabla {\boldsymbol v}^{\varepsilon}\cdot
		\boldsymbol{u}^{BL}
		\,dxdydz
		\Big|
		\\
		&\qquad\leq \int_\Omega
		\lVert\partial_z {\boldsymbol v}^{\varepsilon}\rVert_{{L}^2([B,B+2])}
		|\nabla {\boldsymbol v}^{\varepsilon}|
		d(z)^{\frac{1}{2}}
		\big|\boldsymbol{u}^{BL}\big|\,{d}x{d}y{d}z\\
		&\qquad\leq  \lVert\nabla{\boldsymbol v}^{\varepsilon}\rVert^2_{{L}^2(\Omega)}
		\lVert d(z)^{\frac{1}{2}} \big|\boldsymbol{u}^{BL}\big| \rVert_{{L}^2([B,B+2];{L}^\infty(\mathbb{R}^2))}\\
		&\qquad\lesssim
		\tfrac{3\sqrt{3\nu}}{4}
		\varepsilon
		\big(
		\lVert{\boldsymbol{u}}_{0,h}\rVert_{L^\infty(\mathbb{R}^2)} 
		+\varepsilon\big(
		\lVert{\boldsymbol{u}}_{0,h}\rVert_{L^2(\mathbb{R}^2)}+\lVert{\omega}_0\rVert_{H^2(\mathbb{R}^2)}
		\big)
		\big)\mathrm{e}^{-\frac{\sqrt{2\nu}}{8}  t}
		\lVert\nabla{\boldsymbol v}^{\varepsilon}\rVert^2_{{L}^2(\Omega)}
		\notag\\
		&\qquad\quad+\tfrac{3\sqrt{3\nu}}{4}\varepsilon^3
		\big(
		1+
		\lVert {\boldsymbol{u}}_{0,h} \rVert_{H^{3}(\mathbb{R}^2)}
		\big)
		\lVert {\boldsymbol{u}}_{0,h} \rVert_{H^{5}(\mathbb{R}^2)}
		\lVert\nabla{\boldsymbol v}^{\varepsilon}\rVert^2_{{L}^2(\Omega)}\,.
	\end{align*}
	Thus, we have
	\begin{align}\label{4.6}
		&\big|\big\langle{\boldsymbol v}^{\varepsilon} \cdot \nabla
		\boldsymbol{u}^{BL},{\boldsymbol v}^{\varepsilon}\big\rangle\big|\notag\\
		&\qquad\lesssim
		\tfrac{3\sqrt{3\nu}}{4}
		\varepsilon
		\big(
		\lVert{\boldsymbol{u}}_{0,h}\rVert_{L^\infty(\mathbb{R}^2)} 
		+\varepsilon\big(
		\lVert{\boldsymbol{u}}_{0,h}\rVert_{L^2(\mathbb{R}^2)}+\lVert{\omega}_0\rVert_{H^2(\mathbb{R}^2)}
		\big)
		\big)\mathrm{e}^{-\frac{\sqrt{2\nu}}{8}  t}
		\lVert\nabla{\boldsymbol v}^{\varepsilon}\rVert^2_{{L}^2(\Omega)}
		\notag\\
		&\qquad\quad+\tfrac{3\sqrt{3\nu}}{4}\varepsilon^3
		\big(
		1+
		\lVert {\boldsymbol{u}}_{0,h} \rVert_{H^{3}(\mathbb{R}^2)}
		\big)
		\lVert {\boldsymbol{u}}_{0,h} \rVert_{H^{5}(\mathbb{R}^2)}
		\lVert\nabla{\boldsymbol v}^{\varepsilon}\rVert^2_{{L}^2(\Omega)}\,.
	\end{align}

	Using H\"{o}lder's inequality, one has
	\begin{align}\label{4.7}
		&|\langle{\boldsymbol v}^{\varepsilon} \cdot \nabla
		(
		\boldsymbol{u}^I+
		\boldsymbol{u^c}
		),{\boldsymbol v}^{\varepsilon}\rangle|\notag\\
		&\qquad\leq\lVert {\boldsymbol v}^{\varepsilon} \rVert^2_{{L}^2(\Omega)}
		\lVert \nabla
		\big(
		\boldsymbol u^{I}+\boldsymbol u^{c}
		\big)
		\rVert_{{L}^\infty(\Omega)}
		\notag\\
		&\qquad\leq
		\big(
		\lVert{\boldsymbol{u}}_{0,h}\rVert_{L^2(\mathbb{R}^2)}+\lVert{\omega}_0\rVert_{H^2(\mathbb{R}^2)}
		\big)\mathrm{e}^{-\frac{\sqrt{2\nu}}{8}  t}\lVert {\boldsymbol v}^{\varepsilon} \rVert^2_{{L}^2(\Omega)}
		\notag\\
		&\qquad\quad+\big(\varepsilon
		\lVert {\boldsymbol{u}}_{0,h} \rVert_{H^{4}(\mathbb{R}^2)}
		+\varepsilon^2
		\big(
		1+
		\lVert {\boldsymbol{u}}_{0,h} \rVert_{H^{4}(\mathbb{R}^2)}
		\big)
		\lVert {\boldsymbol{u}}_{0,h} \rVert_{H^{6}(\mathbb{R}^2)}\big)\lVert {\boldsymbol v}^{\varepsilon} \rVert^2_{{L}^2(\Omega)}
		\,,
	\end{align}
	where the above estimates utilize the results in Propositions \ref{propIC}.
	
	Based on Theorem \ref{thapp}, $|\langle
	\boldsymbol \rho^{\varepsilon}	,{\boldsymbol v}^{\varepsilon}\rangle|$ can be shown as
	\begin{equation}\label{4.10}
		|\langle
		\boldsymbol \rho^{\varepsilon}	,{\boldsymbol v}^{\varepsilon}\rangle|
		\lesssim
		\varepsilon^2
		\lVert
		{\boldsymbol{u}}_{0,h}
		\rVert_{{H}^{5}(\mathbb{R}^2)}
		\big(
		\lVert
		{\boldsymbol{u}}_{0,h}
		\rVert_{{H}^{3}(\mathbb{R}^2)}
		+1\big)^2\lVert{\boldsymbol v}^{\varepsilon}\rVert_{{L}^2(\Omega)}
		\,.
	\end{equation}

	Summing all inequalities \eqref{4.6}-\eqref{4.10} and applying Proposition \ref{prop3}, for sufficiently small $\varepsilon > 0$, we can derive from \eqref{4.3} that
	\begin{align}
		\label{4.15}
		&\frac{1}{2}\frac{d}{dt}\lVert\boldsymbol v^{\varepsilon}\rVert^2_{{L}^2(\Omega)}
		+\nu\varepsilon
		\lVert\nabla\boldsymbol v^{\varepsilon}\rVert^2_{{L}^2(\Omega)}\notag\\
		&\qquad\lesssim
		\tfrac{3\sqrt{3\nu}}{4}
		\varepsilon
		\big(
		\lVert{\boldsymbol{u}}_{0,h}\rVert_{L^\infty(\mathbb{R}^2)} 
		+\varepsilon\big(
		\lVert{\boldsymbol{u}}_{0,h}\rVert_{L^2(\mathbb{R}^2)}+\lVert{\omega}_0\rVert_{H^2(\mathbb{R}^2)}
		\big)
		\big)\mathrm{e}^{-\frac{\sqrt{2\nu}}{8}  t}
		\lVert\nabla{\boldsymbol v}^{\varepsilon}\rVert^2_{{L}^2(\Omega)}
		\notag\\
		&\qquad\quad+\tfrac{3\sqrt{3\nu}}{4}\varepsilon^3
		\big(
		1+
		\lVert {\boldsymbol{u}}_{0,h} \rVert_{H^{3}(\mathbb{R}^2)}
		\big)
		\lVert {\boldsymbol{u}}_{0,h} \rVert_{H^{5}(\mathbb{R}^2)}
		\lVert\nabla{\boldsymbol v}^{\varepsilon}\rVert^2_{{L}^2(\Omega)}\notag\\
		&\qquad\quad+
		\big(
		\lVert{\boldsymbol{u}}_{0,h}\rVert_{L^2(\mathbb{R}^2)}+\lVert{\omega}_0\rVert_{H^2(\mathbb{R}^2)}
		\big)\mathrm{e}^{-\frac{\sqrt{2\nu}}{8}  t}\lVert {\boldsymbol v}^{\varepsilon} \rVert^2_{{L}^2(\Omega)}
		\notag\\
		&\qquad\quad+\big(\varepsilon
		\lVert {\boldsymbol{u}}_{0,h} \rVert_{H^{4}(\mathbb{R}^2)}
		+\varepsilon^2
		\big(
		1+
		\lVert {\boldsymbol{u}}_{0,h} \rVert_{H^{4}(\mathbb{R}^2)}
		\big)
		\lVert {\boldsymbol{u}}_{0,h} \rVert_{H^{6}(\mathbb{R}^2)}\big)\lVert {\boldsymbol v}^{\varepsilon} \rVert^2_{{L}^2(\Omega)}\notag\\
		&\qquad\quad+\varepsilon^2
		\lVert
		{\boldsymbol{u}}_{0,h}
		\rVert_{{H}^{5}(\mathbb{R}^2)}
		\big(
		\lVert
		{\boldsymbol{u}}_{0,h}
		\rVert_{{H}^{3}(\mathbb{R}^2)}
		+1\big)^2\lVert{\boldsymbol v}^{\varepsilon}\rVert_{{L}^2(\Omega)}\,.
	\end{align}
Finally, using the Lemma \ref{f2} and condition \eqref{1.26}, we can establish the smallness estimate in Proposition \ref{prop5}.
\end{proof}

\subsection{{$L^2$ estimations for $\partial_t \boldsymbol{v}^\varepsilon$ and $\nabla \boldsymbol{v}^\varepsilon$}}\label{sub4.2}

In this subsection, we compute the value of $\lVert\partial_t\boldsymbol v^{\varepsilon}\rVert_{{L}^2(\Omega)}$
and derive the estimate for $\lVert\nabla\boldsymbol v^{\varepsilon}\rVert_{{L}^2(\Omega)}$
using Proposition \ref{prop5}.
\begin{proposition}\label{prop8}
	Under the assumptions of Theorem \ref{thm2}, then there exist some $\varepsilon>0$ such that
	\begin{equation}\label{4.1.1}
		\lVert\partial_t\boldsymbol v^{\varepsilon}\rVert^2_{L^\infty([0,T),{L}^2(\Omega))}
		\lesssim\varepsilon^3\,.
	\end{equation}
	In particular, it holds that
	\begin{equation}\label{4.2.1}
		\lVert\nabla\boldsymbol v^{\varepsilon}\rVert^2_{L^\infty([0,T),{L}^2(\Omega))}
		\lesssim\varepsilon^{\frac52}\,.
	\end{equation}
\end{proposition}
\begin{proof}
	First acting the operator $\partial_t$ on the system \eqref{2.65} and then making an inner product with $\partial_t\boldsymbol{v}^\varepsilon$ yields
	\begin{align}
		\label{4.16.1}
		&\frac{1}{2}\frac{d}{dt}\lVert\partial_t\boldsymbol v^{\varepsilon}\rVert^2_{{L}^2(\Omega)}
		+\nu\varepsilon\lVert\nabla\partial_t\boldsymbol v^{\varepsilon}\rVert^2_{{L}^2(\Omega)}\notag\\
		&\qquad\leq
		|\langle{\boldsymbol u}^{\varepsilon} \cdot \nabla \partial_t{\boldsymbol v}^{\varepsilon}
		+	{\varepsilon}^{-1}\boldsymbol{R}
		\partial_t\boldsymbol v^{\varepsilon}
		+		{\varepsilon}^{-1}\nabla \partial_tp_v,\partial_t{\boldsymbol v}^{\varepsilon}\rangle|\notag\\
		&\qquad\quad+
		|\langle
		\partial_t{\boldsymbol v}^{\varepsilon}\cdot
		\nabla\boldsymbol{u}^{\varepsilon}_{app},
		\partial_t{\boldsymbol v}^{\varepsilon}
		\rangle|
		+
		|\langle
		\partial_t{\boldsymbol v}^{\varepsilon} \cdot \nabla\boldsymbol{v}^\varepsilon
		,\partial_t{\boldsymbol v}^{\varepsilon}\rangle|\notag\\
		&\qquad\quad+
		|\langle
		\partial_t{\boldsymbol u}_{app}^{\varepsilon} \cdot \nabla\boldsymbol{v}^\varepsilon
		,\partial_t{\boldsymbol v}^{\varepsilon}\rangle|
		+
		|\langle{\boldsymbol v}^{\varepsilon} \cdot \nabla\partial_t\boldsymbol{u}^\varepsilon_{app},
		\partial_t{\boldsymbol v}^{\varepsilon}\rangle|	
		+|\langle
		\partial_t\boldsymbol \rho^{\varepsilon}	,\partial_t{\boldsymbol v}^{\varepsilon}\rangle|
		\,.
	\end{align}

	For the first two terms, by applying methods analogous to those used in the $L^2$ estimates of equations \eqref{4.4}-\eqref{4.7}, we can obtain
	\begin{equation}\label{3.2.1}
		|\langle{\boldsymbol u}^{\varepsilon} \cdot \nabla \partial_t{\boldsymbol v}^{\varepsilon}
		+	{\varepsilon}^{-1}\boldsymbol{R}
		\partial_t\boldsymbol v^{\varepsilon}
		+		{\varepsilon}^{-1}\nabla \partial_t p_v,\partial_t{\boldsymbol v}^{\varepsilon}\rangle|=0\,,
	\end{equation}
	and
	\begin{align}\label{3.3.1}
		&|\langle
		\partial_t{\boldsymbol v}^{\varepsilon}\cdot
		\nabla\boldsymbol{u}^{\varepsilon}_{app},
		\partial_t{\boldsymbol v}^{\varepsilon}
		\rangle|\notag\\
		&\qquad\lesssim
		\tfrac{3\sqrt{3\nu}}{4}
		\varepsilon
		\big(
		\lVert{\boldsymbol{u}}_{0,h}\rVert_{L^\infty(\mathbb{R}^2)} 
		+\varepsilon\big(
		\lVert{\boldsymbol{u}}_{0,h}\rVert_{L^2(\mathbb{R}^2)}+\lVert{\omega}_0\rVert_{H^2(\mathbb{R}^2)}
		\big)
		\big)\mathrm{e}^{-\frac{\sqrt{2\nu}}{8}  t}
		\lVert\nabla\partial_t{\boldsymbol v}^{\varepsilon}\rVert^2_{{L}^2(\Omega)}
		\notag\\
		&\qquad\quad+\tfrac{3\sqrt{3\nu}}{4}\varepsilon^3
		\big(
		1+
		\lVert {\boldsymbol{u}}_{0,h} \rVert_{H^{3}(\mathbb{R}^2)}
		\big)
		\lVert {\boldsymbol{u}}_{0,h} \rVert_{H^{5}(\mathbb{R}^2)}
		\lVert\nabla\partial_t{\boldsymbol v}^{\varepsilon}\rVert^2_{{L}^2(\Omega)}\notag\\
		&\qquad\quad+
		\big(
		\lVert{\boldsymbol{u}}_{0,h}\rVert_{L^2(\mathbb{R}^2)}+\lVert{\omega}_0\rVert_{H^2(\mathbb{R}^2)}
		\big)\mathrm{e}^{-\frac{\sqrt{2\nu}}{8}  t}\lVert \partial_t{\boldsymbol v}^{\varepsilon} \rVert^2_{{L}^2(\Omega)}
		\notag\\
		&\qquad\quad+\big(\varepsilon
		\lVert {\boldsymbol{u}}_{0,h} \rVert_{H^{4}(\mathbb{R}^2)}
		+\varepsilon^2
		\big(
		1+
		\lVert {\boldsymbol{u}}_{0,h} \rVert_{H^{4}(\mathbb{R}^2)}
		\big)
		\lVert {\boldsymbol{u}}_{0,h} \rVert_{H^{6}(\mathbb{R}^2)}\big)\lVert\partial_t {\boldsymbol v}^{\varepsilon} \rVert^2_{{L}^2(\Omega)}\,.
	\end{align}

	For the third term, by employing the Cauchy-Schwarz inequality and the Gagliardo-Nirenberg inequality, we derive the following result
	\begin{align}\label{4.17.1}
		|\langle
		\partial_t{\boldsymbol v}^{\varepsilon} \cdot \nabla\boldsymbol{v}^\varepsilon)
		,\partial_t{\boldsymbol v}^{\varepsilon}\rangle|
		\leq
		&\lVert
		\nabla{\boldsymbol v}^{\varepsilon} 	\rVert_{{L}^2(\Omega)}
		\lVert
		\partial_t{\boldsymbol v}^{\varepsilon} 	\rVert^2_{{L}^4(\Omega)}
		\notag\\
		\leq &
		\lVert\nabla{\boldsymbol v}^{\varepsilon} 	\rVert_{{L}^2(\Omega)}
		\lVert
		\partial_t{\boldsymbol v}^{\varepsilon} 	\rVert^{\frac12}_{{L}^2(\Omega)}
		\lVert
		\nabla\partial_t{\boldsymbol v}^{\varepsilon} 	\rVert^{\frac32}_{{L}^2(\Omega)}
		\notag\\
		\lesssim&
		\frac{\nu\varepsilon}{8}
		\rVert
		\nabla\partial_t{\boldsymbol v}^{\varepsilon} 	\lVert^{2}_{{L}^2(\Omega)}
		+
		\varepsilon^{-3}
		\lVert\nabla{\boldsymbol v}^{\varepsilon} 	\rVert^4_{{L}^2(\Omega)}
		\lVert
		\partial_t{\boldsymbol v}^{\varepsilon} 	\rVert^{2}_{{L}^2(\Omega)}
		\,.
	\end{align}
	Below we estimate item $\lVert\nabla{\boldsymbol v}^{\varepsilon} 	\rVert_{{L}^2(\Omega)}$ in \eqref{4.17.1}. Using Proposition \ref{prop5}, we have  
	\begin{equation*}
		\frac{d}{dt}\lVert\boldsymbol v^{\varepsilon}\rVert^2_{{L}^2(\Omega)}
		+\frac{\nu\varepsilon}{2}
		\lVert\nabla\boldsymbol v^{\varepsilon}\rVert^2_{{L}^2(\Omega)} \lesssim \varepsilon^4 \,,
	\end{equation*}
	which implies  
	\begin{equation*}
		\varepsilon\lVert\nabla\boldsymbol v^{\varepsilon}\rVert^2_{{L}^2(\Omega)} \lesssim \varepsilon^4 
		+\Big|
		\frac{d}{dt}\lVert\boldsymbol v^{\varepsilon}\rVert^2_{{L}^2(\Omega)}
		\Big|\,.
	\end{equation*} 
	Applying the Cauchy-Schwarz inequality, we further derive 
	\begin{align}\label{3.1.1}
		\lVert\nabla\boldsymbol v^{\varepsilon}\rVert^2_{{L}^2(\Omega)} \lesssim
		& \varepsilon^3 
		+\varepsilon^{-1}\Big|
		\int_{\Omega} 2 \boldsymbol v^{\varepsilon} \cdot {\partial_t}\boldsymbol v^{\varepsilon} \, dx
		\Big|\notag\\
		\lesssim&\varepsilon^3 
		+
		2\varepsilon^{-1}\lVert\boldsymbol v^{\varepsilon}\rVert_{{L}^2(\Omega)}
		\lVert\partial_t\boldsymbol v^{\varepsilon}\rVert_{{L}^2(\Omega)}\notag\\
		\lesssim&\varepsilon^3 
		+
		\varepsilon 
		\lVert\partial_t\boldsymbol v^{\varepsilon}\rVert_{{L}^2(\Omega)}\,.
	\end{align} 
	Thus, \eqref{4.17.1} can be rewritten as
	\begin{align}\label{4.17.2}
		|\langle
		\partial_t{\boldsymbol v}^{\varepsilon} \cdot \nabla\boldsymbol{v}^\varepsilon)
		,\partial_t{\boldsymbol v}^{\varepsilon}\rangle|
		\lesssim&
		\frac{\nu\varepsilon}{8}
		\rVert
		\nabla\partial_t{\boldsymbol v}^{\varepsilon} 	\lVert^{2}_{{L}^2(\Omega)}
		+
		\varepsilon^{3}\lVert
		\partial_t{\boldsymbol v}^{\varepsilon} 	\rVert^{2}_{{L}^2(\Omega)}
		+
		\varepsilon^{-1}\lVert
		\partial_t{\boldsymbol v}^{\varepsilon} 	\rVert^{4}_{{L}^2(\Omega)}
		\,.
	\end{align}

	For the nonlinear terms related to the approximate solution, we decompose the approximate solution following the approach in Proposition \ref{prop5}. Using the results from Propositions \ref{propIC} and \ref{propBL} and the incompressible conditions of $\boldsymbol{v}^\varepsilon$ and $\boldsymbol{u}^\varepsilon_{app}$, we obtain
	\begin{align}\label{4.21.1}
		&|\langle
		\partial_t{\boldsymbol u}_{app}^{\varepsilon} \cdot \nabla\boldsymbol{v}^\varepsilon
		,\partial_t{\boldsymbol v}^{\varepsilon}\rangle|
		+
		|\langle{\boldsymbol v}^{\varepsilon} \cdot \nabla\partial_t\boldsymbol{u}^\varepsilon_{app},
		\partial_t{\boldsymbol v}^{\varepsilon}\rangle|	
		\notag\\
		&\qquad=|\langle
		\partial_t{\boldsymbol u}_{app}^{\varepsilon} \cdot \nabla\partial_t\boldsymbol{v}^\varepsilon
		,{\boldsymbol v}^{\varepsilon}\rangle|
		+
		|\langle{\boldsymbol v}^{\varepsilon} \cdot \nabla\partial_t{\boldsymbol v}^{\varepsilon},
		\partial_t\boldsymbol{u}^\varepsilon_{app}\rangle|	
		\notag\\
		&\qquad\leq
		\rVert
		\nabla\partial_t{\boldsymbol v}^{\varepsilon} 	\lVert_{{L}^2(\Omega)}
		\rVert
		{\boldsymbol v}^{\varepsilon} 	\lVert_{{L}^2(\Omega)}
		\rVert
		\partial_t\boldsymbol{u}^\varepsilon_{app}
		\lVert_{{L}^\infty(\Omega)}\notag\\
		&\qquad\lesssim
		\frac{\nu\varepsilon}{8}
		\rVert
		\nabla\partial_t{\boldsymbol v}^{\varepsilon} 	\lVert^{2}_{{L}^2(\Omega)}
		+
		\varepsilon^{-1}\rVert
		{\boldsymbol v}^{\varepsilon} 	\lVert^2_{{L}^2(\Omega)}
		\rVert
		\partial_t\boldsymbol{u}^\varepsilon_{app}\lVert^2_{{L}^\infty(\Omega)}
		\notag\\
		&\qquad\lesssim
		\frac{\nu\varepsilon}{8}
		\rVert
		\nabla\partial_t{\boldsymbol v}^{\varepsilon} 	\lVert^{2}_{{L}^2(\Omega)}
		+
		\varepsilon^3
		\,.
	\end{align}

	According to Theorem \ref{thapp} and , the last term is rearranged as
	\begin{align}\label{4.19.1}
		|\langle
		\partial_t\boldsymbol \rho^{\varepsilon}	,\partial_t{\boldsymbol v}^{\varepsilon}\rangle|
		\leq&
		\lVert
		\partial_t\boldsymbol{\rho}^\varepsilon\rVert_{{L}^{2}(\Omega)}
		\rVert
		\partial_t{\boldsymbol v}^{\varepsilon}
		\lVert_{{L}^2(\Omega)}
		\notag\\	
		\lesssim&
		\rVert
		\partial_t{\boldsymbol v}^{\varepsilon}
		\lVert^2_{{L}^2(\Omega)}
		+\varepsilon^4
		\lVert
		{\boldsymbol{u}}_{0,h}
		\rVert^2_{{H}^{6}(\mathbb{R}^2)}
		\big(
		\lVert
		{\boldsymbol{u}}_{0,h}
		\rVert^2_{{H}^{4}(\mathbb{R}^2)}
		+1
		\big)^2
		\,.
	\end{align}

	According to \eqref{3.2.1}-\eqref{4.19.1}, we can rewrite equation \eqref{4.16.1} as follows
	\begin{align}
		\label{4.20.1}
		&\frac{1}{2}\frac{d}{dt}\lVert\partial_t\boldsymbol v^{\varepsilon}\rVert^2_{{L}^2(\Omega)}
		+\frac{3\nu\varepsilon}{4}\lVert\nabla\partial_t\boldsymbol v^{\varepsilon}\rVert^2_{{L}^2(\Omega)}\notag\\
		&\qquad\lesssim
		\tfrac{3\sqrt{3\nu}}{4}
		\varepsilon
		\big(
		\lVert{\boldsymbol{u}}_{0,h}\rVert_{L^\infty(\mathbb{R}^2)} 
			+\varepsilon\big(
		\lVert{\boldsymbol{u}}_{0,h}\rVert_{L^2(\mathbb{R}^2)}+\lVert{\omega}_0\rVert_{H^2(\mathbb{R}^2)}
		\big)
		\big)\mathrm{e}^{-\frac{\sqrt{2\nu}}{8}  t}
		\lVert\nabla\partial_t{\boldsymbol v}^{\varepsilon}\rVert^2_{{L}^2(\Omega)}
		\notag\\
		&\qquad\quad+
		\tfrac{3\sqrt{3\nu}}{4}\varepsilon^3
		\big(
		1+
		\lVert {\boldsymbol{u}}_{0,h} \rVert_{H^{3}(\mathbb{R}^2)}
		\big)
		\lVert {\boldsymbol{u}}_{0,h} \rVert_{H^{5}(\mathbb{R}^2)}
		\lVert\nabla\partial_t{\boldsymbol v}^{\varepsilon}\rVert^2_{{L}^2(\Omega)}\notag\\
		&\qquad\quad
		+
		\varepsilon^{-1}\lVert
		\partial_t{\boldsymbol v}^{\varepsilon} 	\rVert^{4}_{{L}^2(\Omega)}
		+		(1+\varepsilon^{3})\lVert
		\partial_t{\boldsymbol v}^{\varepsilon} 	\rVert^{2}_{{L}^2(\Omega)}+
		\varepsilon^3
		\,.
	\end{align}

	When  $\varepsilon$ is sufficiently small, we can simplify equation \eqref{4.20.1} with the following expressions:
	\begin{align}
		 \varepsilon^{-1} \lVert
		\partial_t{\boldsymbol v}^{\varepsilon} 	\rVert^{4}_{{L}^2(\Omega)}&\leq
		\frac{
			\lVert
			\partial_t{\boldsymbol v}^{\varepsilon} 	\rVert^{4}_{{L}^2(\Omega)}
		}{
			\varepsilon^3 
		}\,,\label{1}\\
		(1+\varepsilon^{3} )\lVert
		\partial_t{\boldsymbol v}^{\varepsilon} 	\rVert^{2}_{{L}^2(\Omega)}&\leq
		\frac{
			\lVert
			\partial_t{\boldsymbol v}^{\varepsilon} 	\rVert^{4}_{{L}^2(\Omega)}
		}{
			\varepsilon^3 
		}
		+\varepsilon^3 \,.\label{2}
	\end{align}

	Under condition  \eqref{1.26} in Theorem \ref{thm2}, if $\varepsilon$ is sufficiently small, we obtain 
	\begin{equation}
		\label{4.20.2}
		\frac{1}{2}\frac{d}{dt}\lVert\partial_t\boldsymbol v^{\varepsilon}\rVert^2_{{L}^2(\Omega)}
		\lesssim
		\frac{
			\lVert
			\partial_t{\boldsymbol v}^{\varepsilon} 	\rVert^{4}_{{L}^2(\Omega)}
		}{
			\varepsilon^3
		}
		+\varepsilon^3\,.
	\end{equation}
	
	Finally, by applying Lemma \ref{f4} for $t< \frac\pi2$, we derive the result \eqref{4.1.1}, and subsequently obtain \eqref{4.2.1}.
\end{proof}

\subsection{$L^\infty$ estimation for $\boldsymbol{v}^\varepsilon$}\label{sub4.3}

\textbf{Proof of the Theorem \ref{thm2}.}
Since the $L^\infty$ estimate of \( \boldsymbol{v}^\varepsilon \) involves $\lVert \Delta\boldsymbol{v}^\varepsilon\rVert_{{L}^{2}(\Omega)}$, using the equations \eqref{2.65} satisfied by $\boldsymbol{v}^\varepsilon$, we can obtain
\begin{align}\label{3.3.1.1}
	\nu\varepsilon
	\lVert\Delta{\boldsymbol v}^{\varepsilon}\rVert_{{L}^{2}(\Omega)}
	=&
	\lVert
	{\boldsymbol v}^{\varepsilon} \cdot \nabla {\boldsymbol v}^{\varepsilon}
	+
	{\boldsymbol u}_{app}^{\varepsilon} \cdot \nabla {\boldsymbol v}^{\varepsilon}
	+
	{\boldsymbol v}^{\varepsilon} \cdot \nabla\boldsymbol{u}^\varepsilon_{app}
	+
	{\varepsilon}^{-1}\boldsymbol{R}\boldsymbol v^{\varepsilon}
	+\boldsymbol \rho^{\varepsilon}
	\rVert_{{L}^{2}(\Omega)}\notag
	\\
	&+\lVert
	\partial_{t} {\boldsymbol v}^{\varepsilon}
	\rVert_{{L}^{2}(\Omega)}+
	{\varepsilon}^{-1}\lVert
	\nabla p_v
	\rVert_{{L}^{2}(\Omega)}
	\,.
\end{align}

Using the incompressibility condition of \( \boldsymbol{v}^\varepsilon \), we apply the divergence operator \( \nabla\cdot \) to equation \eqref{2.65} and find that the pressure term satisfies
\begin{equation}\label{3.3.2}
	\varepsilon^{-1}\Delta p_v
	=-\nabla\cdot
	\big(
	{\boldsymbol v}^{\varepsilon} \cdot \nabla {\boldsymbol v}^{\varepsilon}
	+
	{\boldsymbol u}_{app}^{\varepsilon} \cdot \nabla {\boldsymbol v}^{\varepsilon}
	+
	{\boldsymbol v}^{\varepsilon} \cdot \nabla\boldsymbol{u}^\varepsilon_{app}
	+
	{\varepsilon}^{-1}\boldsymbol{R}\boldsymbol v^{\varepsilon}
	+\boldsymbol \rho^{\varepsilon}\,.
	\big)
\end{equation}

Moreover, the boundedness of the Riesz operator yields an estimate for the pressure term
\begin{equation}\label{3.3.3}
	{\varepsilon}^{-1}\lVert
	\nabla p_v
	\rVert_{{L}^{2}(\Omega)}
	\lesssim
	\lVert
	{\boldsymbol v}^{\varepsilon} \cdot \nabla {\boldsymbol v}^{\varepsilon}
	+
	{\boldsymbol u}_{app}^{\varepsilon} \cdot \nabla {\boldsymbol v}^{\varepsilon}
	+
	{\boldsymbol v}^{\varepsilon} \cdot \nabla\boldsymbol{u}^\varepsilon_{app}
	+
	{\varepsilon}^{-1}\boldsymbol{R}\boldsymbol v^{\varepsilon}
	+\boldsymbol \rho^{\varepsilon}
	\rVert_{{L}^{2}(\Omega)}\,.
\end{equation}

Given the results from Proposition \ref{prop5} and \ref{prop8}, combining equations \eqref{3.3.1}-\eqref{3.3.3} and Proposition \ref{propIC}-\ref{propBL}, we obtain
\begin{align}\label{3.3.4}
	\nu\varepsilon
	\lVert\Delta{\boldsymbol v}^{\varepsilon}\rVert_{{L}^{2}(\Omega)}
	\lesssim&
	\lVert
	{\boldsymbol v}^{\varepsilon} \cdot \nabla {\boldsymbol v}^{\varepsilon}
	\rVert_{{L}^{2}(\Omega)}
	+\lVert
	{\boldsymbol u}_{app}^{\varepsilon} \cdot \nabla {\boldsymbol v}^{\varepsilon}
	\rVert_{{L}^{2}(\Omega)}
	+\lVert
	{\boldsymbol v}^{\varepsilon} \cdot \nabla\boldsymbol{u}^\varepsilon_{app}
	\rVert_{{L}^{2}(\Omega)}\notag\\
	&+
	\lVert{\varepsilon}^{-1}\boldsymbol{R}\boldsymbol v^{\varepsilon}
	\rVert_{{L}^{2}(\Omega)}
	+\lVert\boldsymbol \rho^{\varepsilon}
	\rVert_{{L}^{2}(\Omega)}
	+\lVert
	\partial_{t} {\boldsymbol v}^{\varepsilon}
	\rVert_{{L}^{2}(\Omega)}
	\notag\\
	\lesssim&
	\lVert
	{\boldsymbol v}^{\varepsilon}
	\rVert_{{L}^{6}(\Omega)}
	\lVert\nabla {\boldsymbol v}^{\varepsilon}
	\rVert_{{L}^{3}(\Omega)}
	+\lVert
	{\boldsymbol u}_{app}^{\varepsilon}
	\rVert_{{L}^{\infty}(\Omega)} 
	\lVert \nabla {\boldsymbol v}^{\varepsilon}
	\rVert_{{L}^{2}(\Omega)}\notag\\
	&
	+\lVert
	\partial_z{\boldsymbol v}^{\varepsilon}
	\rVert_{{L}^{2}(\Omega)}
	\lVert d(z)^{\frac{1}{2}} |\nabla\boldsymbol{u}^{BL}| \rVert_{{L}^2([B,B+2];{L}^\infty(\mathbb{R}^2))}\notag\\
	&+
	{\varepsilon}^{-1}\lVert\boldsymbol v^{\varepsilon}
	\rVert_{{L}^{2}(\Omega)}
	+\lVert\boldsymbol \rho^{\varepsilon}
	\rVert_{{L}^{2}(\Omega)}
	+\lVert
	\partial_{t} {\boldsymbol v}^{\varepsilon}
	\rVert_{{L}^{2}(\Omega)}
	\notag\\
	\lesssim&
	\lVert
	\nabla{\boldsymbol v}^{\varepsilon}
	\rVert^{\frac32}_{{L}^{2}(\Omega)}
	\lVert\Delta {\boldsymbol v}^{\varepsilon}
	\rVert^{\frac12}_{{L}^{2}(\Omega)}
	+\lVert
	{\boldsymbol u}^{I}
	\rVert_{{W}^{1,\infty}(\Omega)} 
	\lVert \nabla {\boldsymbol v}^{\varepsilon}
	\rVert_{{L}^{2}(\Omega)}\notag\\
	&+
	{\varepsilon}^{-1}\lVert\boldsymbol v^{\varepsilon}
	\rVert_{{L}^{2}(\Omega)}
	+\lVert\boldsymbol \rho^{\varepsilon}
	\rVert_{{L}^{2}(\Omega)}
	+\lVert
	\partial_{t} {\boldsymbol v}^{\varepsilon}
	\rVert_{{L}^{2}(\Omega)}\notag\\
	\lesssim&
	\frac{\nu\varepsilon}{2}
	\lVert\Delta {\boldsymbol v}^{\varepsilon}
	\rVert_{{L}^{2}(\Omega)}
	+
	\varepsilon^{-1}(\lVert
	\nabla{\boldsymbol v}^{\varepsilon}
	\rVert^{3}_{{L}^{2}(\Omega)}
	+\lVert
	{\boldsymbol v}^{\varepsilon}
	\rVert_{{L}^{2}(\Omega)})\notag\\
	&
	+\lVert
	\partial_{t} {\boldsymbol v}^{\varepsilon}
	\rVert_{{L}^{2}(\Omega)}
	+\lVert
	\nabla{\boldsymbol v}^{\varepsilon}
	\rVert_{{L}^{2}(\Omega)}+\lVert\boldsymbol \rho^{\varepsilon}
	\rVert_{{L}^{2}(\Omega)}\,.
\end{align}
By simplifying equation \eqref{3.3.4} and taking the $L^\infty$ norm with respect to $t\in[0,T)~(T<\min(T^*,\pi/2))$ on both sides of the inequality, we can obtain
\begin{align*}\label{3.3.5}
	\lVert\Delta{\boldsymbol v}^{\varepsilon}\rVert_{L^\infty([0,T),{L}^2(\Omega))}
	\lesssim&
	\varepsilon^{-2}(\lVert
	\nabla{\boldsymbol v}^{\varepsilon}
	\rVert^{3}_{L^\infty([0,T),{L}^2(\Omega))}
	+\lVert
	{\boldsymbol v}^{\varepsilon}
	\rVert_{L^\infty([0,T),{L}^2(\Omega))})\notag\\
	&
	+\varepsilon^{-1}(\lVert
	\partial_{t} {\boldsymbol v}^{\varepsilon}
	\rVert_{L^\infty([0,T),{L}^2(\Omega))}+
	\lVert
	\nabla{\boldsymbol v}^{\varepsilon}
	\rVert_{L^\infty([0,T),{L}^2(\Omega))})+\varepsilon
	\lesssim1
	\,.
\end{align*}
Given that 
\begin{equation}
	\lVert{\boldsymbol v}^{\varepsilon}\rVert^2_{L^\infty([0,T)\times\Omega)}
	\leq
	\lVert{\boldsymbol v}^{\varepsilon}\rVert^{\frac12}_{L^\infty([0,T),{L}^2(\Omega))}
	\lVert\Delta{\boldsymbol v}^{\varepsilon}\rVert^{\frac32}_{L^\infty([0,T),{L}^2(\Omega))}
	\lesssim\varepsilon\,,
\end{equation}
it follows that \( {\boldsymbol v}^{\varepsilon} \) exhibits convergence within the \( L^\infty \) framework.

\backmatter

%
%
%

\bmhead{Acknowledgements}
This work was partially supported by the National Key R\&D Program of China (No. 2020YFA072500), the NSFC (No. 12161084), the Natural Science Foundation of Xinjiang, P.R. China
(No. 2022D01E42), and the Innovation Project of Excellent Doctoral Students of Xinjiang University (No. XJU2024BS038).

\section*{Declarations}
\begin{itemize}
\item Competing interests: On behalf of all authors, the corresponding author states that there are no Conflict of interest.
\item Data availability: No data, models or code were generated or used during the study.
\end{itemize}


%
%
%
%

\begin{appendices}
%
%




\end{appendices}



\begin{thebibliography}{00}
	\bibitem{Bresch2004}
	Bresch, D., G{\'e}rard-Varet, D.: Roughness-induced effects on the quasi-geostrophic model.
	Comm. Math. Phys. \textbf{253}, 81--119 (2005)
	
	
	
	\bibitem{Chemin2006}
	Chemin, J.~Y., Desjardins, B., Gallagher, I., Grenier, E.: Mathematical geophysics. An introduction to rotating fluids and the Navier-Stokes equations.
	Oxford Lecture Ser. Math. Appl., 32. The Clarendon Press, Oxford University Press, Oxford (2006) 
	
	
	
	
	
	\bibitem{Constantin2022}
	 Constantin, A.,  Germain, P.: Stratospheric planetary flows from the perspective of the Euler equation on a rotating sphere. { Arch. Ration. Mech. Anal.} {\bf 245},  587--644 (2022)
	

	
	
	
	
	
	\bibitem{Gerard2017}
  Dalibard, A. L.,  G{\'e}rard-Varet, D.: Nonlinear boundary layers for rotating fluids. {Anal. PDE}  {\bf 10},  1--42 (2017)
	
 

	
	\bibitem{Galdi1994}
	Galdi, G.~P.: An introduction to the mathematical theory of the {Navier}-{Stokes} equations, {Vol}. {I}: {Linearized} steady problems.
	Springer Tracts Nat. Philos. 38, Springer-Verlag, New York (1994) 
	
	
	\bibitem{Gerard2003}
	G{\'e}rard-Varet, D.: Highly rotating fluids in rough domains.
	J. Math. Pures Appl. (9) \textbf{82}, 1453--1498 (2003)
	
    \bibitem{Gerard2005}
    G{\'e}rard-Varet, D.: Formal derivation of boundary  layers in fluid mechanics.
    J. Math. Fluid Mech. \textbf{7}, 179--200 (2005) 
	
	\bibitem{giga2007}
	Giga, Y., Inui, K., Mahalov, A., Matsui, S., Saal, J.:
	Rotating Navier-Stokes equations in $\Bbb R_+^3$ with initial data nondecreasing at infinity: the Ekman boundary layer problem. {Arch. Ration. Mech. Anal.} {\bf 186},   177--224 (2007)
	

	
		
	\bibitem{giga2013}
	Giga, Y., Saal, J.: An approach to rotating boundary layers based on vector Radon measures. {J. Math. Fluid Mech.} {\bf 15},   89--127  (2013)
	
	
	
	\bibitem{Gong2015}
	 Gong, S. B., Guo, Y.,  Wang,  Y. G.: Ekmann boundary layer expansions of Navier-Stokes equations with rotation. {Bull. Inst. Math. Acad. Sin. (N.S.)} {\bf 10}, 375--392  (2015) 
	
	

	
	
	\bibitem{Greenspan1968}
	Greenspan, H.~P.: The theory of rotating fluids.
	Cambridge Monogr. Mech. Appl. Math.
	Cambridge University Press, Cambridge (1968) 
	
	

	
	\bibitem{Grenier1997}
	Grenier, E., Masmoudi, N.: Ekman layers of rotating fluids, the case of well prepared initial data.
	Comm. Partial Differential Equations \textbf{22}, 953--975 (1997)
	

	
%
	\bibitem{Jia2024}
	Jia, Y. F., Du, Y., Guo, L. H.: Geometric constraints on Ekman boundary layer solutions in non-flat regions with well-prepared data.
	arXiv preprint arXiv: 2408.07582, (2024)
	


	
	\bibitem{Kato1988}
	Kato, T., Ponce, G.:  Commutator estimates and the Euler and Navier-Stokes equations.
	{Comm. Pure Appl. Math.} {\bf 41},  891--907 (1988)
	

	
	\bibitem{Li2019}	
	Li, W. X., Ngo, V. S., Xu, C. J.: Boundary layer analysis for the fast horizontal rotating fluids. {Commun. Math. Sci.} {\bf 17},  299--338  (2019)
	


	
	\bibitem{Liu2014}
	 Liu, C. J., Wang,  Y. G.: Derivation of Prandtl boundary layer equations for the incompressible Navier-Stokes equations in a curved domain. { Appl. Math. Lett.} {\bf 34},  81--85  (2014)
	

	
	
	\bibitem{Liu2019}
	Liu, C. J.,  Xie, F.,  Yang, T.: Justification of Prandtl ansatz for MHD boundary layer. {SIAM J. Math. Anal.} {\bf 51},   2748--2791  (2019)
	
	
	
	\bibitem{Masmoudi2000}
	Masmoudi, N.: Ekman layers of rotating fluids: {the} case of general initial data.
	Comm. Pure Appl. Math. \textbf{53}, 432--483 (2000)
	

	
	
		\bibitem{Pedlosky}
	Pedlosky, J.: Geophysical Fluid Dynamics.
	{Springer}-{Verlag}, New {York} (1979) 
	
	
	
	\bibitem{Rax2019}
	Rax, J.:  Mathematical study of the equatorial Ekman boundary layer. {Z. Angew. Math. Phys.} {\bf 70},  165  (2019)
	
	\bibitem{Rousset2007}
Rousset, F.: Asymptotic behavior of geophysical fluids in highly rotating balls.
Z. Angew. Math. Phys. \textbf{58}, 53--67 (2007) 


\bibitem{Stewartson}
Stewartson, K.: On almost rigid rotations. Part 2.
J. Fluid Mech. \textbf{26}, 131--144 (1966)   
	
	\bibitem{Ukai1986}
Ukai, S.: The incompressible limit and the initial layer of the compressible Euler equation.
J. Math. Kyoto Univ. \textbf{26}, 323--331 (1986)  
 
\end{thebibliography}

\end{document}